\documentclass{amsart}
\usepackage{graphicx}
\graphicspath{{images/}}

\usepackage{a4wide}
\usepackage[hidelinks]{hyperref}
\usepackage{xcolor}
\usepackage{verbatim}
\usepackage{caption}
\usepackage{mathtools}
\usepackage{enumerate}

\usepackage{xcolor}

\theoremstyle{plain}
\newtheorem{thm}{Theorem}

\newtheorem{prop}[thm]{Proposition}
\newtheorem{lemma}[thm]{Lemma}
\newtheorem{cor}[thm]{Corollary}
\theoremstyle{definition}
\newtheorem{definition}[thm]{Definition}

\newtheorem{remark}[thm]{Remark}

\usepackage{amssymb,amsmath,amsthm,mathrsfs}
\usepackage{graphicx}        
\usepackage{upgreek} 

\newcommand{\tn}[1]{\ensuremath{\mathbb{T}^{#1}}}

\newcommand{\rn}[1]{\ensuremath{\mathbb{R}^{#1}}}
\newcommand{\zn}[1]{\ensuremath{\mathbb{Z}^{#1}}}
\newcommand{\sn}[1]{\ensuremath{\mathbb{S}^{#1}}}
\newcommand{\nn}[1]{\ensuremath{\mathbb{N}^{#1}}}
\newcommand{\etot}[1]{\ensuremath{E_{\mathrm{tot},#1}}}
\newcommand{\metot}[1]{\ensuremath{\mathcal{E}_{\mathrm{tot},#1}}}
\newcommand{\hetot}[1]{\ensuremath{\hat{E}_{\mathrm{tot},#1}}}
\newcommand{\R}{\mathbb{R}}

\newcommand{\bD}{\bar{D}}

\newcommand{\g}{\gamma}
\renewcommand{\a}{\alpha}

\newcommand{\bge}{\bar{g}}

\renewcommand{\d}{\partial}

\newcommand{\ric}{\mathrm{Ric}}

\newcommand{\md}{{\mathrm{d}}}

\newcommand{\mfI}{\mathfrak{I}}

\newcommand{\tr}{\mathrm{tr}}

\newcommand{\rodiv}{\mathrm{div}}

\newcommand{\roRic}{\mathrm{Ric}}
\newcommand{\roScal}{\mathrm{Scal}}
\newcommand{\roRiem}{\mathrm{Riem}}

\newcommand{\rograd}{\mathrm{grad}}

\newcommand{\refer}{\mathrm{ref}}

\newcommand{\ldr}[1]{\langle #1\rangle}

\newcommand{\mR}{\mathcal{R}}

\newcommand{\mK}{\mathcal{K}}

\newcommand{\mI}{\mathcal{I}}

\newcommand{\mH}{\mathcal{H}}

\newcommand{\ve}{\varepsilon}

\renewcommand{\o}{\omega}

\newcommand{\bfI}{\mathbf{I}}

\newcommand{\bfJ}{\mathbf{J}}
\newcommand{\bfK}{\mathbf{K}}

\newcommand{\bfc}{\mathbf{c}}

\newcommand{\K}{\mathcal{K}}
\newcommand{\rK}{\mathring{\K}}
\newcommand{\ch}{\mathcal{H}}
\newcommand{\rch}{\mathring{\ch}}
\newcommand{\rphi}{\mathring{\Phi}}
\newcommand{\rpsi}{\mathring{\Psi}}
\newcommand{\diver}{\mathrm{div}}
\newcommand{\rp}{\mathring{p}}

\newcommand{\sric}{\ric_h}
\newcommand{\lie}{\mathcal{L}}
\newcommand{\n}{\nabla}
\newcommand{\ovn}{\overline{\n}}
\newcommand{\I}{\mathrm{\mathbf{I}}}
\newcommand{\tsum}{\textstyle\sum}
\newcommand{\eig}{\widehat e}
\newcommand{\deig}{\widehat \omega}
\newcommand{\p}{\partial}
\newcommand{\mfx}{\mathfrak{X}}
\newcommand{\cb}{\mathcal{B}}
\newcommand{\cR}{\mathcal{R}}
\newcommand{\J}{\mathrm{\mathbf{J}}}
\newcommand{\rce}{\mathring{e}}
\newcommand{\ro}{\mathring{\omega}}

\newcommand{\tint}{\textstyle\int}

\newcommand{\ear}{\vec{e}}

\newcommand{\abs}[1]{|#1|}

\newcommand{\snorm}[1]{\|#1\|}

\renewcommand{\S}{\Sigma}
\newcommand{\s}{\sigma}
\newcommand{\sigmap}{{\sigma_p}}

\renewcommand{\k}{k}

\newcommand{\bk}{\bar \k}

\newcommand{\ce}{{\check e}}
\newcommand{\co}{{\check \o}}

\newcommand{\mrmfK}{\mathring{\mathfrak{K}}}

\newcommand{\mrPsi}{\mathring{\Psi}}
\newcommand{\mrPhi}{\mathring{\Phi}}
\newcommand{\mrp}{\mathring{p}}

\newcommand{\dtp}{e_0(\varphi)}

\DeclareMathOperator{\Id}{Id}

\begin{document}

\title{Complete asymptotics in the formation of quiescent big bang singularities}
\author{Andrés Franco-Grisales \and Hans Ringström}
\date{}
\address{Department of Mathematics, KTH Royal Institute of Technology, 100 44 Stockholm, Sweden}
\email[Andrés Franco-Grisales]{anfg@kth.se}
\email[Hans Ringström]{hansr@kth.se}
\begin{abstract}
  There are three categories of mathematical results concerning quiescent big bang singularities: the derivation of asymptotics
  in a symmetry class; the construction of spacetimes
  given initial data on the singularity; and the proof of big bang formation in the absence of symmetries, including the proof
  of stable big bang formation. In a recent article, the first author demonstrated the
  existence of developments corresponding to a geometric notion of initial data on a big bang singularity. Moreover, this
  article, combined with previous articles by the second author, gives a
  unified and geometric perspective on large classes of seemingly disparate results in the first two categories. Concerning the third category, Oude
  Groeniger et al recently formulated a general condition on initial data ensuring big bang formation,
  including curvature blow up. This result, among other things, generalises previous results on stable big bang formation.
  However, it does not include a statement saying that the 
  solutions induce initial data on the singularity. Here we tie all three categories of results together by demonstrating
  that the solutions of Oude Groeniger et al induce data on the singularity. However, the results are more general and
  can potentially be used to derive similar conclusions in other gauges. 
\end{abstract}

\maketitle

\section{Introduction}

The subject of this article is quiescent big bang singularities. To put this topic into context, it is natural to return to the work of
Belinski\v{\i}, Khalatnikov and Lifschitz (BKL), see, e.g., \cite{bkl70,bkl82} as well as more modern refinements, e.g., \cite{dhn,dn,hur,hul}.
Very crudely, the idea is to study the asymptotics in the direction of a big bang singularity using a foliation by spacelike hypersurfaces, the
mean curvatures of which diverge; i.e., if $K$ denotes the Weingarten map of the leaves of the foliation, $\theta:=\tr K$ diverges. In order
to extract a finite limit, it is therefore natural to expansion normalise by dividing by the
mean curvature: $\mK:=K/\theta$. We refer to $\mK$ as \textit{the expansion normalised Weingarten map}. One suggestion of BKL is that, along a
causal curve going into the singularity, the eigenvalues of $\mK$ should oscillate according to a chaotic one-dimensional map referred to as
the \textit{Kasner map} or \textit{BKL map}. This is referred to as the oscillatory setting. However, for certain matter models and for
certain symmetry classes, the oscillations can be neutralised and the resulting behaviour is \textit{quiescent} in the sense that the eigenvalues
converge. Oscillatory solutions are very interesting, but, to the best of our knowledge, there are only results in the spatially homogeneous
setting; see, e.g., \cite{weaver,cbu,BianchiIXattr, beguin,lhwg,lrt,brehm,bad}. There is a vast literature on spatially homogeneous solutions, and
most of the results concern the quiescent setting. There are also some results in the case of $\tn{3}$-Gowdy symmetry; see, e.g.,
\cite{cim,RinAsVel,RinSCC}. However, going beyond these results has turned out to be difficult. One way to avoid this difficulty is to
turn the problem around, to start with ``data at the singularity'' and to prove that there are corresponding solutions. Over the last 20--30
years, a substantial body of work on this topic has appeared; see, e.g., \cite{kar,iak,ren,aarendall,iam,damouretal,sta,ABIF,klinger,aeta,fal,aaf}.
Unfortunately, these results have certain deficiencies. One deficiency is that many of the results are in the real analytic setting, a regularity
class incompatible with finite speed of propagation, a central notion in general relativity. A second deficiency is that the notions of initial
data in the different results are strongly connected to the particular choice of gauge. It is therefore not easy to see how the conditions of the
different articles relate. Finally, viewing the constructed spacetimes as solutions to the initial value problem given initial data on the
singularity, it would be highly desirable to have a geometric formulation of the initial value problem; isometric data on the singularity should give
rise to isometric developments etc. In \cite{RinQC}, an article based on \cite{RinWave,RinGeoJDG}, the second author developed such a geometric
notion of initial data. In \cite{andres}, the first author demonstrated that this notion of initial data can be used to formulate a geometric
initial value problem for the Einstein--nonlinear scalar field equations with data on the singularity. In particular, given initial data there is,
up to isometry, a unique corresponding development, and isometric initial data on the singularity give rise to isometric developments. We expect
this result to give a unified perspective on the results starting with initial data on the singularity. In the spatially homogeneous Bianchi class
A setting, this is verified in \cite{RinSHID,RinModel}. In the $\tn{3}$--Gowdy setting, it is verified in \cite{RinSHID}. In particular, this
means that the notion of initial data used in \cite{kar,ren} is a special case of the initial data used in \cite{RinQC,andres}. That the notion
of initial data used in \cite{aarendall} is a special case of the notion introduced in \cite{RinQC} is verified in \cite{RinQC}. That the notion
of initial data used in \cite{klinger,fal,aaf} is a special case of the notion introduced in \cite{RinQC} is also verified in \cite{RinQC}.

Studying the stability of big bang formation is an alternative approach which has been very active in the last 5--10 years. The results
began with the work of Rodnianski and Speck, demonstrating stability of spatially homogeneous, isotropic and spatially flat solutions in the
Einstein-scalar field and the Einstein-stiff fluid setting; see \cite{rasql,rasq}. Speck demonstrates a similar result in the case of
$\sn{3}$-spatial topology; see \cite{specks3}. In \cite{rsh}, the authors consider the vacuum equations in higher dimensions in situations
with intermediate anisotropy; the conditions are more general than requiring proximity to isotropy, but they are still quite far from covering
the full subcritical range in which stability is expected. In \cite{fau}, Fajman and Urban prove a stability result in the case of spatial
local homogeneity and isotropy and negative spatial curvature. This result was later generalised to allow the presence of Vlasov
matter and any sign of the sectional curvature in \cite{fauVl}. Finally, Fournodavlos, Rodnianski and Speck demonstrate stable big bang formation
for the full subcritical range of Kasner-scalar field spacetimes in \cite{FRS}. This is a very important result for many reasons. First, because
it covers the maximal range in which one can obtain stability. Second, because the authors introduce a new formulation of the equations using a
Fermi-Walker propagated frame which has turned out to be very fruitful. The extension of this result to stability of the big
bang singularity for the Einstein-Maxwell-scalar field-Vlasov system in the full strong subcritical range is discussed in \cite{anetal}.

All of the stability results mentioned so far
are obtained using constant mean curvature (CMC) foliations. Such foliations are natural for several reasons. To begin with, they are uniquely
determined by the spacetime for large classes of matter models. Moreover, they synchronise the singularity, something which is not to be
expected for most gauges. On the other hand, CMC foliations have deficiencies. In particular, the corresponding gauge fixed system
includes elliptic equations. This means that it is not possible to localise the analysis in space. This is contrary to the expected spatial
localisation of the causal structure. Moreover, thinking ahead, going beyond stability of spatially homogeneous solutions, quite violent
spatial variations are to be expected; this
can be seen already in the spikes appearing in the $\tn{3}$-Gowdy symmetric setting, see, e.g., \cite{raw}. In order to obtain a clear picture
of the behaviour of solutions near the big bang, it can therefore be expected to be necessary to localise the analysis. In the spatially
homogeneous, isotropic and spatially flat setting, localisation is achieved by Beyer and Oliynyk in \cite{bat}, a result which is generalised
to include a fluid in \cite{baotwo}. This requires the development of a new, local, gauge. The authors achieve this by using the scalar field
as a time coordinate. In \cite{boz}, Beyer, Oliynyk and Zheng generalise the result to cover the full subcritical range in $3$-dimensions,
a result which is generalised to higher dimensions in \cite{zheng}. As in \cite{FRS}, the authors use a Fermi-Walker transported frame in order
to obtain their results.

In spite of the importance of the above results, they have one fundamental weakness, namely that they only demonstrate stability of specific
spatially locally homogeneous solutions. Due to the ideas of BKL and the results \cite{aarendall,andres} concerning data on the singularity, it
is natural to expect that solutions should exhibit quiescent asymptotics more generally. Moreover, due to the perspective developed in
\cite{RinWave,RinGeoJDG,RinQC}, it is natural to expect that expansion normalised initial data should play an important role in phrasing the corresponding
condition. In fact, in \cite{OPR}, Oude Groeniger, Petersen and the second author provide a general condition on initial data, without any
reference to a background solution (spatially homogeneous or otherwise) such that the corresponding solutions exhibit big bang formation with a
crushing CMC foliation and curvature blow up. Moreover, the result gives partial asymptotics, corresponding to some of the data on the singularity,
as formulated in \cite{RinQC} or \cite{andres}. In addition, the results of \cite{OPR} not only imply the previous stability results in the
Einstein--scalar field setting with CMC foliations, they substantially extend them. Since \cite{OPR} concerns the Einstein--nonlinear scalar field
setting, the results also yield future and past global nonlinear stability results for large classes of spatially locally homogeneous solutions.

Considering the above circle of ideas, the only thing missing to tie the results given data on the singularity together with the stability results
is a proof that the solutions constructed in \cite{OPR} induce data on the singularity in the sense of \cite{RinQC} or \cite{andres}. In this article,
we provide such a proof. Moreover, we do so by providing tools we expect to be applicable more generally. More specifically, we prove that on the
basis of appropriate $C^2$-assumptions concerning solutions to the FRS equations used in \cite{OPR}, it is possible to deduce $C^k$-conclusions for
any $k$. We also prove that, given an appropriate foliation (not necessarily a CMC foliation) and appropriate assumptions concerning a Fermi-Walker
propagated orthonormal frame of the leaves of the foliation as well as concerning basic geometric quantities expressed with respect this Fermi-Walker
propagated orthonormal frame, it follows that the solution induces data on the singularity in the sense of \cite{RinQC} or \cite{andres}.

\subsection{The Einstein--nonlinear scalar field equations}\label{ssection:enlsf eqs}
In this article, we are interested in the Einstein--nonlinear scalar field equations. They are determined once we specify a
\textit{potential}; i.e., a function $V \in C^\infty(\R)$. A solution consists of a time oriented Lorentz manifold $(M,g)$, i.e. a \textit{spacetime},
and a function $\varphi\in C^{\infty}(M)$, referred to as the \textit{scalar field}, satisfying
\begin{subequations}\label{seq:nlsf sys}
\begin{align}
    \ric_g - \tfrac{1}{2} \roScal_g g &= T,\label{eq:EE nlsf}\\
    \Box_g \varphi &= V' \circ \varphi. \label{eq:sf we}
\end{align}
\end{subequations}
Here $\ric_g$ and $\roScal_g$ denote the Ricci and scalar curvatures of $g$ respectively, $\Box_g$ is the wave operator associated with $g$, and
\begin{equation}\label{eq:Tdef}
  T = \md\varphi \otimes \md\varphi - \big( \tfrac{1}{2}|\md\varphi|_g^2 + V \circ \varphi \big)g
\end{equation}
is the energy-momentum tensor of $\varphi$. Note that adding a cosmological constant term to the left hand side of (\ref{eq:EE nlsf}) would not
make the system more general, as the cosmological constant can be absorbed in the potential.

There is a Cauchy problem associated with (\ref{seq:nlsf sys}). 
\begin{definition}
  \textit{Initial data for the Einstein--nonlinear scalar field equations} (\ref{seq:nlsf sys}) consist of a Riemannian manifold $(\S,h)$, a symmetric
  covariant $2$-tensor field $k$ and two smooth functions $\phi$ and $\psi$ satisfying
  \begin{subequations}\label{seq:constraints}
    \begin{align}
      \roScal_h - \abs k^2_{h} + (\tr_h k)^2
      &= \psi^2 + |\md \phi|_{h}^{2} + 2 V \circ \phi, \label{eq: Hamiltonian constraint} \\
      \mathrm{div}_hk - \md \tr_h k
      &= \psi \md \phi. \label{eq: momentum constraint}
    \end{align}
  \end{subequations}  
  If $\theta:=\tr_hk$ is constant, the initial data are said to be \textit{constant mean curvature} (CMC) initial data.
  Given initial data for the Einstein--nonlinear scalar field equations, say $(\S,h,k,\phi,\psi)$ a corresponding \textit{development} is a
  solution $(M,g,\varphi)$ to (\ref{seq:nlsf sys}), and an embedding $i:\Sigma\rightarrow M$ such that the following hold. First, $i^*g=h$, so that,
  in particular, $i(\S)$ is a spacelike hypersurface in $(M,g)$. Second, if $\bk$ is the second fundamental form induced on $i(\S)$ and $U$ is the
  future directed unit normal to $i(\S)$, then
  \[
  i^*\bk=k,\ \ \ \varphi\circ i=\phi,\ \ \
  (U\varphi)\circ i=\psi.
  \]
\end{definition}
\begin{remark}
  Equation \eqref{eq: Hamiltonian constraint} is called the \emph{Hamiltonian constraint equation} and \eqref{eq: momentum constraint} is called the
  \emph{momentum constraint equation}.
\end{remark}
\begin{remark}
  Due to arguments going back to the work of Yvonne Choquet-Bruhat and Robert Geroch \cite{cbag}, there is a unique maximal globally hyperbolic
  development associated with initial data; see also \cite[Chapter~23]{stab} for a textbook presentation. 
\end{remark}
\begin{remark}\label{remark:induced id}
  Let $(M,g,\varphi)$ be a solution to (\ref{seq:nlsf sys}) and let $\S$ be a spacelike hypersurface in $(M,g)$ with future directed unit normal $U$.
  Define $h$ and $k$ to be the associated induced first and second fundamental forms. Let, moreover, $\phi := \varphi|_\S$ and
  $\psi := U(\varphi)|_\S$. Then $(\S,h,k,\phi,\psi)$ are initial data for the Einstein--nonlinear scalar field equations; see, e.g.,
  \cite[Chapter~13]{RinCauchy}.
\end{remark}

Given initial data such that the associated mean curvature is strictly positive, it is possible to define associated expansion normalised initial
data, a concept which is of central importance in the definition of initial data on the singularity.

\begin{definition}\label{def:exp norm id}
  Let $(\S,h,k,\phi,\psi)$ be initial data for the Einstein--nonlinear scalar field equations and $\theta:=\tr_hk$. Assume that $\theta>0$.
  Let $K$ be the $(1,1)$-tensor field obtained by raising one index in $k$ using $h$. Then $\mK:=\theta^{-1}K$ is called the
  \textit{expansion normalised Weingarten map}. Define
  \[
  \mH(X,Y):=h(\theta^\mK X,\theta^\mK Y),\ \ \
  \Psi:=\theta^{-1}\psi,\ \ \
  \Phi:=\phi+\Psi\ln \theta,
  \]
  for all $X,Y\in T_p\S$ and all $p\in\S$, where
  \[
  \theta^\K X := \tsum_{m=0}^\infty \tfrac{(\ln\theta)^m}{m!} \K^m X.
  \]
  Then $(\S,\mH,\mK,\Phi,\Psi)$ are referred to as \textit{expansion normalised initial data} associated with the initial data $(\S,h,k,\phi,\psi)$. 
\end{definition}

\subsection{Initial data on the singularity}
Next, we define the notion of initial data on the singularity. 
\begin{definition}\label{def:idos}
    For integers $\ell \geq 1$ and $n \geq 3$, let $(\Sigma,\rch)$ be an $n$-dimensional $C^\ell$ Riemannian manifold, $\rK$ a $C^\ell$ $(1,1)$-tensor field on $\Sigma$ and $\rphi, \rpsi \in C^\ell(\Sigma)$. Assume that 
    \begin{enumerate}
        \item $\tr\rK = 1$ and $\rK$ is self adjoint with respect to $\rch$.
        \item $\tr\rK^2 + \rpsi^2 = 1$ and $\diver_{\rch} \rK = \rpsi \md\rphi$.
        \item The eigenvalues $\rp_1,\ldots,\rp_n$ of $\rK$ are everywhere distinct, and satisfy 
        \begin{equation}\label{eq:cond on eigenvalues}
          \rp_I + \rp_J - \rp_K < 1
        \end{equation}
        for all $I,J$ and $K$ such that $I \neq J$.
    \end{enumerate}
    Then $(\Sigma,\rch,\rK,\rphi,\rpsi)$ are $C^\ell$ \emph{robust nondegenerate quiescent initial data on the singularity for the
      Einstein--nonlinear scalar field equations}. If they are $C^\ell$ for any $\ell\in\nn{}$, they are said to be smooth. 
\end{definition}
\begin{remark}
  It is possible to omit the nondegeneracy requirement; i.e., the requirement that the eigenvalues be distinct. However, we are mainly
  interested in the nondegenerate setting in this article. It is also possible to omit the requirement that (\ref{eq:cond on eigenvalues})
  hold. However, it is then necessary to introduce separate conditions on the structure coefficients associated with an eigenframe of $\rK$;
  see, e.g., \cite[Definition~10, p.~9]{RinQC}. Moreover, the solutions inducing initial data violating (\ref{eq:cond on eigenvalues}) are
  expected to be unstable. Since we are here interested in the asymptotics induced by solutions arising in stability results, such a
  generalisation is therefore not of great interest here. 
\end{remark}
In order to relate a solution as in Subsection~\ref{ssection:enlsf eqs} to initial data on the singularity in the sense of
Definition~\ref{def:idos}, we first need to introduce a foliation. Let $(M,g,\varphi)$ be a globally hyperbolic solution to the Einstein--nonlinear
scalar field equations. Then $M$ is diffeomorphic to $\S\times \mI$ for some $n$-manifold $\S$ and some open interval $\mI=(t_-,t_+)$, where
$\S_t:=\S\times\{t\}$ are spacelike Cauchy hypersurfaces. When speaking of big bang singularities, we have
crushing singularities in mind. In other words, we here assume the mean curvature of the leaves $\S_t$ of the foliation to diverge uniformly to
$\infty$ as $t\downarrow t_-$. The metric $g$ and scalar field $\varphi$ induce initial data, cf. Remark~\ref{remark:induced id}, on each
leaf: $(\S,h(t),k(t),\phi(t),\psi(t))$; we think of
the data as being defined on $\S$. Moreover, these data induce expansion normalised initital data $(\S,\mH(t),\mK(t),\Phi(t),\Psi(t))$ as explained
in Definition~\ref{def:exp norm id}. We say that the solution induces initial data on the singularity along the foliation $\{\Sigma_t\}_{t\in\mI}$
if there are $t_0\in\mI$; constants $\delta>0$ and $C_\ell$ for $\ell\in \nn{}_0$; and robust quiescent initial data on the singularity for the
Einstein--nonlinear scalar field equations, say $(\Sigma,\rch,\rK,\rphi,\rpsi)$, such that
\begin{subequations}\label{seq:conv to dos}
  \begin{align}
    \|\theta^{\delta}\bD^\ell(\mH(t)-\rch)\|_{C^0(\S)}+\|\theta^{\delta}\bD^\ell(\mK(t)-\rK)\|_{C^0(\S)} &\leq C_\ell,\\
    \|\theta^{\delta}\bD^\ell(\Phi(t)-\rphi)\|_{C^0(\S)}
    +\|\theta^{\delta}\bD^\ell(\Psi(t)-\rpsi)\|_{C^0(\S)} &\leq C_\ell
  \end{align}
\end{subequations}
for all $t\leq t_0$ and all $\ell\in\nn{}_0$. Here $\bD$ is the Levi-Civita connection associated with a fixed reference Riemannian metric on $\S$,
say $h_{\refer}$, and the $C^0$-norm of a tensor field $T$ is defined as follows:
\[
\|T\|_{C^0(\S)}:=\sup_{x\in\S}|T(x)|_{h_{\refer}}.
\]
In case $\S$ is non-compact, one could relax (\ref{seq:conv to dos}) to the requirement that for each compact subset $K$ of $\Sigma$, there are
constants $\delta$ and $C_\ell$ such that (\ref{seq:conv to dos}) holds with $\S$ replaced by $K$. 

In $3+1$-dimensions, the step from data on the singularity to corresponding developments is discussed in \cite{andres}. In particular, due
to \cite[Theorem~1.21, p.~10]{andres}, initial data on the singularity give rise to unique (up to isometry) maximal globally hyperbolic
developments with a local Gaussian foliation close to the singularity. Next, we turn to the opposite direction. In particular, we prove that
the general condition on initial data introduced in \cite{OPR}, ensuring that the corresponding maximal globally hyperbolic developments
have crushing CMC singularities with curvature blow up, in fact yields solutions inducing data on the singularity along the CMC foliation. 

\subsection{The main result}

In this subsection, we state the main result, namely that the solutions constructed in \cite{OPR} induce data on the
singularity. However, to begin with, we need to introduce some terminology. 

\begin{definition} \label{def: alpha admissible}
  Let $\mfI=(\S,h,k,\phi,\psi)$ be CMC initial data with mean curvature $\theta>0$. Let $\mK$ be the associated expansion normalised Weingarten
  map, see Definition~\ref{def:exp norm id}, and let $\sigmap \in (0, 1)$. Then $\mfI$ and the eigenvalues $p_1, \hdots, p_n$ of $\mK$ are
  said to be $\sigmap$-\textit{admissible}
  if they satisfy
  \begin{equation}
    p_I + p_J - p_K
    < 1 - \sigmap, \label{cond: alternating q}
  \end{equation}
  for all $I, J, K$, such that $I \neq J$.
\end{definition}
Concerning the potential, we make the following assumption.
\begin{definition}\label{def:Vadm}
  Fix $\s_{V} \in (0, 1)$. If $V\in C^{\infty}(\rn{})$ has the property that for each $k\in \nn{}_0$,
  there is a constant $c_{k}>0$ such that
  \begin{equation}
    \textstyle{\sum}_{\ell \leq k}|V^{(\ell)}(x)|
        \leq c_{k} e^{2(1-\s_{V})|x|} \label{eq: V assumption}
  \end{equation}
  for all $x\in\rn{}$, then $V$ is said to be a \textit{$\s_V$-admissible potential}. 
\end{definition}
\begin{remark}
  In comparison with \cite[Definition~1, p.~5]{OPR}, we have here omitted the requirement that $V$ be non-negative. We include this requirement
  explicitly in Theorem~\ref{thm: big bang formation} instead.
\end{remark}

\begin{thm}[The main theorem] \label{thm: big bang formation} 
  Fix $3\leq n\in\nn{}$ and admissibility thresholds $\s_V$, $\sigmap \in (0,1)$. Let $(\S, h_\refer)$ be a closed Riemannian manifold of dimension
  $n$ and let $V\in C^{\infty}(\rn{})$ be a non-negative $\s_V$--admissible potential. Then there is a $\kappa\in\nn{}$, depending only on $n$, $\s_V$
  and $\sigmap$, such that the following holds. For any $\zeta_0 > 0$, there is a $\zeta_1 > 0$ such that:

  If $\mfI$ are $\sigmap$-admissible CMC initial data on $\S$ for the Einstein--nonlinear scalar field equations with potential $V$, such that the
  associated expansion-normalized initial data $(\S, \mH_0,\mK_0,\Phi_0,\Psi_0)$ satisfy
  \begin{equation} \label{eq: Sobolev bound assumption}
    \snorm{\mH^{-1}_0}_{C^{0}_\refer(\S)}+\snorm{\mH_0}_{H^{\kappa}_\refer(\S)} + \snorm{\mK_0}_{H^{\kappa}_\refer(\S)}
    + \snorm{\Phi_0}_{H^{\kappa}_\refer(\S)} + \snorm{\Psi_0}_{H^{\kappa}_\refer(\S)}
    < \zeta_0;
  \end{equation}
  $|p_I-p_J|>\zeta_{0}^{-1}$ for $I\neq J$, where $p_I$ are the eigenvalues of $\mK_0$; and the mean curvature satisfies $\theta_0 > \zeta_1$,
  then the maximal globally hyperbolic development of $\mfI$, say $(M, g, \varphi)$, with associated embedding $\iota: \S \hookrightarrow M$, has a
  past crushing big bang singularity in the following sense:

  \noindent \textbf{CMC foliation:}
  There is a diffeomorphism $\Xi$ from $(0, t_0] \times \S$ to $J^-\left(\iota(\S)\right)$, i.e.~the causal past of the Cauchy hypersurface
  $\iota(\S)$, such that $\Xi (\S_{t_0}) = \iota(\S)$ and the hypersurfaces $\Xi(\S_t) \subset M$ are spacelike Cauchy hypersurfaces with
  constant mean curvature $\theta(t) = \frac1t$, for each $t \in (0, t_0]$.

  \noindent \textbf{Asymptotic data:}
  There are smooth robust nondegenerate quiescent initial data on the singularity for the Einstein--nonlinear scalar field equations,
  say $(\Sigma,\rch,\rK,\rphi,\rpsi)$, and constants $\delta>0$ and $C_\ell$ for $\ell\in\nn{}_0$, such that
  \begin{equation}\label{eq:Cl asymptotics as data}
    \|\mH(t)-\rch\|_{C^\ell_\refer(\S)}+\|\mK(t)-\rK\|_{C^\ell_\refer(\S)}+\|\Phi(t)-\rphi\|_{C^\ell_\refer(\S)}
    +\|\Psi(t)-\rpsi\|_{C^\ell_\refer(\S)}\leq C_\ell t^{\delta}
  \end{equation}
  for all $t\leq t_0$ and all $\ell\in\nn{}_0$, where $[\Sigma,\mH(t),\mK(t),\Phi(t),\Psi(t)]$ are the expansion normalised initial data
  induced on $\S_t$. 
    
  \noindent \textbf{Curvature blow-up:}
  Let $\rp_I$ denote the eigenvalues of $\rK$. Then there are constants $C_\ell>0$, $\ell\in\nn{}_0$, such that
  $\mathfrak{R}_g := \roRic_{g, \mu \nu} \roRic_{g}^{\mu \nu}$ and $\mathfrak{K}_g := \roRiem_{g,\mu \nu \xi \rho} \roRiem_g^{\mu \nu \xi \rho}$
  satisfy
  \begin{subequations}\label{seq:Kretschmann Ricci asymptotics intro}
    \begin{align}
      \big\| t^4 \mathfrak{K}_g(t, \cdot) - 
      4 \big[ \textstyle{\sum}_{I} \rp_I^2 (1 - \rp_I^2)
        + \textstyle{\sum}_{I < J} \rp_I^2 \rp_J^2 \big]
      \big\|_{C^{\ell}_\refer(\S)}\leq C_\ell t^{\delta}, \\
      \| t^4 \mathfrak{R}_g(t, \cdot) - \rpsi^4
      \|_{C^{\ell}_\refer(\S)}\leq C_\ell t^{\delta}
    \end{align}
  \end{subequations}
  for all $t \in (0, t_0]$ and all $\ell\in\nn{}_0$, so that $(M, g)$ is $C^2$ past inextendible. Moreover, every past directed causal geodesic
  in $M$ is incomplete and $\mathfrak{K}_g$ blows up along every past inextendible causal curve.
\end{thm}
\begin{remark}
  The constant $\kappa$ is given as follows. In the main theorem of \cite{OPR}, i.e. \cite[Theorem~12, pp.~6--7]{OPR}, the constants $k_0$
  and $k_1$ are introduced. Choosing $k_0$ by demanding that equality hold in \cite[(12a), p.~6]{OPR}; and choosing $k_1$ to be the smallest
  integer such that \cite[(12b), p.~6]{OPR} holds for this choice of $k_0$, the integer $\kappa$ is given by $\kappa=k_1+2$. 
\end{remark}
\begin{remark}
  The Sobolev spaces $H^\ell_\refer$ and $C^\ell$ spaces $C^\ell_\refer$ in (\ref{eq: Sobolev bound assumption}), (\ref{eq:Cl asymptotics as data}) and
  (\ref{seq:Kretschmann Ricci asymptotics intro}) are defined using the reference metric $h_{\refer}$ and the associated Levi--Civita
  connection. In most of the article we use the $C^\ell$ and Sobolev spaces introduced in Subsection~\ref{ssection:norms and basic est}.
\end{remark}
\begin{remark}
  The requirement that $V$ be non--negative can be relaxed; cf. \cite[Remark~16, p.~7]{OPR}.
\end{remark}
\begin{remark}
  In \cite{OPR}, the stability of \textit{quiescent model solutions}, a notion introduced in \cite[Definition~30, p.~11]{OPR}, is also discussed; cf.
  \cite[Theorem~33, p.~12]{OPR} and \cite[Subsection~1.6, pp.~13--19]{OPR}. The background solutions to the Einstein-scalar field equations considered in,
  e.g., \cite{specks3,FRS,fau,fauVl} are examples of quiescent model solutions, but the notion is more general than this. There are quiescent model
  solutions inducing degenerate data on the singularity; i.e., data such that the $\rp_I$ are not all distinct. For this reason, the solutions obtained
  by appealing to the stability results are not necessarily such that Theorem~\ref{thm: big bang formation} applies. However, even in the degenerate
  setting, one does obtain conclusions. Consider a solution obtained by appealing to \cite[Theorem~33, p.~12]{OPR}. Then the first part of the proof of
  Theorem~\ref{thm: big bang formation}, which simply consists of an application of Theorem~\ref{thm:hod}, still applies. Next, assume that at some
  point, say $x\in \S$, the eigenvalues $\rp_I(x)$ are all distinct. Then there is a closed ball, say $\bar{B}_R(x)$, of radius $R>0$ with respect to
  the reference metric $h_{\refer}$ such that, for every $y\in\bar{B}_R(x)$, the eigenvalues $\rp_I(y)$ are all distinct. Choosing $T>0$, $\delta>0$ and
  $r>0$ small enough, it can then be ensured that $g$ and $\varphi$, restricted to $M_{\mathrm{loc}}:=\bar{B}_R(x)\times (0,T)$, constitute a solution to
  the Einstein–nonlinear scalar field equations satisfying the $(\delta,n,k_0,r)$-assumptions for any $k_0$. This means that we can apply
  Theorem~\ref{thm:asymptotics} in order to obtain smooth initial data on the singularity in the ball $\bar{B}_R(x)$. Moreover,
  (\ref{eq:Cl asymptotics as data}) holds with $\S$ replaced by $\bar{B}_R(x)$. 
\end{remark}
\begin{proof}
  The proof is to be found in Section~\ref{section: proof of main theorem}.
\end{proof}
The statement of the result may seem a bit strange at first sight. However, from the point of view of initial data on the singularity, it is
quite natural. In fact, consider a solution that induces initial data on the singularity along a CMC foliation. Then it is
clear that for any $k\in\nn{}$, the expansion normalised initial data on the leaves of the CMC foliation are uniformly bounded in $H^k_\refer$
all the way to the singularity; i.e., as the mean curvature diverges to infinity. It is therefore natural to first fix a bound
on the size of the expansion normalised initial data, without imposing any restriction on how large the bound may be, and then to require a
lower bound on the mean curvature. 

\subsection{The FRS equations}\label{ssection:eqs}
One topic of interest in this article is solutions to the FRS equations, introduced in \cite{FRS}. The purpose of the present
section is to recall the slight reformulation of these equations introduced in \cite[Theorem~12, pp.~5--6]{OPR}. These equations arise from the Einstein--nonlinear scalar field equations by assuming a CMC foliation, a vanishing shift vector field, and using a Fermi-Walker transported frame. The underlying
assumptions are the following. Let $(\S,h_{\refer})$ be a closed parallelisable
Riemannian manifold of dimension $n$; let $\{E_i\}_{i=1}^n$ be a smooth global orthonormal frame with respect to $h_{\refer}$
with dual frame $\{\eta^i\}_{i=1}^n$; and let $V\in C^\infty(\rn{})$. The manifold on which we wish to solve the equations is of the form
$M:=\S\times\mI$ for some open interval $\mI\subset (0,\infty)$ and the basic variables are $N$, $e_I^i$, $\omega^I_i$, $\gamma_{IJK}$,
$k_{IJ}$ and $\varphi$. Moreover, we assume $N>0$ on $M$. As clarified in the statement of \cite[Proposition~70, pp.~23--24]{OPR},
the geometric interpretation of these variables is the following. Let $e_0:=N^{-1}\d_t$, where $\d_t$ signifies differentiation
with respect to the second variable in $M\times\mI$; $e_I:=e_I^iE_i$; and $\omega^I:=\omega^I_i\eta^i$. Then $\{e_\a\}_{\a=0}^n$
is a time oriented orthonormal frame for a Lorentz metric $g$ on $M$ and if $h$ denotes the induced metric on
the constant-$t$ hypersurfaces $\S_t:=\S\times\{t\}$, then $\{e_I\}_{I=1}^n$ is an orthonormal frame with respect to $h$ with dual frame
$\{\omega^I\}_{I=1}^n$. Moreover,
\[
  k:=k_{IJ}\omega^I\otimes\omega^J, \ \ \
  \gamma_{IJK}=\omega^K([e_I,e_J])
\]
are the induced second fundamental form on $\S_t$ and the structure coefficients of the frame $\{e_I\}_{I=1}^n$, respectively,
and the mean curvature satisfies $k_{II}=t^{-1}$ (here and below, we sum over repeated capital Latin indices, even if both are
subscripts). Given the above underlying assumptions, the equations, see \cite[Proposition~70, pp.~23--24]{OPR}, are the following, using
the conventions introduced in Subsection~\ref{ssection:symm and antisymm}: 

\textit{The evolution equations for the frame and co-frame:}
\begin{align}
  e_0(e_I^i) 
  &= - k_{IJ}e_J^i, \label{eq: transport frame} \\
  e_0(\o^I_i) 
  &= k_{IJ} \o^J_i. \label{eq: transport co-frame}
\end{align}
\textit{The evolution equations for $\gamma$ and $k$:}
\begin{align}
  \begin{split}
    e_{0}\g_{IJK}
    &= - 2 N^{-1} e_{[I}(Nk_{J]K}) 
    - k_{IL} \g_{LJK} - k_{JL} \g_{ILK} + k_{KL} \g_{IJL},
  \end{split} \label{eq:concoef} \\
  \begin{split}
    e_{0}k_{IJ} 
    &= N^{-1}e_{(I}e_{J)} (N) - N^{-1} e_{K} (N\g_{K(IJ)})
    - e_{(I}(\g_{J)KK}) - t^{-1}k_{IJ} \\
    &\quad +\g_{KLL}\g_{K(IJ)} + \g_{I(KL)}\g_{J(KL)} - \tfrac14 \g_{KLI} \g_{KLJ}\\
    &\quad + e_I(\varphi)e_J(\varphi)+\tfrac2{n-1}(V\circ\varphi) \delta_{IJ}.
  \end{split} \label{eq:sff}
\end{align}
\textit{The evolution equations for the derivatives of the scalar field:}
\begin{align}
  e_0 \left( e_I \varphi \right)
  &= N^{-1} e_I \left( N e_0 \varphi \right) - k_{IJ} e_J(\varphi), \label{eq:spatial scalar field derivative} \\
  e_{0}\left( e_{0}(\varphi) \right)
  &= e_{I}\left( e_{I}(\varphi) \right) -t^{-1}e_{0}(\varphi)+N^{-1}e_{I}(N)e_{I}(\varphi) -\g_{JII}e_{J}(\varphi)-V'\circ\varphi. \label{eq:scalarfield}
\end{align}
\textit{The Hamiltonian constraint equation:}
\begin{align}
  \begin{split}
    2e_I(\g_{IJJ})
    &- \tfrac14\g_{IJK}(\g_{IJK} + 2 \g_{IKJ}) -\g_{IJJ}\g_{IKK} - k_{IJ}k_{IJ} + t^{-2} \\
    &= (e_{0}\varphi)^{2} + e_{I}(\varphi)e_{I}(\varphi) + 2V \circ \varphi.
  \end{split} \label{eq:hamconstraint}
\end{align}
\textit{The momentum constraint equation:}
\begin{equation} \label{eq:MC}
  e_I k_{IJ}
  =\g_{LII}k_{LJ} + \g_{IJL}k_{IL} + e_{0}(\varphi) e_{J}(\varphi).
\end{equation}
\textit{The lapse equation:}
\begin{equation}
  \begin{split}
    e_{I}e_{I}(N) &= t^{-2}(N-1)+\g_{JII}e_{J}(N) - \big( e_{I}(\varphi)e_{I}(\varphi) + \tfrac{2n}{n-1} V \circ \varphi \big)N \\
    &\quad + \left( 2e_I(\g_{IJJ}) - \tfrac14\g_{IJK}(\g_{IJK} + 2 \g_{IKJ}) - \g_{IJJ}\g_{IKK} \right) N.
  \end{split}\label{eq:LapseEquation}
\end{equation}
It is convenient to recall that combining (\ref{eq:hamconstraint}) and (\ref{eq:LapseEquation}) yields
\cite[(85), p.~32]{OPR}:
\begin{equation} \label{eq: alternative lapse}
  \begin{split}
    e_{I}e_{I}(N) &= t^{-2}(N-1)+\g_{JII}e_{J}(N) \\
    &\quad - \big( t^{-2} - k_{IJ} k_{IJ} - e_0(\varphi) e_0(\varphi) + \tfrac{2}{n-1}V\circ\varphi \big)  N.
  \end{split}
\end{equation}
In the above equations and in what follows, it is useful to recall \cite[Lemma~72, p.~25]{OPR}:
\begin{lemma}[Lemma~70, p.~25, \cite{OPR}] \label{lemma: spatial curvatures}
  Let $(\S,h)$ be a Riemannian manifold with a smooth global orthonormal frame $\{e_I\}_{I=1}^{n}$. Then the Ricci and scalar curvature are given by
  \begin{align}
    \begin{split}
    \roRic_h(e_I, e_J) 
    &= e_K (\g_{K(IJ)}) + e_{(I} (\g_{J)KK}) - \g_{I(KL)} \g_{J(KL)}\\
    &\quad + \tfrac14 \g_{KLI} \g_{KLJ} - \g_{KLL} \g_{K(IJ)},
    \end{split}\label{eq: SpatialRicciCurvature} \\
    \roScal_h 
    &= 2 e_I (\g_{IJJ}) - \tfrac14 \g_{IJK} (\g_{IJK} + 2 \g_{IKJ}) - \g_{IJJ} \g_{IKK},\label{eq:SpatialScalarCurvature}
  \end{align}
  where $\gamma_{IJK} = h([e_I, e_J], e_K)$.
\end{lemma}
When deriving energy estimates, it is, in the end, useful to consider an energy associated with
\begin{equation}\label{eq:aIdef}
  a_I:=\gamma_{IJJ}.
\end{equation}
Interestingly, this quantity satisfies an evolution equation which is in some respects better than that of $\g$. In fact, due to
(\ref{eq:concoef}) and (\ref{eq:MC}),
\begin{equation*}
  \begin{split}
    e_0a_I = & - 2 N^{-1} e_{[I}(Nk_{J]J})- k_{IL} \g_{LJJ} - k_{JL} \g_{ILJ} + k_{JL} \g_{IJL}\\
    = & -N^{-1}t^{-1}e_IN+N^{-1}e_J(N)k_{IJ}+e_J(k_{JI})- k_{IL} \g_{LJJ} \\
    = & -N^{-1}t^{-1}e_IN+N^{-1}e_J(N)k_{IJ}- k_{IL} a_L +a_Lk_{LI} + \g_{JIL}k_{JL} + e_{0}(\varphi) e_{I}(\varphi).
  \end{split}
\end{equation*}
In other words,
\begin{equation}\label{eq:dt aI}
  \d_ta_I=-t^{-1}e_I(N)+e_J(N)k_{IJ}+ N\g_{JIL}k_{JL} + Ne_{0}(\varphi) e_{I}(\varphi).
\end{equation}

\subsection{Higher order derivative estimates for solutions to the FRS equations}

Even though our main application in the present paper is to derive complete asymptotics for the solutions constructed in \cite{OPR},
it is of interest to derive energy estimates using minimal assumptions. In this way, the result might be useful in other
contexts. The main assumptions are the following.

\begin{definition}\label{def:bas sup ass}
  Consider a solution as described in Subsection~\ref{ssection:eqs}. In particular, the mean curvature of $\S_t:=\S\times\{t\}$ is $\theta=t^{-1}$.
  Assume that the existence interval includes $(0,t_0]$, where $t_0\in (0,1]$. Next, assume that there is an $\ve>0$ and a constant $C$ such that 
  \begin{subequations}\label{eq:estdynvarmfltot gen}
    \begin{align}
      t^{1-3\ve}\|e\|_{C^{1}(\S_t)}+ t^{1-3\ve}\|\omega\|_{C^{1}(\S_t)} +t^{1-2\ve}\|\g\|_{C^{1}(\S_t)}
      +t^{1-2\ve}\|\ear\varphi\|_{C^1(\S_t)}  &\leq  C,\label{eq:e om g phi be gen}\\      
      t \|k\|_{C^{1}(\S_t)}+t \|e_0 \varphi\|_{C^{1}(\S_t)} &\leq  C,\label{eq:k phi be gen}\\
      \|N-1\|_{C^2(\S_t)} &\leq  Ct^{\ve}\label{eq:N be gen}
    \end{align}
  \end{subequations}
  for $t\leq t_0$. In addition to this, assume, for $t\leq t_0$,
  \begin{equation}\label{eq:almost asymptotic Hcon gen}
    t^2 \|t^{-2} - k_{IJ} k_{IJ} - \dtp \dtp \|_{C^{1}(\S_t)} \leq Ct^{2\ve}.
  \end{equation}
  Then the solution is said to satisfy \textit{the basic supremum assumptions}. 
\end{definition}
\begin{remark}
  In the definition we use the notation introduced in Subsection~\ref{ssection:norms and basic est}.
\end{remark}
Given these assumptions, the following conclusions hold.
\begin{thm}\label{thm:hod}
  Consider a solution to (\ref{eq: transport frame})--(\ref{eq:LapseEquation}) as described in Subsection~\ref{ssection:eqs}.
  Assume, in addition,  that $V$ is a $\s_V$-admissible potential for some $\s_V\in (0,1)$ and that the basic supremum assumptions, see
  Definition~\ref{def:bas sup ass}, are satisfied with $3\ve\leq \s_V$. Then there is, for every $\ell\in\nn{}_0$, a constant $C_\ell$ such that 
  \begin{subequations}\label{seq:hod}
    \begin{align}
      t^{1-2\ve}\|e\|_{C^{\ell}(\S_t)}+ t^{1-2\ve}\|\omega\|_{C^{\ell}(\S_t)} +t^{1-\ve}\|\g\|_{C^{\ell}(\S_t)}+t^{1-\ve}\|\ear\varphi\|_{C^\ell(\S_t)}
      &\leq C_\ell,\label{eq:eogearphi Cl est}\\
      t\|k\|_{C^{\ell}(\S_t)}+t\|e_0 \varphi\|_{C^{\ell}(\S_t)} &\leq C_\ell,\label{eq:k ezphi opt Cl bd}\\
      t^{1-3\ve}\|\ear N\|_{C^\ell(\S_t)}+t^{-\ve/2}\|N-1\|_{C^\ell(\S_t)} &\leq C_\ell\label{eq:lapse est interm}
    \end{align}
  \end{subequations}
  hold for all $\ell\in\nn{}_0$ and all $t\leq t_0$. Moreover, there are $\mrPhi,\mrPsi\in C^{\infty}(\S)$ and, for every $I,J\in\{1,\dots,n\}$,
  $\mrmfK_{IJ}\in C^{\infty}(\S)$ such that
  \begin{subequations}\label{seq:k phi prel est}
    \begin{align}
      \|t k_{IJ}(\cdot,t)-\mrmfK_{IJ}\|_{C^\ell(\S)}&\leq C_\ell t^{\ve/4},\label{eq:tkIJ conv}\\
      \|t (e_0\varphi)(\cdot,t)-\mrPsi\|_{C^\ell(\S)}+\|\varphi(\cdot,t)-\mrPsi\ln t-\mrPhi\|_{C^\ell(\S)} &\leq C_\ell t^{\ve/4}\label{eq:ezPhi and Phi conv final}
    \end{align}
  \end{subequations}
  hold for all $\ell\in\nn{}_0$, all $I,J\in\{1,\dots,n\}$ and all $t\leq t_0$.
\end{thm}
\begin{remark}
  The eigenvalues of the expansion normalised Weingarten map are not assumed to be distinct in this theorem. 
\end{remark}
\begin{proof}
  The proof is to be found at the end of Subsection~\ref{ssection:summing up}.
\end{proof}

\subsection{Deriving asymptotics, given higher order derivative bounds}

Our next goal is to derive asymptotics, given higher order derivative bounds. We do so for crushing foliations that are not necessarily CMC,
since the results are applicable in a more general setting we now describe. To begin with, we consider a spacetime $(M,g)$ with a smooth
foliation by spacelike hypersurfaces, lapse $N>0$ and vanishing shift vector field, so that $M=\S\times\mI$ and
\begin{equation}\label{eq:standard form g ito h}
  g = -N^2\md t \otimes \md t + h,
\end{equation}
where $h$ denotes the family of induced (Riemannian) metrics on the $\Sigma_t:=\S\times\{t\}$; see Subsection~\ref{ssection:conv geo} for our
conventions concerning the geometry. In this setting, and with the conventions introduced in Subsection~\ref{ssection:conv geo}, the Einstein
equations read as follows. 

\begin{prop} \label{einsteins equations in a foliation}
  In the setting described above and with the conventions introduced in Subsection~\ref{ssection:conv geo}, let $\varphi \in C^\infty(M)$.
  Then $(M,g,\varphi)$ solves the Einstein--nonlinear scalar field equations with potential $V\in C^\infty(\R)$ if and only if the following
  system of equations holds: the evolution equations for $K$ and $\varphi$,
  \begin{subequations}
    \begin{align}
      \lie_{e_0} K + \sric^{\sharp} + \theta K - N^{-1}(\ovn^2N)^\sharp
      = & \md\varphi \otimes  \rograd_h \varphi + \tfrac{2}{n-1}(V \circ \varphi)\Id,\label{evolution equation for the second fundamental form}\\
      -e_0e_0 \varphi + \Delta_h \varphi - \theta e_0 \varphi + \ldr{\md\varphi,\md(\ln N)}_h = & V' \circ \varphi,\label{eq:we varphi split}
    \end{align}
  \end{subequations}
  where $\Id$ denotes the family of identity $(1,1)$-tensors on the $\Sigma_t$; and the constraint equations:
  \begin{subequations}\label{seq:constraints split}
    \begin{align}
      \roScal_h + \theta^2 - \tr K^2 &= e_0(\varphi)^2 + |\md\varphi|_h^2 + 2V \circ \varphi,\label{hamiltonian constraint}\\
      \diver_h K - \md\theta &= e_0(\varphi)\md\varphi.\label{momentum constraint}
    \end{align}
  \end{subequations}  
\end{prop}
\begin{remark}
  Equations~\eqref{hamiltonian constraint} and \eqref{momentum constraint} are of course the Hamiltonian and momentum constraint equations on
  the leaves of the foliation; see \eqref{seq:constraints}.
\end{remark}
\begin{proof}
  To begin with, (\ref{evolution equation for the second fundamental form}) is, e.g., an immediate consequence of \cite[(269), p.~65]{RinGeo},
  keeping in mind that
  \[
    \rho=T(e_0,e_0),\ \ \
    \bar{p}=n^{-1}h^{ij}T_{ij},
  \]
  $\mathcal{P}$ is defined by \cite[(17), p.~8]{RinGeo} and the fact that $T$ is given by (\ref{eq:Tdef}). Next, (\ref{eq:we varphi split})
  follows from a straightforward computation and (\ref{seq:constraints split}) follow from \cite[Chapter~13]{RinCauchy}.

  Turning to the other direction, (\ref{eq:sf we}) is an immediate consequence of (\ref{eq:we varphi split}). Due to
  (\ref{seq:constraints split}) the normal--normal and the normal--tangential components of Einstein's equations hold. Next, combining
  \cite[(257), p.~63]{RinGeo} and \cite[(269), p.~65]{RinGeo} with (\ref{evolution equation for the second fundamental form}), it is clear
  that the spatial--spatial components of the spacetime Ricci tensor equal what they would if Einstein's equations would hold. Combining this
  information with the fact that the normal--normal and the normal--tangential components of Einstein's equations hold, it can be deduced that
  (\ref{eq:EE nlsf}) holds. 
\end{proof}

\begin{definition}\label{def:asympt assump}
  Let $3\leq n\in\nn{}$, $k_0\in\nn{}$, $r>0$, $\delta\in (0,1)$ and fix a $\delta$-admissible potential $V\in C^\infty(\R)$. Let $(M,g,\varphi)$ be a
  solution to the
  Einstein--nonlinear scalar field equations with potential $V$ such that $M = \S\times (0,T)$, where $0<T\in\R$; $\Sigma$ is a closed $n$-dimensional
  manifold; $\Sigma$ has a global frame $\{E_i\}_{i=1}^n$ with dual frame $\{\eta^i\}_{i=1}^n$; and the metric $g$ takes the form
  (\ref{eq:standard form g ito h}), where $N>0$ and $h$ denotes a family of Riemannian metrics on $\Sigma_t :=\S\times \{t\}$. Let $\theta$ denote
  the mean curvature of $\S_t$. Let $\{e_I\}_{I=1}^n$ be
  a family of Fermi-Walker transported orthonormal frames on the $\S_t$ with respect to $h$, so that
  \begin{equation} \label{fermi walker equation}
    \n_{e_0} e_I = e_I(\ln N)e_0,
  \end{equation}
  where $e_0$ is the future directed unit normal to the $\S_t$ and $\nabla$ denotes the Levi-Civita connection of $(M,g)$. Assume that $\theta>0$ and
  let $\ch$, $\K$, $\Phi$ and $\Psi$ be the expansion normalised initial data induced on $\Sigma_t$. Let $p_I$ be the eigenvalues of $\K$ and assume
  them to be distinct everywhere. Assume that there are everywhere distinct functions $\rp_1, \ldots, \rp_n \in C^{k_0}(\Sigma)$, and functions
  $\rphi, \rpsi \in C^{k_0}(\Sigma)$ satisfying
  \begin{equation}\label{eq:Kasner relations}
    \textstyle{\sum}_I \rp_I = \textstyle{\sum}_I \rp_I^2 + \rpsi^2 = 1,
  \end{equation}
  and
  \begin{equation} \label{subcritical condition}
    \rp_I + \rp_J - \rp_K < 1 - 2\delta, \qquad I \neq J.
  \end{equation}
  Assume that there is a constant $\bfc_0$ such that
  \begin{equation}\label{eq:pI lim Phi lim Psi lim}
    \|p_I(t) - \rp_I\|_{C^{k_0}(\Sigma)} + \|\Phi(t) - \rphi\|_{C^{k_0}(\Sigma)} + \|\Psi(t) - \rpsi\|_{C^{k_0}(\Sigma)} \leq \bfc_0 t^\delta
  \end{equation}
  for all $t<T$. Let $k_{IJ}:=k(e_I,e_J)$; $\{\omega^I\}_{I=1}^n$ denote the frame dual to $\{e_I\}_{I=1}^n$; and $\gamma_{IJK}:=\omega^K([e_I,e_J])$ denote the
  structure coefficients of the frame $\{e_I\}_{I=1}^n$. Let $e_I = e_I^iE_i$ and $\omega^I = \omega_i^I\eta^i$. Define $\Theta := N\theta$. Assume
  that the following estimates hold:
  \begin{subequations} \label{bounds on the solution}
    \begin{align}
      t\|k\|_{C^{k_0}(\Sigma)} + t\|e_0\varphi\|_{C^{k_0}(\Sigma)} &\leq \bfc_0,\\
      t^{1-3\delta}\|e\|_{C^{k_0}(\Sigma)} + t^{1-3\delta}\|\omega\|_{C^{k_0}(\Sigma)} + t^{-\delta}\|t\Theta - 1\|_{C^{k_0}(\Sigma)}
      + t^{1-3\delta}\|\gamma\|_{C^{k_0}(\Sigma)} &\leq \bfc_0,\\
      \|N\|_{C^{k_0}(\Sigma)} + \|N^{-1}\|_{C^{0}(\Sigma)} + t^{1-\delta}\|\ear(\ln N)\|_{C^{k_0}(\Sigma)} + t^{1-\delta}\|\ear\varphi\|_{C^{k_0}(\Sigma)} &\leq \bfc_0,\\
      \|t\Theta - 1\|_{C^{0}(\Sigma)} &\leq \tfrac{1}{2}\label{eq:tTheta and N invertible}
    \end{align}
  \end{subequations}
  for $t<T$, where we use the notation introduced in Subsection~\ref{ssection:norms and basic est}. Finally, assume that $|p_I-p_J| \geq r$ for $I \neq J$. Then $(M,g,\varphi)$ is said to be a
  \textit{solution to the Einstein--nonlinear scalar field equations satisfying the} $(\delta,n,k_0,r)$-\textit{assumptions}.
\end{definition}

\begin{thm}\label{thm:asymptotics}
  Let $3\leq n\in\nn{}$, $\ell\in\nn{}$, $r>0$, $\delta\in (0,1)$ and fix a $\delta$-admissible potential $V\in C^\infty(\R)$. Then there is a $\kappa_0$,
  depending only on $\delta$ and $\ell$, such that if $k_0\geq \kappa_0$ and $(M,g,\varphi)$ is a solution to the Einstein--nonlinear scalar field
  equations satisfying the $(\delta,n,k_0,r)$-assumptions, then there is a $t_0>0$ and a $C^\ell$ Riemannian metric $\rch$ and $(1,1)$-tensor field
  $\rK$ on $\Sigma$, such that
  \[
    \|\ch(t) - \rch\|_{C^\ell(\Sigma)} + \|\K(t) - \rK\|_{C^\ell(\Sigma)} \leq Ct^\delta
  \]
  for some constant $C$ and all $t\leq t_0$. Moreover, $\rp_I$ are the eigenvalues of $\rK$, and $(\Sigma,\rch,\rK,\rphi,\rpsi)$ are $C^\ell$ robust
  nondegenerate quiescent initial data on the singularity for the Einstein--nonlinear scalar field equations.
\end{thm}
\begin{remark}
  The same statement holds if we assume $\S$ to be a smooth compact manifold with boundary. In particular, in case $\S$ is non--compact, localising
  the assumptions to compact balls in $\S$, we obtain correspondingly localised conclusions. 
\end{remark}
\begin{remark}
  In the case of the FRS equations, combining Theorem~\ref{thm:hod} with Theorem~\ref{thm:asymptotics} yields stronger conclusions (smooth initial
  data on the singularity) under weaker assumptions. The reader is encouraged to formulate the corresponding result. 
\end{remark}
\begin{proof}
  The proof is to be found at the end of Section~\ref{section:asymptotics}.
\end{proof}

\subsection{Outline of the arguments}

The main result of this article is Theorem~\ref{thm: big bang formation}. However, Theorem~\ref{thm: big bang formation} is an immediate consequence
of Theorems~\ref{thm:hod} and \ref{thm:asymptotics}. Moreover, the latter theorems have an independent interest in their own right.
In particular, it is to be expected that Theorem~\ref{thm:asymptotics} can be used to obtain data on the singularity in other
settings. For these reasons, we here focus on a description of the arguments required to prove Theorems~\ref{thm:hod} and \ref{thm:asymptotics}.

\textit{Higher order derivative estimates.} In \cite{OPR}, the authors derive partial information concerning the asympotics of solutions in the $C^k$-norm
for fixed finite $k$. However, as $k$ tends to infinity, the number of derivatives you need to control initially tends to infinity. 
This is unfortunate. In the present article, we therefore begin by deriving partial information concerning the asymptotics for all derivatives, given
assumptions concerning the initial data in a fixed finite degree of regularity. This is the purpose of Theorem~\ref{thm:hod}. The basic idea is to prove energy
estimates inductively on the order $\ell$ of the energy $\metot{\ell}$ introduced in (\ref{eq:metotldef}). For the first order energy, we have estimates by
assumption. For a higher order, say $\ell$, the goal is to prove an estimate of the form (\ref{eq:dtmetotlest}). In this estimate, we effectively control the
last term on the right hand side due to the induction hypothesis and we control the second term on the right hand side because the factor in front of the
energy is integrable. For these reasons, we can effectively ignore the last two terms on the right hand side of (\ref{eq:dtmetotlest}), with a slight loss.
This means, inductively, that the energy only grows at the rate $t^{-9}$ in the direction of the singularity; see
Corollary~\ref{cor:energy estimates}. Moreover, the rate is independent of the number of derivatives and of the dimension. Combining this estimate
with interpolation yields good enough estimates for most of the quantities of interest. However, the estimates for $k_{IJ}$ and $e_0\varphi$ are not
immediately optimal. Nevertheless, combining our improved estimates with the equations, optimal estimates for these quantities can also be derived.
Moreover, we obtain the conclusions of Theorem~\ref{thm:hod}.

To conclude, the main point is to derive an estimate of the form (\ref{eq:dtmetotlest}). We begin by deriving some basic consequences of the
assumptions of Theorem~\ref{thm:hod} in Subsections~\ref{ssection:basic cons hod} and
\ref{ssection:est pot hod}. Moreover, in  Subsection~\ref{ssection:en hod}, we introduce notation for our basic energies.
It is of particular importance to note that we include additional time dependent weights (that tend to infinity in the direction of the singularity)
in the energies associated with the frame and co-frame; see (\ref{eq:high ord en e and omega}). In Subsection~\ref{ssection:algorithm} we derive an
algorithm for estimating expressions that occur in the energy estimates. The main point here is to develop a calculus that allows us to divide
the expressions that appear in the energy estimates into terms that can be estimated by the last two terms on the right hand side of
(\ref{eq:dtmetotlest}) and terms that contribute to the first term on the right hand side. In the development of the algorithm, it is crucial that
we have the extra time dependent weights in the energies for the frame and co-frame. With the algorithm at hand, we derive elliptic
estimates for the lapse function in Subsection~\ref{ssection:lapse hod}. The identity (\ref{eq:lapse ito kEk etc}) may look odd at
first sight. However, it turns out to be important since the second term on the left hand side appears in the energy estimates. We also need
(\ref{eq:ear N Hl est}) and (\ref{eq:ENl final est lemma}). It is important to note here that the coefficients in front of the squared Sobolev
norms on the right hand sides of (\ref{eq:ear N Hl est}) and (\ref{eq:ENl final est lemma}) are numerical. Deriving energy estimates for the
frame and co-frame is straightforward; see Subsection~\ref{ssection:frame est hod}. Estimating the second fundamental form and the structure
coefficients is more complicated. However, integrating by parts and appealing to the algorithm yields (\ref{eq:dt E gamma k}). In this estimate,
it is important to note that the terms appearing in the integrands have the form $aE_{\bfI}(f)E_{\bfI}(g)$. Moreover, the $f$ and $g$ appearing are
basic variables in the FRS equations. In addition, $a$ is either a numerical multiple of $t^{-1}$ or a numerical multiple of $k_{IJ}$ or $e_0\varphi$.
The reason estimates of this form are useful when deriving (\ref{eq:dtmetotlest}) is that, due to (\ref{eq:almost asymptotic Hcon gen}),
\[
  t^2k_{IJ}k_{IJ}+t^2(e_0\varphi)^2\leq 1+Ct^{2\varepsilon}.
\]
There is only one problem with (\ref{eq:dt E gamma k}): expressions of the form $E_{\bfI}(\g_{IJJ})$ appear.
Such terms could potentially give rise to dimensionally dependent constants in the first term on the right hand side of
(\ref{eq:dtmetotlest}). We avoid this possibility by estimating the energy for $a_I$, introduced in (\ref{eq:Ealdef}), separately in
Subsection~\ref{ssection:the energy of aI}. Note that the energy estimate again has a similar structure to that described above.
All that remains is to estimate the energy associated with the scalar field. We do so in Subsection~\ref{ssection:the scalar field hod}.
Again, the relevant estimate has the form described above. Combining all the above observations, we, in the end, arrive at the estimate
(\ref{eq:dtmetotlest}) and thereby at the desired conclusion.

\textit{Asymptotics.} In the context of quiescent big bang formation, one would like to show that solutions induce data on the singularity as in \cite{RinQC} or \cite{andres}. In particular, if the solution is expected to exhibit stable big bang formation, one would like to show that it induces data on the singularity as in Definition~\ref{def:idos}. In \cite{OPR}, the authors obtain some of these data. Specifically, if $( \Sigma, \ch, \K, \Phi, \Psi )$ denotes the expansion normalized initial data along the foliation, the authors prove the convergence of the eigenvalues of $\K$, of $\Phi$ and of $\Psi$, in $C^{k_0}$, for some fixed $k_0 \in \mathbb{N}$, to nontrivial limits satisfying the expected constraints, \eqref{eq:Kasner relations} and \eqref{subcritical condition}. However, they fail to obtain convergence of $\ch$ and $\K$, due to limited information on the asymptotic behavior of the Fermi-Walker frame $\{e_I\}_{I=1}^n$. The authors obtain estimates of the form
\[
\|e_I\|_{C^{k_0}(\Sigma)} \leq Ct^{-1+3\sigma}
\]
for some $\sigma > 0$, along with similar estimates for the corresponding dual frame and structure coefficients, which are not sharp. That this is enough information to prove big bang formation is a remarkable property of the Fermi-Walker frame. One does not need to have full knowledge about the asymptotics to prove curvature blow-up. 

Following the results of \cite{OPR}, it makes sense to set as our starting point some $C^{k_0}$ bounds on a Fermi-Walker frame and on the components of geometric quantities in terms of this frame, along with $C^{k_0}$ convergence of the eigenvalues of $\K$, of $\Phi$ and of $\Psi$. These are precisely the estimates \eqref{bounds on the solution} and \eqref{eq:pI lim Phi lim Psi lim}. Note that these estimates correspond to the conclusions obtained in \cite{OPR} with two exceptions. First, since we do not assume the CMC condition, we need a way to control the mean curvature. This is done through the quantity $\Theta = N\theta$, which is assumed to be asymptotic to $t^{-1}$ as $t \downarrow 0$. Second, in \cite{OPR}, the authors obtain convergence of the lapse, $N \to 1$ in $C^{k_0}$ as $t \downarrow 0$. By contrast, here we only assume that $N$ is bounded in $C^{k_0}$, and that $N^{-1}$ is bounded in $C^0$. However, we point out that if the CMC condition is assumed, then the assumption that $t\Theta-1 \to 0$ becomes $N-1 \to 0$. We begin by deriving some basic consequences of the assumptions in Subsection~\ref{ssection: basic consequences of the assumptions}.

The idea to derive the remaining asymptotics, is to introduce an eigenframe $\{\eig_I\}_{I=1}^n$ for $\K$. That is, an orthonormal frame for $h$ satisfying $\K(\eig_I) = p_I\eig_I$. This is done in Subsection~\ref{ssection:the eigenframe}. Here the nondegeneracy assumption is important, since it ensures that one can take this frame to be smooth. Since the eigenframe diagonalizes both $\ch$ and $\K$ simultaneously, and we already know the eigenvalues of $\K$ to converge, if we can show that appropriate normalizations of the $\eig_I$ converge to nontrivial limits as $t \downarrow 0$, it will follow that $\ch$ and $\K$ also converge. This is the purpose of Theorem~\ref{asymptotics of the eigenframe}, of which Theorem~\ref{thm:asymptotics} is a direct consequence. To begin with, it is thus necessary to relate the Fermi-Walker frame to the eigenframe. This is done through a family of orthogonal matrices $A$ satisfying $\eig_I = A_I^Je_J$. Einstein's equations plus the assumptions can then be used to deduce that the $A_I^J$ are globally bounded in $C^{k_0-1}$. In fact, one can prove that there is a family of orthogonal matrices $\mathring{A}$ such that $A_I^J \to \mathring{A}_I^J$ in $C^{k_0-1}$ as $t \downarrow 0$; see Proposition~\ref{bounds for change of basis matrix}. Once this has been done, one can translate the estimates for the Fermi-Walker frame, the corresponding dual frame, and the structure coefficients, to the eigenframe and the associated objects. Of course, the resulting estimates are not optimal yet, but it is possible to improve on them.

Once estimates for the eigenframe have been obtained, we proceed, in Subsection~\ref{ssection: the expansion normalized eigenframe}, to consider the expansion normalized eigenframe $\{\ce_I\}_{I=1}^n$ given by $\ce_I = \theta^{-p_I} \eig_I$ (so that $\ch(\ce_I,\ce_I) = 1$), along with the corresponding dual frame and structure coefficients. These objects are now expected to converge as $t \downarrow 0$. They also inherit some estimates from the Fermi-Walker frame. However, at this point, the bound that we obtain blows up as $t \downarrow 0$. Einstein's equations then imply a system of evolution equations for the $\ce_I$ and related objects, which we use in Subsection~\ref{ssection: asymptotics} to derive asymptotics. It turns out that these equations can be integrated repeatedly as ODEs to improve on the estimates, at the price of losing two derivatives on each iteration. Thus, given an $\ell \in \nn{}$ one concludes that the $\ce_I$ and related objects are globally bounded in $C^\ell$ after a finite number of steps, provided that $k_0$ is large enough. Finally, one more integration of the evolution equations yields the existence of nontrivial $C^\ell$ limits for the $\ce_I$ and related objects. These limits can then be used to define $\rch$ and $\rK$, the limits of $\ch$ and $\K$, and it can be shown that $(\Sigma,\rch,\rK,\rphi,\rpsi)$ satisfies the conditions of Definition~\ref{def:idos}.

\textit{Proof of the main result.} Once we have Theorems~\ref{thm:hod} and \ref{thm:asymptotics} at our disposal, we can combine them with
the arguments in \cite{OPR} in order to deduce that Theorem~\ref{thm: big bang formation} holds. The details are to be found in
Section~\ref{section: proof of main theorem}.

\subsection{Acknowledgements}

This research was funded by the Swedish Research Council (Vetenskapsrådet), dnr. 2022-03053. The first author was supported by foundations
managed by The Royal Swedish Academy of Sciences.

\section{Higher order energy estimates for solutions to the FRS equations}\label{section:Higher order derivatives}

The purpose of the present section is to prove Theorem~\ref{thm:hod}. In particular, we, in the present section, consider solutions to
(\ref{eq: transport frame})--(\ref{eq:LapseEquation}) as described in Subsection~\ref{ssection:eqs}. Moreover, we, unless otherwise stated,
assume that $V$ is a $\s_V$-admissible potential and that the basic supremum assumptions, see Definition~\ref{def:bas sup ass}, are satisfied
with $3\ve\leq \s_V$. 

\subsection{Elementary consequences of the basic supremum assumptions}\label{ssection:basic cons hod}
Due to (\ref{eq:almost asymptotic Hcon gen}), we know that $t|e_0\varphi|\leq 1+Ct^{2\ve}$ for $t\leq t_0$.
Combining this estimate with (\ref{eq:N be gen}) yields
\[
  t\|\d_t\varphi\|_{C^0(\S_t)}\leq 1+Ct^{\ve}
\]
for $t\leq t_0$. This means that
\begin{equation}\label{eq:phi bd gen}
  \|\varphi\|_{C^0(\S_t)}\leq -\ln t+C
\end{equation}
for all $t\leq t_0$. Similarly, combining (\ref{eq:k phi be gen}) and (\ref{eq:N be gen}), we conclude that
\begin{equation}\label{eq:phi Co be gen}
  \|\varphi\|_{C^1(\S_t)}\leq C\ldr{\ln t}
\end{equation}
for all $t\leq t_0$. Next, since
\[
  \rodiv_{h_{\refer}}e_I=\rodiv_{h_{\refer}}(e_I^iE_i)=E_i(e_I^i)+e_I^i\rodiv_{h_{\refer}}E_i,
\]
it follows from (\ref{eq:e om g phi be gen}) that 
\begin{equation}\label{eq:diveIest gen}
  t^{1-3\ve}\|\rodiv_{h_{\refer}}e_{I}\|_{C^{0}(\S_t)}\leq C 
\end{equation}
for all $t\leq t_0$. Finally, note that combining (\ref{eq:e om g phi be gen}) and (\ref{eq:N be gen}) yields
\begin{equation}\label{eq:ear N gen}
  t^{1-4\ve}\|\ear N\|_{C^1(\S_t)}\leq C
\end{equation}
for all $t\leq t_0$.

\subsection{Energies and basic estimates}\label{ssection:en hod}
The energies we estimate in the end are, given $\ell\in\nn{}_0$, defined as follows
\begin{subequations}\label{seq:higher order energies}
  \begin{align}
    E_{(e,\omega),\ell}(t) &:= t^{-2\ve}\|e\|_{H^{\ell}(\S_t)}^2+t^{-2\ve}\|\omega\|_{H^{\ell}(\S_t)}^2,\label{eq:high ord en e and omega}\\
    E_{(\gamma,k),\ell}(t) &:= \tfrac{1}{2}\|\gamma\|_{H^{\ell}(\S_t)}^2+\|k\|_{H^{\ell}(\S_t)}^2,\\
    E_{(\varphi),\ell}(t) &:= \|e_0\varphi\|_{H^\ell(\S)}^2+\|\ear\varphi\|_{H^\ell(\S_t)}^2+t^{-2+2\ve}\|\varphi\|_{H^{\ell}(\S_t)}^2,\label{eq:high ord en phi}\\
    E_{(N),\ell}(t) &:= \|\ear N\|_{H^{\ell}(\S_t)}^2+t^{-2}\|N-1\|_{H^\ell(\S_t)}^2,\\
    \etot{\ell}(t) &:= E_{(e,\omega),\ell}(t)+E_{(\gamma,k),\ell}(t)+E_{(\varphi),\ell}(t),\\
    \hetot{\ell}(t) &:= \etot{\ell}+E_{(N),\ell}(t).
  \end{align}
\end{subequations}
Here $\ve>0$ is the constant whose existence is assumed in Definition~\ref{def:bas sup ass}. Using notation analogous to the one introduced
in (\ref{seq:e g Ck Hk norms}), we, next, define an energy for $a$ introduced in (\ref{eq:aIdef}):
\begin{equation}\label{eq:Ealdef}
  E_{(a),\ell}(t):=\|a\|_{H^\ell(\S_t)}^2.
\end{equation}
It turns out to be convenient to include this expression in the total energy, and in the final step of the energy estimates, we therefore focus
on 
\begin{equation}\label{eq:metotldef}
  \metot{\ell}:=\etot{\ell}+E_{(a),\ell}.
\end{equation}
In Subsection~\ref{ssection:bas es and identities} we record some of the basic estimates and identities we use when deriving energy estimates. 

\subsection{Estimating the potential}\label{ssection:est pot hod}

Next, we record some basic estimates concerning the potential. 

\begin{lemma}
  Consider a solution to (\ref{eq: transport frame})--(\ref{eq:LapseEquation}) as described in Subsection~\ref{ssection:eqs}. Assume, in addition,
  that $V$ is a $\s_V$-admissible potential and that the basic supremum assumptions, see Definition~\ref{def:bas sup ass}, are satisfied with
  $3\ve\leq \s_V$. Then, for any $\ell\in\nn{}_0$, there is a constant $C_\ell$, depending only on $(\S,h_{\refer})$, $c_\ell$, $\ve$ and $\ell$, such that
  the following holds for all $t\leq t_0$:
  \begin{equation}\label{eq:V and V prime Hl est}
    t^2\|V\circ\varphi\|_{H^\ell(\S_t)}+t^2\|V'\circ\varphi\|_{H^\ell(\S_t)}\leq C_\ell t^{5\ve}(1+\|\varphi\|_{H^\ell(\S_t)}).
  \end{equation}
  Moreover,
  \begin{equation}\label{eq:EiV and EiV prime Hl est}
    t^2\|E_i(V\circ\varphi)\|_{H^{\ell-1}(\S_t)}+t^2\|E_i(V'\circ\varphi)\|_{H^{\ell-1}(\S_t)}\leq C_\ell t^{5\ve}\|\varphi\|_{H^\ell(\S_t)}.
  \end{equation}
  Finally,
  \begin{equation}\label{eq:V and V prime Cone est gen}
    t^2\|V\circ\varphi\|_{C^1(\S_t)}+t^2\|V'\circ\varphi\|_{C^1(\S_t)}\leq C_\ell t^{5\ve}.
  \end{equation}
\end{lemma}
\begin{remark}
  In particular, using the terminology introduced in (\ref{eq:high ord en phi}),
  \[
    t^2\|E_i(V\circ\varphi)\|_{H^{\ell-1}(\S_t)}+t^2\|E_i(V'\circ\varphi)\|_{H^{\ell-1}(\S_t)}\leq C_\ell t^{1+4\ve}E_{(\varphi),\ell}^{1/2}.
  \]  
\end{remark}
\begin{proof}
  Combining (\ref{eq:phi bd gen}) with (\ref{eq: V assumption}) yields the conclusion that
  \begin{equation}\label{eq:Vm circ phi est}
    \textstyle{\sum}_{m \leq k}|V^{(m)}\circ\varphi|\leq c_{k}e^{2C} t^{-2+2\s_V}\leq C_kt^{-2+6\ve}
  \end{equation}
  for any $k\in \nn{}_0$, where we used the fact that $\s_V\geq 3\ve$. Next, note that
  $E_{\bfI}V^{(i)}\circ\varphi$ consists of a linear combination of terms of the form
  \[
    V^{(i+m)}\circ\varphi E_{\bfI_1}\varphi\cdots E_{\bfI_m}\varphi.
  \]
  The estimates (\ref{eq:V and V prime Hl est}) and (\ref{eq:EiV and EiV prime Hl est}) follow by Gagliardo-Nirenberg estimates, i.e.
  Lemma \ref{lemma: products in L2}, (\ref{eq:phi bd gen}) and the fact that $\ldr{\ln t}^{m-1}t^{\ve}$ is bounded for $t\leq 1$. Finally,
  combining (\ref{eq:phi Co be gen}) with (\ref{eq:Vm circ phi est}) yields (\ref{eq:V and V prime Cone est gen}).
\end{proof}

\subsection{Algorithm}\label{ssection:algorithm}

Consider a solution to (\ref{eq: transport frame})--(\ref{eq:LapseEquation}) as described in Subsection~\ref{ssection:eqs}. Assume, in addition, 
that $V$ is a $\s_V$-admissible potential and that the basic supremum assumptions, see Definition~\ref{def:bas sup ass}, are satisfied with
$3\ve\leq \s_V$. Next, we formulate an algorithm for estimating expressions that occur frequently.

\textit{General terms.} Consider a term $T$ consisting of the following factors: $l_1$ factors of the form $N$; $l_2$ factors of the form $N-1$; $l_{3,a}$
factors of the form $e_I^i$ or $\omega_i^I$; $l_{3,b}$ factors of the form $t^{-\ve}e_I^i$ or $t^{-\ve}\omega_i^I$; $l_4$ factors of the form $e_IN$;
$l_5$ factors of the form $\gamma_{IJK}$; $l_6$ factors of the form $k_{IJ}$; $l_7$ factors of the form $e_0\varphi$; $l_8$ factors of the form $e_I\varphi$; and
$l_9$ factors of the form $V\circ\varphi$ or $V'\circ\varphi$. In what follows, we refer to such a term as a \textit{general term}. Let
\begin{align*}
  m_{\mathrm{int}} &:= -l_{3,a}-l_{3,b}-l_4-l_5-l_6-l_7-l_8-2l_9,\\
  m_{\ve} &:= l_2+3l_{3,a}+2l_{3,b}+4l_4+2l_5+2l_8+5l_9,\\
  \ell_{\ve} &:= l_2+l_{3,a}+l_{3,b}+l_4+l_5+l_8+l_9.    
\end{align*}
By (\ref{eq:estdynvarmfltot gen}), (\ref{eq:ear N gen}), (\ref{eq:EiV and EiV prime Hl est}), (\ref{eq:V and V prime Cone est gen}) and Gagliardo-Nirenberg estimates, $E_{\bfI}T$,
where $0<|\bfI|\leq \ell$, can be estimated in $L^2$ by
\begin{equation*}
  \begin{split}
    & C_\ell t^{m_{\mathrm{int}}+\ve m_{\ve}}[\iota_1tE_{(N),\ell}^{1/2}+\iota_2t^{1-\ve}E_{(N),\ell}^{1/2}+\iota_{3,a}t^{1-3\ve}t^{\ve}E_{(e,\omega),\ell}^{1/2}
      +\iota_{3,b}t^{1-2\ve}E_{(e,\omega),\ell}^{1/2}+\iota_4t^{1-4\ve}E_{(N),\ell}^{1/2}]\\
    & +C_\ell t^{m_{\mathrm{int}}+\ve m_{\ve}}[\iota_5t^{1-2\ve}E_{(\g,k),\ell}^{1/2}+\iota_6tE_{(\g,k),\ell}^{1/2}+\iota_7tE_{(\varphi),\ell}^{1/2}+\iota_8t^{1-2\ve}E_{(\varphi),\ell}^{1/2}
      +\iota_9t^{2-5\ve}t^{-1+4\ve}E_{(\varphi),\ell}^{1/2}],
  \end{split}
\end{equation*}
where $\iota_i=1$ if $l_i>0$ and $\iota_i=0$ if $l_i=0$. The quantities $\iota_{3,a}$ and $\iota_{3,b}$ are defined similarly. 

\textit{The good terms.} If $T$ is a general term with $\ell_{\ve}\geq 2$ or $l_{3,a}+l_9\geq 1$, then $T$ is called a \textit{good term}. Then
\begin{equation}\label{eq:good term}
  \|E_{\bfI}T\|_{L^2(\S_t)}\leq C_\ell t^{m_{\mathrm{int}}+1+\ve}\hetot{\ell}^{1/2}.
\end{equation}
Moreover, if $l_1+l_2+l_4=0$, then $\hetot{\ell}$ can be replaced by $\etot{\ell}$ on the right hand side. 

\textit{The bad terms.} If $T$ is a general term with $\ell_{\ve}=1$ and $l_{3,a}+l_9=0$, then $T$ is called a \textit{bad term}. In this case, we can
assume $T$ to consist of $l_1$ factors of the form $N$; one factor, say $f$, which is one of $t^{-1}(N-1)$, $t^{-\ve}e_I^i$, $t^{-\ve}\omega^I_i$,
$e_IN$, $\gamma_{IJK}$ or $e_I\varphi$; $l_6$ factors of the form $k_{IJ}$; and $l_7$ factors of the form $e_0\varphi$. Again, we apply $E_{\bfI}$ to such a term.
If no derivative hits $f$, then, appealing to Lemma~\ref{lemma: products in L2}, (\ref{eq:estdynvarmfltot gen}) and (\ref{eq:ear N gen}), the
corresponding expression can be estimated by
\begin{equation}\label{eq:the bad term good est}
  C_\ell t^{-l_6-l_7+1}t^{-1+\ve}\hetot{\ell}^{1/2}.
\end{equation}
Moreover, if $l_1=0$, then $\hetot{\ell}$ can be replaced by $\etot{\ell}$ on the right hand side. 
If at least one derivative hits $f$ and there is one other factor that is hit by at least one derivative, then one can retain one derivative on
$f$ and one derivative on one of the other factors. This leaves $\ell-2$ derivatives. Again, due to Lemma~\ref{lemma: products in L2},
(\ref{eq:estdynvarmfltot gen}) and (\ref{eq:ear N gen}),
\[
C_\ell t^{-l_6-l_7+1}t^{-1}\hetot{\ell-1}^{1/2}
\]
for the corresponding expressions. Moreover, $\hetot{\ell-1}$ can be replaced by $\etot{\ell-1}$ if $l_1=0$ and $f$ is neither $t^{-1}(N-1)$ nor $e_IN$.
What remains is when $E_{\bfI}$ only hits $f$. Since $N-1$ can be bounded by $Ct^{\ve}$ in $C^0$, we can replace
all the occurrences of $N$ in the corresponding expression with $1$ at the price of an error of the form (\ref{eq:the bad term good est}). What
we are left with are $E_{\bfI} f$ times the original $l_6$ factors of the form $k_{IJ}$ and the original $l_7$ factors of the form $e_0\varphi$. In
other words,
\begin{equation}\label{eq:EbfI bad term}
  E_{\bfI}T=E_{\bfI}f\cdot k_{I_1J_1}\cdots k_{I_{l_6}J_{l_6}}(e_0\varphi)^{l_7}+R_B,
\end{equation}
where
\begin{equation}\label{eq:RB est}
  \|R_B\|_{L^2(\S_t)}\leq C_\ell t^{-l_6-l_7+1}t^{-1+\ve}\hetot{\ell}^{1/2}+C_\ell t^{-l_6-l_7+1}t^{-1}\hetot{\ell-1}^{1/2}.
\end{equation}
Again, the hats can be removed on the right hand side if $l_1=0$ and $f$ is neither $t^{-1}(N-1)$ nor $e_IN$. For future reference, it is also of
interest to note that
\[
\|E_{\bfI}T\|_{L^2(\S_t)}\leq C_\ell t^{-l_6-l_7+1}t^{-1}\hetot{\ell}^{1/2}.
\]
Again, if $l_1=0$ and $f$ is neither $t^{-1}(N-1)$ nor $e_IN$, the hat can be removed on the right hand side.

\textit{The leading order terms.} If $T$ is a general term with $\ell_{\ve}=0$, then $T$ is called a \textit{leading order term}. Then, by
arguments similar to the above, 
\begin{equation}\label{eq:leading order terms}
  \begin{split}
    & E_{\bfI}[N^{l_1}k_{I_1J_1}\cdots k_{I_{l_6}J_{l_6}}(e_0\varphi)^{l_7}]\\
    ={} & l_1E_{\bfI}(N)k_{I_1J_1}\cdots k_{I_{l_6}J_{l_6}}(e_0\varphi)^{l_7}
    +E_{\bfI}(k_{I_1J_1})\cdots k_{I_{l_6}J_{l_6}}(e_0\varphi)^{l_7}\\
    & \negmedspace+\dots+k_{I_1J_1}\cdots E_{\bfI}(k_{I_{l_6}J_{l_6}})(e_0\varphi)^{l_7}
    +l_7k_{I_1J_1}\cdots k_{I_{l_6}J_{l_6}}E_{\bfI}(e_0\varphi)\cdot (e_0\varphi)^{l_7-1}+R_B,
  \end{split}
\end{equation}
where the first term is absent if $l_1=0$; the terms involving derivatives of the $k_{IJ}$ are absent if $l_6=0$; the second last term is absent
if $l_7=0$; and $R_B$ can be estimated in $L^2$ by
\begin{equation}\label{eq:leading order terms RB}
  \|R_B\|_{L^2(\S_t)}\leq C_\ell t^{-l_6-l_7+1}t^{\ve}\hetot{\ell}^{1/2}+C_\ell t^{-l_6-l_7+1}\hetot{\ell-1}^{1/2}.
\end{equation}
Moreover, if $l_1=0$, the hats can be removed from the right hand side. For future reference, it is also of interest to keep in mind that 
\begin{equation}\label{eq:est lead order terms}
  \|E_{\bfI}T\|_{L^2(\S_t)}\leq C_\ell t^{-l_6-l_7+1}\hetot{\ell}^{1/2}.
\end{equation}
Again, if $l_1=0$, the hat can be removed from the right hand side.

\textit{Commutators.} Next, consider $[E_{\bfI},e_I]\psi$. It can be written as $e_I^i[E_{\bfI},E_i]\psi$, plus a linear combination of terms of the form
\begin{equation}\label{eq:comm terms eiI}
  E_{\bfJ}(e_I^i)E_{\bfK}\psi,
\end{equation}
where $|\bfJ|+|\bfK|=|\bfI|+1$, $|\bfJ|\geq 1$ and $|\bfK|\geq 1$. Due to (\ref{eq:e om g phi be gen}),
\[
\|e_I^i[E_{\bfI},E_i]\psi\|_{L^2(\S_t)}\leq C_\ell t^{-1+3\ve}\|\psi\|_{H^\ell(\S_t)}
\]
for some constant $C_\ell$, where $|\bfI|\leq \ell$. Next, note that (\ref{eq:comm terms eiI}) can be written
\[
E_{\bfJ_1}E_{j_1}(e_I^i)E_{\bfK_1}E_{k_1}\psi,
\]
where $|\bfJ_1|+|\bfK_1|\leq |\bfI|-1$. By Gagliardo-Nirenberg estimates, we conclude that
\[
\|E_{\bfJ}(e_I^i)E_{\bfK}\psi\|_{L^2(\S_t)}\leq C_\ell\|e\|_{C^1(\S_t)}\|\psi\|_{H^{\ell}(\S_t)}+C_\ell\|\psi\|_{C^1(\S_t)}\|e\|_{H^\ell(\S_t)}.
\]
Note that the first term on the right hand side can be estimated by appealing to (\ref{eq:e om g phi be gen}). This yields, assuming
$|\bfI|\leq \ell$ and keeping (\ref{eq:high ord en e and omega}) in mind,
\begin{equation}\label{eq:EbfI eI comm est}
  \|[E_{\bfI},e_I]\psi\|_{L^2(\S_t)}\leq C_\ell t^{-1+3\ve}\|\psi\|_{H^\ell(\S_t)}+C_\ell t^{\ve}\|\psi\|_{C^1(\S_t)}E_{(e,\omega),\ell}^{1/2}.
\end{equation}

\subsection{Estimating the lapse function}\label{ssection:lapse hod}
Concerning the lapse function, the following estimates hold.
\begin{lemma}
  Consider a solution to (\ref{eq: transport frame})--(\ref{eq:LapseEquation}) as described in Subsection~\ref{ssection:eqs}.
  Assume, in addition,  that $V$ is a $\s_V$-admissible potential and that the basic supremum assumptions, see
  Definition~\ref{def:bas sup ass}, are satisfied with $3\ve\leq \s_V$. Then, for each $\ell\in\nn{}$, there is a $t_\ell>0$ such
  that for $t\leq t_\ell$,
  \begin{equation}\label{eq:lapse ito kEk etc}
    \begin{split}      
      & E_{(N),\ell}+2\textstyle{\sum}_{0<|\bfI|\leq \ell}\textstyle{\int}_{\S} \big(k_{IJ} E_{\bfI}k_{IJ} + \dtp E_{\bfI}\dtp \big)E_{\bfI}(N-1)\md\mu_{h_{\refer}}=\mR,
    \end{split}
  \end{equation}
  where
  \[
  |\mR|\leq C_\ell t^{\ve}\etot{\ell}+C_\ell\etot{\ell-1}^{1/2}\etot{\ell}^{1/2}.
  \]
  Moreover,
  \begin{equation}\label{eq:ear N Hl est}
    \begin{split}
      \|\ear N\|_{H^\ell(\S_t)}^2 \leq \|k\|_{H^\ell(\S_t)}^2+\|e_0\varphi\|_{H^\ell(\S_t)}^2
      +C_\ell t^{\ve}\etot{\ell}+C_\ell\etot{\ell-1}^{1/2}\etot{\ell}^{1/2}
    \end{split}
  \end{equation}
  for all $t\leq t_\ell$. Finally,
  \begin{equation}\label{eq:ENl final est lemma}
    \begin{split}
      \|\ear N\|_{H^\ell(\S_t)}^2+\tfrac{1}{2}t^{-2}\|N-1\|_{H^\ell(\S_t)}^2 \leq{} & 2\big(\|k\|_{H^\ell(\S_t)}^2+\|e_0\varphi\|_{H^\ell(\S_t)}^2\big)\\
      &\negmedspace +C_\ell t^{\ve}\etot{\ell}+C_l\etot{\ell-1}^{1/2}\etot{\ell}^{1/2}
    \end{split}
  \end{equation}
  for all $t\leq t_\ell$.
\end{lemma}
\begin{remark}
  One crude special case of (\ref{eq:ENl final est lemma}) is the following:
  \[
  E_{(N),\ell}=\|\ear N\|_{H^{\ell}(\S_t)}^2+t^{-2}\|N-1\|_{H^\ell(\S_t)}^2\leq C_\ell \etot{\ell}.
  \]
  In particular,
  \begin{equation}\label{eq:lapse estimate crude}
    \hetot{\ell}\leq C_\ell\etot{\ell}.
  \end{equation}
\end{remark}
\begin{proof}
  As in the proof of \cite[Lemma~112, p.~52]{OPR}, applying $E_{\bfI}$ to (\ref{eq: alternative lapse}) and then multiplying with $E_{\bfI}(N-1)$ yields
  \begin{equation*}
    \begin{split}
      & e_I E_\bfI e_I (N-1) E_{\bfI} (N-1) + [E_\bfI, e_I] e_I (N-1) E_{\bfI}(N-1)
      - t^{-2} \big(E_{\bfI}(N-1)\big)^2 \\
      ={} &  E_{\bfI} \big( \g_{JII} e_{J} (N) - N \big\{ t^{-2} - k_{IJ} k_{IJ} - \dtp \dtp + \tfrac{2}{n-1} V \circ \varphi \big\}\big)
      E_{\bfI}(N-1).
    \end{split}
  \end{equation*}
  Integrating the result over $\S$ yields
  \begin{equation}\label{eq:lapse first step}
    \begin{split}
      & -\textstyle{\int}_{\S} \big(E_\bfI e_I (N-1) E_{\bfI}e_I (N-1) +t^{-2} \big(E_{\bfI}(N-1)\big)^2 \big)\md\mu_{h_{\refer}}\\
      & -\textstyle{\int}_{\S} \big(E_\bfI e_I (N-1) [e_I,E_{\bfI}](N-1) - [E_\bfI, e_I] e_I (N-1) E_{\bfI}(N-1)\big)\md\mu_{h_{\refer}}\\
      & -\textstyle{\int}_{\S} E_\bfI e_I (N-1) E_{\bfI}(N-1)\rodiv_{h_{\refer}}e_I \md\mu_{h_{\refer}}\\
      ={} &  \textstyle{\int}_{\S}\big(E_{\bfI} \big( \g_{JII} e_{J} (N) - N \big\{ t^{-2} - k_{IJ} k_{IJ} - \dtp \dtp
      + \tfrac{2}{n-1} V \circ \varphi \big\} \big)
      E_{\bfI}(N-1)\big)\md\mu_{h_{\refer}},
    \end{split}
  \end{equation}
  where we used (\ref{eq:divergence formula}). Next, note that, due to (\ref{eq:diveIest gen}), 
  \begin{equation}\label{eq:diveI contr}
    \begin{split}
      & \textstyle{\sum}_I\sum_{|\bfI|\leq \ell}\big|\textstyle{\int}_{\S} E_\bfI e_I (N-1) E_{\bfI}(N-1)\rodiv_{h_{\refer}}e_I \md\mu_{h_{\refer}}\big|
      \leq Ct^{3\ve}E_{(N),\ell}.
    \end{split}
  \end{equation}
  Next, (\ref{eq:N be gen}) and (\ref{eq:EbfI eI comm est}) yield
  \[
  \|[E_{\bfI},e_I](N-1)\|_{L^2(\S_t)}\leq C_\ell t^{3\ve}E_{(N),\ell}^{1/2}+C_\ell t^{2\ve}E_{(e,\omega),\ell}^{1/2}.
  \]
  Thus
  \begin{equation}\label{eq:comm on N m one}
    \begin{split}
      & \big|\textstyle{\sum}_{|\bfI|\leq \ell}\sum_I\int_{\S_t}E_\bfI e_I (N-1) [e_I,E_{\bfI}](N-1)\md\mu_{h_{\refer}}\big|
      \leq C_\ell t^{2\ve}(E_{(N),\ell}+E_{(e,\omega),\ell}).
    \end{split}
  \end{equation}
  Next, combining (\ref{eq:ear N gen}) and  (\ref{eq:EbfI eI comm est}) yields
  \[
  \|[E_{\bfI},e_I]e_I(N-1)\|_{L^2(\S_t)}\leq C_\ell t^{-1+3\ve}\|\ear N\|_{H^\ell(\S_t)}+C_\ell t^{-1+5\ve}E_{(e,\omega),\ell}^{1/2}.
  \]
  Thus
  \begin{equation}\label{eq:comm on ear N}
    \begin{split}
      & \big|\textstyle{\sum}_{|\bfI|\leq \ell}\sum_I\textstyle{\int}_{\S} [E_\bfI, e_I] e_I (N-1) E_{\bfI}(N-1)\md\mu_{h_{\refer}}\big|
      \leq C_\ell t^{3\ve}(E_{(N),\ell}+E_{(e,\omega),\ell}).
    \end{split}
  \end{equation}
  Next, note that $\g_{JII} e_{J} (N)$ and $\tfrac{2N}{n-1} V \circ \varphi$ are good terms. In particular, (\ref{eq:good term}) yields 
  \begin{equation}\label{eq:geN pot contr lapse est}
    \|E_{\bfI}(\g_{JII} e_{J} (N))\|_{L^2(\S_t)}+\|E_{\bfI}\big(\tfrac{2N}{n-1} V \circ \varphi\big)\|_{L^2(\S_t)}
    \leq C_\ell t^{-1+\ve}\hetot{\ell}^{1/2}.
  \end{equation}  
  Combining (\ref{eq:lapse first step}), (\ref{eq:diveI contr}), (\ref{eq:comm on N m one}), (\ref{eq:comm on ear N})
  and (\ref{eq:geN pot contr lapse est}) yields
  \begin{equation}\label{eq:lapse second step}
    \begin{split}      
      & E_{(N),\ell}+\textstyle{\sum}_{|\bfI|\leq \ell}\textstyle{\int}_{\S}E_{\bfI} \big( N \big\{ -t^{-2}
      + k_{IJ} k_{IJ} + \dtp \dtp \big\} \big)E_{\bfI}(N-1)\md\mu_{h_{\refer}}=R_1,
    \end{split}
  \end{equation}
  where
  \[
  |R_1|\leq C_\ell t^{\ve}\hetot{\ell}.
  \]
  Next, we wish to estimate $E_{\bfI}(NR)$ in $L^2$, where 
  \[
  R:=-t^{-2} + k_{IJ} k_{IJ} + \dtp \dtp.
  \]
  Note that $E_\bfI(NR)$ can be written as a sum of $NE_{\bfI}R$ and a linear combination of terms of the form
  \[
  E_{\bfJ}E_jN\cdot E_{\bfK}R,
  \]
  where $|\bfJ|+|\bfK|=|\bfI|-1$. Due to Lemma~\ref{lemma: products in L2}, (\ref{eq:N be gen}) and (\ref{eq:almost asymptotic Hcon gen}), 
  \begin{equation*}
    \begin{split}
      \|E_{\bfJ}E_jN\cdot E_{\bfK}R\|_{L^2(\S_t)} &\leq C_\ell\|N-1\|_{C^1(\S_t)}\textstyle{\sum}_i\|E_iR\|_{H^{\ell-2}(\S_t)}+C_\ell\|R\|_{C^1(\S_t)}\|N-1\|_{H^\ell(\S_t)}\\
      &\leq C_\ell t^{\ve}\textstyle{\sum}_i\|E_iR\|_{H^{\ell-2}(\S_t)}+C_\ell t^{-2+2\ve}\|N-1\|_{H^\ell(\S_t)},
    \end{split}
  \end{equation*}
  where the first term on the right hand sides can be omitted if $|\bfK|=0$. Moreover,
  \[
  \|E_iR\|_{H^{\ell-2}(\S_t)}\leq C_\ell t^{-1}(\|k\|_{H^{\ell-1}(\S_t)}+\|e_0\varphi\|_{H^{\ell-1}(\S_t)})\leq C_\ell t^{-1}\etot{\ell-1}^{1/2},
  \]
  where we appealed to (\ref{eq:k phi be gen}) and Gagliardo-Nirenberg estimates. Summing up,
  \begin{equation*}
    \begin{split}
      \big|\textstyle{\int}_{\S}E_{\bfJ}E_jN\cdot E_{\bfK}R\cdot E_{\bfI}(N-1)\md\mu_{h_{\refer}}\big|
      \leq C_\ell t^{\ve}\etot{\ell-1}^{1/2}E_{(N),\ell}^{1/2}+C_\ell t^{2\ve}E_{(N),\ell}.
    \end{split}
  \end{equation*}
  Note also that, due to (\ref{eq:N be gen}) and (\ref{eq:almost asymptotic Hcon gen}), 
  \[
  \big|\textstyle{\int}_{\S}NR(N-1)\md\mu_{h_{\refer}}\big|\leq Ct^{-2+3\ve}.
\]
What remains is thus to consider expressions of the form 
  \begin{equation*}
    \begin{split}
      & \textstyle{\int}_{\S} N E_{\bfI} \big(k_{IJ} k_{IJ} + \dtp \dtp \big)E_{\bfI}(N-1)\md\mu_{h_{\refer}}\\
      ={} & \textstyle{\int}_{\S} 2\big(k_{IJ} E_{\bfI} k_{IJ} + \dtp E_{\bfI} \dtp \big)E_{\bfI}(N-1)\md\mu_{h_{\refer}}+R_2,
    \end{split}
  \end{equation*}  
  where $\bfI\neq 0$. Similarly to the discussion of the leading order terms, see Subsection~\ref{ssection:algorithm},
  \[
  |R_2|\leq C_\ell t^{\ve}\hetot{\ell}+C_\ell\etot{\ell-1}^{1/2}E_{(N),\ell}^{1/2}.
  \]
  Summing up, we conclude that
  \begin{equation}\label{eq:lapse third step}
    \begin{split}      
      & E_{(N),\ell}+2\textstyle{\sum}_{0<|\bfI|\leq \ell}\textstyle{\int}_{\S} \big(k_{IJ} E_{\bfI}k_{IJ} + \dtp E_{\bfI}\dtp \big)E_{\bfI}(N-1)\md\mu_{h_{\refer}}=R_3,
    \end{split}
  \end{equation}
  where
  \begin{equation}\label{eq:Rthree est}
    |R_3|\leq C_\ell t^{\ve}\hetot{\ell}+C_\ell\etot{\ell-1}^{1/2}E_{(N),\ell}^{1/2}+Ct^{-2+3\ve}.
  \end{equation}
  At this point, it is convenient to remember that there is a $t_a>0$ such that for $t\leq t_a$,
  \[
  t^2k_{IJ}k_{IJ}+t^2e_0(\varphi)e_0(\varphi)\geq\tfrac{1}{2};
  \]
  this is an immediate consequence of \eqref{eq:almost asymptotic Hcon gen}. Integrating this inequality over $\S$ yields the existence of a constant $c_a>0$
  such that 
  \[
  t^{-2}\leq c_a\textstyle{\int}_{\S}[k_{IJ}k_{IJ}+e_0(\varphi)e_0(\varphi)]\md\mu_{h_{\refer}}\leq c_a\etot{\ell}.
  \]
  In particular, this means that for $t\leq t_a$, we can remove the last term on the right hand side of (\ref{eq:Rthree est}). Next, note that   
  \[
  |2k_{IJ}E_{\bfI}k_{IJ}+2e_0\varphi\cdot E_{\bfI}e_0\varphi|\leq 2\big(\textstyle{\sum}_{I,J}k_{IJ}^2+(e_0\varphi)^2\big)^{1/2}
  \big(\textstyle{\sum}_{I,J}(E_{\bfI}k_{IJ})^2+(E_{\bfI}e_0\varphi)^2\big)^{1/2}
  \]
  This means that
  \begin{equation*}
    \begin{split}
      & \big|\textstyle{\sum}_{0<|\bfI|\leq \ell}\textstyle{\int}_{\S}
       \big( 2k_{IJ} E_{\bfI}k_{IJ} + 2e_0\varphi\cdot E_{\bfI}e_0\varphi \big)E_{\bfI}(N-1)\md\mu_{h_{\refer}}\big|\\
      \leq{} & 2\big\|\big(\textstyle{\sum}_{I,J}k_{IJ}^2+(e_0\varphi)^2\big)^{1/2}\big\|_{C^0(\S_t)}
      \big(\|k\|_{H^\ell(\S_t)}^2+\|e_0\varphi\|_{H^\ell(\S_t)}^2\big)^{1/2}\|N-1\|_{H^\ell(\S_t)}.
    \end{split}
  \end{equation*}
  On the other hand, (\ref{eq:almost asymptotic Hcon gen}) yields
  \begin{equation}\label{eq:Cz k ezphi bd}
    \begin{split}
      t\big\|\big(\textstyle{\sum}_{I,J}k_{IJ}^2+(e_0\varphi)^2\big)^{1/2}\big\|_{C^0(\S_t)}
      &\leq 1+\big\|\big(t^2\textstyle{\sum}_{I,J}k_{IJ}^2+t^2(e_0\varphi)^2\big)^{1/2}-1\big\|_{C^0(\S_t)}\\
      &\leq 1+Ct^{2\ve}.
    \end{split}
  \end{equation}
  To conclude
  \begin{equation*}
    \begin{split}
      & \big|\textstyle{\sum}_{0<|\bfI|\leq \ell}\textstyle{\int}_{\S}
      \big( 2k_{IJ} E_{\bfI}k_{IJ} + 2e_0\varphi\cdot E_{\bfI}e_0\varphi \big)E_{\bfI}(N-1)\md\mu_{h_{\refer}}\big|\\
      \leq{} & 2t^{-1}(1+Ct^{\ve})\big(\|k\|_{H^\ell(\S_t)}^2+\|e_0\varphi\|_{H^\ell(\S_t)}^2\big)^{1/2}\|N-1\|_{H^\ell(\S_t)}.
    \end{split}
  \end{equation*}
  Combining this estimate with (\ref{eq:lapse third step}) yields the conclusion that, for $t\leq t_a$, 
  \begin{equation}\label{eq:ENl first step}
    \begin{split}
      E_{(N),\ell} &\leq  2t^{-1}\big(\|k\|_{H^\ell(\S_t)}^2+\|e_0\varphi\|_{H^\ell(\S_t)}^2\big)^{1/2}\|N-1\|_{H^\ell(\S_t)}\\
      &\quad +C_\ell t^{\ve}(E_{(N),\ell}+\etot{\ell})+C_\ell\etot{\ell-1}^{1/2}E_{(N),\ell}^{1/2}.
    \end{split}
  \end{equation}
  Using the estimate $2ab\leq a^2/2+2b^2$ on the first term on the right hand side, it follows that
  \begin{equation*}
    \begin{split}
      \|\ear N\|_{H^{\ell}(\S_t)}^2+\tfrac{1}{2}t^{-2}\|N-1\|_{H^\ell(\S_t)}^2&\leq  2\big(\|k\|_{H^\ell(\S_t)}^2+\|e_0\varphi\|_{H^\ell(\S_t)}^2\big)\\
      &\quad +C_\ell t^{\ve}(E_{(N),\ell}+\etot{\ell})+C_\ell\etot{\ell-1}^{1/2}E_{(N),\ell}^{1/2}
      \end{split}
  \end{equation*}
  for $t\leq t_a$. By a similar argument, it is, from this estimate, straightforward to deduce that there is, for each $\ell$, a $t_\ell>0$ and a $C_\ell$
  such that
  \begin{equation}\label{eq:ENl ito El}
    E_{(N),\ell}\leq C_\ell\etot{\ell}
  \end{equation}
  for all $t\leq t_\ell$. Inserting this information into (\ref{eq:ENl first step}) yields
  \begin{equation}\label{eq:ENl final step}
    \begin{split}
      E_{(N),\ell} &\leq 2t^{-1}\big(\|k\|_{H^\ell(\S_t)}^2+\|e_0\varphi\|_{H^\ell(\S_t)}^2\big)^{1/2}\|N-1\|_{H^\ell(\S_t)}\\
      &\quad +C_\ell t^{\ve}\etot{\ell}+C_\ell\etot{\ell-1}^{1/2}\etot{\ell}^{1/2}.
    \end{split}
  \end{equation}
  Finally, combining (\ref{eq:lapse third step}) with (\ref{eq:ENl ito El}) yields (\ref{eq:lapse ito kEk etc}). Moreover,
  (\ref{eq:ENl final step}) and arguments similar to the above yield (\ref{eq:ear N Hl est}) and (\ref{eq:ENl final est lemma}).
  The estimate (\ref{eq:lapse estimate crude}) follows from (\ref{eq:ENl ito El}). The lemma follows
\end{proof}

\subsection{Estimating the frame}\label{ssection:frame est hod}
Next, we estimate the time derivative of $E_{(e,\omega),\ell}$ from below.
\begin{lemma}\label{lemma:dt Eeomega lb}
  Consider a solution to (\ref{eq: transport frame})--(\ref{eq:LapseEquation}) as described in Subsection~\ref{ssection:eqs}.
  Assume, in addition,  that $V$ is a $\s_V$-admissible potential and that the basic supremum assumptions, see
  Definition~\ref{def:bas sup ass}, are satisfied with $3\ve\leq \s_V$. Then, for each $\ell\in\nn{}$, there are constants $C_\ell$ and
  $t_\ell>0$ such that for $t<t_\ell$,
  \[
  \d_t E_{(e,\omega),\ell}\geq -2(1+\ve)t^{-1}E_{(e,\omega),\ell}-C_\ell t^{-1+\ve}\etot{\ell}-C_{\ell}t^{-1}\etot{\ell-1}^{1/2}\etot{\ell}^{1/2}.
  \]
\end{lemma}
\begin{proof}
  When the time derivative hits the factor $t^{-2\ve}$, the result is $-2\ve t^{-1}E_{(e,\omega),\ell}$. Next, let us focus on the energy associated
  with the frame. This means we have to estimate 
  \begin{equation}\label{eq:time der frame energy}
    2t^{-2\ve}\textstyle{\sum}_{I,i}\textstyle{\int}_{\S}E_{\bfI}e_I^i\cdot E_{\bfI}\d_t e_I^i\md\mu_{h_{\refer}}
    =2t^{-2\ve}\textstyle{\sum}_{I,J,i}\textstyle{\int}_{\S}E_{\bfI}e_I^i\cdot E_{\bfI}(-Nk_{IJ}e_J^i)\md\mu_{h_{\refer}}.
  \end{equation} 
  The expression $Nk_{IJ}t^{-\ve}e_J^i$ is a bad term. In particular, due to the arguments presented in Subsection~\ref{ssection:algorithm},
  it can be written
  \[
  E_{\bfI}(Nk_{IJ}t^{-\ve}e_J^i)=k_{IJ}t^{-\ve}E_{\bfI}e_J^i+R_B,
  \]
  where $R_B$ satisfies the estimate
  \[
  \|R_B\|_{L^2(\S_t)}\leq C_\ell t^{-1+\ve}\etot{\ell}^{1/2}+C_\ell t^{-1}\etot{\ell-1}^{1/2};
  \]
  cf. (\ref{eq:RB est}) and (\ref{eq:lapse estimate crude}). Next, due to (\ref{eq:Cz k ezphi bd}), 
  \[
  -2t^{-2\ve}\textstyle{\sum}_{I,J,i}\textstyle{\int}_{\S}k_{IJ}E_{\bfI}e_I^i\cdot E_{\bfI}e_J^i\md\mu_{h_{\refer}}
  \geq -2(1+Ct^{\ve})t^{-1}t^{-2\ve}\textstyle{\sum}_{I,i}\textstyle{\int}_{\S}|E_{\bfI}e_I^i|^2 \md\mu_{h_{\refer}}
  \]
  Due to this estimate and an analogous estimate concerning the energy for $\omega$, the lemma follows. 
\end{proof}

\subsection{Estimating the second fundamental form and the structure coefficients}
Next, we estimate the energy associated with the second fundamental form and the structure coefficients. 
\begin{lemma}\label{lemma:E gamma k est}
  Consider a solution to (\ref{eq: transport frame})--(\ref{eq:LapseEquation}) as described in Subsection~\ref{ssection:eqs}.
  Assume, in addition,  that $V$ is a $\s_V$-admissible potential and that the basic supremum assumptions, see
  Definition~\ref{def:bas sup ass}, are satisfied with $3\ve\leq \s_V$. Then, for each $\ell\in\nn{}$, there is a $t_\ell>0$ such that for $t\leq t_\ell$,
  \begin{equation}\label{eq:dt E gamma k}
    \begin{split}
      \d_t E_{(\gamma,k),\ell} = & -\textstyle{\sum}_{I,J,K,L}\sum_{|\bfI|\leq \ell}\int_{\S}k_{IL}E_{\bfI}(\gamma_{IJK})E_{\bfI}(\gamma_{LJK})\md\mu_{h_{\refer}}\\
      & -\textstyle{\sum}_{I,J,K,L}\sum_{|\bfI|\leq \ell}\int_{\S}k_{JL}E_{\bfI}(\gamma_{IJK})E_{\bfI}(\gamma_{ILK})\md\mu_{h_{\refer}}\\
      & +\textstyle{\sum}_{I,J,K,L}\sum_{|\bfI|\leq \ell}\int_{\S}k_{KL}E_{\bfI}(\gamma_{IJK})E_{\bfI}(\gamma_{IJL})\md\mu_{h_{\refer}}\\
      & -2t^{-1}\textstyle{\sum}_{I,J}\sum_{|\bfI|\leq \ell}\int_{\S}|E_{\bfI}(k_{IJ})|^2\md\mu_{h_{\refer}}\\
      & -2t^{-1}\textstyle{\sum}_{I,J}\sum_{0<|\bfI|\leq \ell}\int_{\S}k_{IJ}E_{\bfI}(N)E_{\bfI}(k_{IJ})\md\mu_{h_{\refer}}\\        
      & -2\textstyle{\sum}_{I,J,L}\sum_{|\bfI|\leq \ell}\int_{\S}k_{LJ}E_{\bfI}(\g_{LII})E_{\bfI}(e_{J} (N))\md\mu_{h_{\refer}}\\
      & -2\textstyle{\sum}_{J}\sum_{|\bfI|\leq \ell}\int_{\S}e_{0}(\varphi) E_{\bfI}[e_{J}(\varphi)]E_{\bfI}(e_{J} (N))\md\mu_{h_{\refer}}\\
      & +2\textstyle{\sum}_{I,J,L,K}\sum_{|\bfI|\leq \ell}\int_{\S}(k_{LJ}E_{\bfI}(\g_{LII}) + k_{IL}E_{\bfI}(\g_{IJL}))E_{\bfI}(\g_{JKK})\md\mu_{h_{\refer}}\\
      & +2\textstyle{\sum}_{J,K}\sum_{|\bfI|\leq \ell}\int_{\S}e_{0}(\varphi) E_{\bfI}[e_{J}(\varphi)]E_{\bfI}(\g_{JKK})\md\mu_{h_{\refer}}+\mR,
    \end{split}
  \end{equation}
  where
  \[
  |\mR|\leq C_\ell t^{-1+\ve}\etot{\ell}+C_\ell t^{-1}\etot{\ell-1}^{1/2}\etot{\ell}^{1/2}.
  \]
\end{lemma}
\begin{proof}
  Note, to begin with, that (\ref{eq:concoef}) and (\ref{eq:sff}) can be written
  \begin{align}
    \d_t\g_{IJK} &= - 2 e_{[I}(Nk_{J]K}) +R_{\gamma,IJK},\\
    \d_tk_{IJ} &= e_{(I}e_{J)} (N) - e_{K} (N\g_{K(IJ)})
    - Ne_{(I}(\g_{J)KK})+R_{k,IJ},\label{eq:dtkIJ}
  \end{align}
  where
  \begin{subequations}
    \begin{align}
    R_{\gamma,IJK} := & - Nk_{IL} \g_{LJK} - Nk_{JL} \g_{ILK} + Nk_{KL} \g_{IJL},\label{eq:Rgammadef}\\
    R_{k,IJ} := & - t^{-1}Nk_{IJ} +N\g_{KLL}\g_{K(IJ)} + N\g_{I(KL)}\g_{J(KL)} - \tfrac14 N\g_{KLI} \g_{KLJ}\label{eq:Rkdef}\\
                      & + Ne_I(\varphi)e_J(\varphi)+\tfrac{2N}{n-1}(V\circ\varphi) \delta_{IJ}.\nonumber
    \end{align}
  \end{subequations}  
  Differentiating the energy for $k$ and $\gamma$ then yields
  \begin{equation*}
    \begin{split}
      \d_t E_{(\gamma,k),\ell} :={} & \textstyle{\sum}_{I,J,K}\sum_{|\bfI|\leq \ell}\int_{\S}E_{\bfI}(\gamma_{IJK})E_{\bfI}(- 2 e_{[I}(Nk_{J]K}) +R_{\gamma,IJK})\md\mu_{h_{\refer}}\\
      &\negmedspace +2\textstyle{\sum}_{I,J}\sum_{|\bfI|\leq \ell}\int_{\S}E_{\bfI}(k_{IJ})E_{\bfI}(e_{(I}e_{J)} (N) - e_{K} (N\g_{K(IJ)}))\md\mu_{h_{\refer}}\\
      &\negmedspace +2\textstyle{\sum}_{I,J}\sum_{|\bfI|\leq \ell}\int_{\S}E_{\bfI}(k_{IJ})E_{\bfI}(- Ne_{(I}(\g_{J)KK})+R_{k,IJ})\md\mu_{h_{\refer}}.
    \end{split}
  \end{equation*}  
  The terms arising from $R_{\gamma,IJK}$ and $R_{k,IJ}$ can be estimated directly in terms of the energy. However, the remaining terms contain too many
  derivatives. In order to estimate the corresponding expressions, we need to integrate by parts, appeal to the momentum constraint, etc. In order
  to emphasise this division it is useful to, using the symmetries of $\g_{IJK}$ and $k_{IJ}$, rewrite the time derivative of the energy as
  \begin{equation*}
    \begin{split}
      \d_t E_{(\gamma,k),\ell} :={} & \textstyle{\sum}_{I,J,K}\sum_{|\bfI|\leq \ell}\int_{\S}E_{\bfI}(\gamma_{IJK})E_{\bfI}(- 2 e_{I}(Nk_{JK}))\md\mu_{h_{\refer}}\\
      &\negmedspace +2\textstyle{\sum}_{I,J}\sum_{|\bfI|\leq \ell}\int_{\S}E_{\bfI}(k_{IJ})E_{\bfI}(e_{I}e_{J} (N) - e_{K} (N\g_{KIJ}))\md\mu_{h_{\refer}}\\
      &\negmedspace +2\textstyle{\sum}_{I,J}\sum_{|\bfI|\leq \ell}\int_{\S}E_{\bfI}(k_{IJ})E_{\bfI}(- Ne_{I}(\g_{JKK}))\md\mu_{h_{\refer}}\\
      &\negmedspace +\textstyle{\sum}_{I,J,K}\sum_{|\bfI|\leq \ell}\int_{\S}E_{\bfI}(\gamma_{IJK})E_{\bfI}(R_{\gamma,IJK})\md\mu_{h_{\refer}}\\
      &\negmedspace +2\textstyle{\sum}_{I,J}\sum_{|\bfI|\leq \ell}\int_{\S}E_{\bfI}(k_{IJ})E_{\bfI}(R_{k,IJ})\md\mu_{h_{\refer}}.
    \end{split}
  \end{equation*}
  Let us begin by considering
  \begin{equation}\label{eq:gamma fbt}
    \textstyle{\sum}_{I,J,K}\sum_{|\bfI|\leq \ell}\textstyle{\int}_{\S}E_{\bfI}(\gamma_{IJK})E_{\bfI}(- 2 e_{I}(Nk_{JK}))\md\mu_{h_{\refer}}.
  \end{equation}
  Due to (\ref{eq:k phi be gen}), (\ref{eq:N be gen}), (\ref{eq:est lead order terms}), (\ref{eq:EbfI eI comm est}) and (\ref{eq:lapse estimate crude}), this
  expression equals
  \begin{equation}\label{eq:k main term pi}
    -2\textstyle{\sum}_{I,J,K}\sum_{|\bfI|\leq \ell}\textstyle{\int}_{\S}E_{\bfI}(\gamma_{IJK}) e_{I}E_{\bfI}(Nk_{JK})\md\mu_{h_{\refer}}
  \end{equation}
  up to terms that can be estimated by $C_\ell t^{-1+\ve}\etot{\ell}$. Next, note that $e_IE_{\bfI}(Nk_{JK})$ can be written as a sum of
  \begin{equation}\label{eq:imp terms k pi}
    e_I[NE_{\bfI}(k_{JK})],\ \
    E_{\bfI}(e_IN)k_{JK},
  \end{equation}
  of
  \begin{equation}\label{eq:junk terms k ft}
    [e_I,E_{\bfI}](N)k_{JK},\ \
    E_{\bfI}(N)e_Ik_{JK}
  \end{equation}
  and linear combinations of terms of the form
  \begin{equation}\label{eq:junk terms k st}
    e_I[E_{\bfI_1}E_{i_1}(N)E_{\bfI_2}E_{i_2}(k_{JK})].
  \end{equation}
  Moreover, if $\bfI=0$, then the only term we need to handle is the first expression in (\ref{eq:imp terms k pi}). The contribution to
  (\ref{eq:k main term pi}) from the first expression in (\ref{eq:imp terms k pi}) can be rewritten by integrating by parts; cf.
  (\ref{eq:divergence formula}). When doing so, there is one term involving the divergence of $e_I$. Due to (\ref{eq:diveIest gen}), this
  term can be estimated in absolute value by $C_\ell t^{-1+3\ve}\etot{\ell}$. What remains to be considered is
  \[
  2\textstyle{\sum}_{I,J,K}\sum_{|\bfI|\leq \ell}\textstyle{\int}_{\S}e_IE_{\bfI}(\gamma_{IJK}) NE_{\bfI}(k_{JK})\md\mu_{h_{\refer}}.
  \]
  Appealing to (\ref{eq:e om g phi be gen}) and (\ref{eq:EbfI eI comm est}), we can commute $e_I$ and $E_{\bfI}$ in the first factor in the
  integrand; the commutator term can be estimated by $C_\ell t^{-1+3\ve}\etot{\ell}$. What remains is then
  \begin{equation}\label{eq:gamma fbt frt}
    2\textstyle{\sum}_{I,J,K}\sum_{|\bfI|\leq \ell}\textstyle{\int}_{\S}E_{\bfI}(e_I\gamma_{IJK}) NE_{\bfI}(k_{JK})\md\mu_{h_{\refer}}.
  \end{equation}
  The contribution (\ref{eq:k main term pi}) from the second expression in (\ref{eq:imp terms k pi}), we keep as it is:
  \begin{equation}\label{eq:gamma fbt srt}
    -2\textstyle{\sum}_{I,J,K}\sum_{|\bfI|\leq \ell}\textstyle{\int}_{\S}E_{\bfI}(\gamma_{IJK}) E_{\bfI}(e_{I}N)k_{JK}\md\mu_{h_{\refer}}.
  \end{equation}
  In order to estimate the contribution to (\ref{eq:k main term pi}) from the first expression in (\ref{eq:junk terms k ft}), it is sufficient
  to appeal to (\ref{eq:k phi be gen}), (\ref{eq:N be gen}), (\ref{eq:EbfI eI comm est}) and (\ref{eq:lapse estimate crude}). The
  contribution can be estimated by
  \begin{equation}\label{eq:protot error term}
    C_\ell t^{-1+\ve}\etot{\ell}
  \end{equation}
  (in fact, a better estimate holds). The contribution to (\ref{eq:k main term pi}) from the second expression in (\ref{eq:junk terms k ft})
  can, keeping in mind that we can assume that $\bfI\neq 0$, also be bounded by (\ref{eq:protot error term}); cf. (\ref{eq:e om g phi be gen}),
  (\ref{eq:k phi be gen}) and (\ref{eq:lapse estimate crude}). Finally, due to Gagliardo-Nirenberg estimates, (\ref{eq:estdynvarmfltot gen})
  and (\ref{eq:lapse estimate crude}), the contribution from (\ref{eq:junk terms k st}) can be bounded by
  (\ref{eq:protot error term}). To summarise, up to terms that can be estimated in absolute value by (\ref{eq:protot error term}), the expression
  (\ref{eq:gamma fbt}) can be written as a sum of (\ref{eq:gamma fbt frt}) and (\ref{eq:gamma fbt srt}). 

  Next, consider
  \begin{equation}\label{eq:gamma tmd}
    2\textstyle{\sum}_{I,J}\sum_{|\bfI|\leq \ell}\int_{\S}E_{\bfI}(k_{IJ})E_{\bfI}(- e_{K} (N\g_{KIJ}))\md\mu_{h_{\refer}}.
  \end{equation}
  Note that $E_{\bfI}(- e_{K} (N\g_{KIJ}))$ can be written as a sum of 
  \begin{equation}\label{eq:gamma tmd fjt}
    E_{\bfI}(- e_{K} (N)\g_{KIJ}),\ \ \
    E_{\bfI}(-Ne_K\g_{KIJ})+NE_{\bfI}(e_K\g_{KIJ}),
  \end{equation}
  and
  \[
    -NE_{\bfI}(e_K\g_{KIJ}).
  \]
  Since the first expression in (\ref{eq:gamma tmd fjt}) is a good term, the contribution to (\ref{eq:gamma tmd}) from the first expression in
  (\ref{eq:gamma tmd fjt}) can be bounded by (\ref{eq:protot error term}); cf. (\ref{eq:good term}) and (\ref{eq:lapse estimate crude}).
  Appealing to Gagliardo-Nirenberg estimates, (\ref{eq:estdynvarmfltot gen}) and (\ref{eq:lapse estimate crude}), the contribution to
  (\ref{eq:gamma tmd}) from the second expression in (\ref{eq:gamma tmd fjt}) can be bounded by (\ref{eq:protot error term}) as well. What remains is
  \[
    -2\textstyle{\sum}_{I,J,K}\sum_{|\bfI|\leq \ell}\int_{\S}E_{\bfI}(k_{IJ})NE_{\bfI}(e_{K} \g_{KIJ})\md\mu_{h_{\refer}}.
  \]
  Clearly, this term cancels (\ref{eq:gamma fbt frt}). To summarise,
  \begin{equation}\label{eq:fbt good comb}
    \begin{split}
      & \textstyle{\sum}_{I,J,K}\sum_{|\bfI|\leq \ell}\textstyle{\int}_{\S}E_{\bfI}(\gamma_{IJK})E_{\bfI}(- 2 e_{I}(Nk_{JK}))\md\mu_{h_{\refer}}\\
      & +2\textstyle{\sum}_{I,J,K}\sum_{|\bfI|\leq \ell}\int_{\S}E_{\bfI}(k_{IJ})E_{\bfI}(- e_{K} (N\g_{KIJ}))\md\mu_{h_{\refer}}\\
      = & -2\textstyle{\sum}_{I,J,K}\sum_{|\bfI|\leq \ell}\textstyle{\int}_{\S}E_{\bfI}(\gamma_{IJK}) E_{\bfI}(e_{I}N)k_{JK}\md\mu_{h_{\refer}}+\dots,
    \end{split}
  \end{equation}
  where the dots signify terms that can be estimated in absolute value by (\ref{eq:protot error term}). 
  
  Next, let us turn to
  \begin{equation}\label{eq:eIeJN fbt}
    2\textstyle{\sum}_{I,J}\sum_{|\bfI|\leq \ell}\int_{\S}E_{\bfI}(k_{IJ})E_{\bfI}(e_{I}e_{J} (N))\md\mu_{h_{\refer}}.
  \end{equation}
  Note that $E_{\bfI}(e_{I}e_{J}(N))$ can be written as a sum of
  \begin{equation}\label{eq:eIeJ fbt division}
    [E_{\bfI},e_I](e_JN),\ \ \
    e_IE_{\bfI}(e_JN).
  \end{equation}
  Due to (\ref{eq:ear N gen}), (\ref{eq:EbfI eI comm est}) and (\ref{eq:lapse estimate crude}), the contribution to (\ref{eq:eIeJN fbt}) from the
  first expression in (\ref{eq:eIeJ fbt division}) can be bounded by (\ref{eq:protot error term}). The contribution of the second
  expression in (\ref{eq:eIeJ fbt division}) to (\ref{eq:eIeJN fbt}) can be estimated by first integrating by parts. This results in two terms. The term
  which involves the divergence of $e_I$ can be bounded by (\ref{eq:protot error term}). The remaining term is
  \[
    -2\textstyle{\sum}_{I,J}\sum_{|\bfI|\leq \ell}\int_{\S}e_IE_{\bfI}(k_{IJ})E_{\bfI}(e_{J} (N))\md\mu_{h_{\refer}}.
  \]
  Commuting $e_I$ and $E_{\bfI}$ in the first factor in the integrand leads to terms that can be estimated in absolute value by (\ref{eq:protot error term});
  cf. (\ref{eq:k phi be gen}) and (\ref{eq:EbfI eI comm est}). This means that we are left with the problem of estimating
  \begin{equation}\label{eq:contr MC ft}
    -2\textstyle{\sum}_{I,J}\sum_{|\bfI|\leq \ell}\int_{\S}E_{\bfI}(e_Ik_{IJ})E_{\bfI}(e_{J} (N))\md\mu_{h_{\refer}}.
  \end{equation}
  In order to estimate this expression, we appeal to (\ref{eq:MC}). However, we only write down the details later.

  Next, consider
  \begin{equation}\label{eq:eIgJKK k fbt}
    2\textstyle{\sum}_{I,J,K}\sum_{|\bfI|\leq \ell}\int_{\S}E_{\bfI}(k_{IJ})E_{\bfI}(- Ne_{I}(\g_{JKK}))\md\mu_{h_{\refer}}.
  \end{equation}
  Note that
  \begin{equation}\label{eq:eIgJKK k fbt div}
    - Ne_{I}(\g_{JKK})=e_{I}(N)\g_{JKK}- e_{I}(N\g_{JKK}).
  \end{equation}
  The first term on the right hand side of (\ref{eq:eIgJKK k fbt div}) is a good term, and the corresponding contribution to (\ref{eq:eIgJKK k fbt})
  can be estimated by (\ref{eq:protot error term}) due to (\ref{eq:good term}) and (\ref{eq:lapse estimate crude}). What remains is
  \[
    -2\textstyle{\sum}_{I,J,K}\sum_{|\bfI|\leq \ell}\int_{\S}E_{\bfI}(k_{IJ})E_{\bfI}e_{I}(N\g_{JKK})\md\mu_{h_{\refer}}.
  \]
  At this point we can commute $E_{\bfI}$ and $e_I$, integrate by parts, and commute $e_I$ and $E_{\bfI}$ again. This results, by arguments similar to
  the above, in
  \begin{equation}\label{eq:contr MC st}
    2\textstyle{\sum}_{I,J,K}\sum_{|\bfI|\leq \ell}\int_{\S}E_{\bfI}(e_Ik_{IJ})E_{\bfI}(N\g_{JKK})\md\mu_{h_{\refer}}
  \end{equation}
  plus terms that can be estimated in absolute value by (\ref{eq:protot error term}). Again, we can appeal to (\ref{eq:MC}) in order to estimate this
  expression.

  Next, let us turn to
  \begin{equation}\label{eq:Rgamma term}
    \textstyle{\sum}_{I,J,K}\sum_{|\bfI|\leq \ell}\int_{\S}E_{\bfI}(\gamma_{IJK})E_{\bfI}(R_{\gamma,IJK})\md\mu_{h_{\refer}},
  \end{equation}
  where $R_{\g,IJK}$ is introduced in (\ref{eq:Rgammadef}). The three terms appearing in $R_{\gamma,IJK}$ are all bad terms; cf.
  Subsection~\ref{ssection:algorithm}. In fact, due to (\ref{eq:EbfI bad term}), (\ref{eq:RB est}) and (\ref{eq:lapse estimate crude}),
  \begin{equation*}
    \begin{split}
      & \textstyle{\sum}_{I,J,K}\sum_{|\bfI|\leq \ell}\int_{\S}E_{\bfI}(\gamma_{IJK})E_{\bfI}(R_{\gamma,IJK})\md\mu_{h_{\refer}}\\
      = & -\textstyle{\sum}_{I,J,K,L}\sum_{|\bfI|\leq \ell}\int_{\S}k_{IL}E_{\bfI}(\gamma_{IJK})E_{\bfI}(\gamma_{LJK})\md\mu_{h_{\refer}}\\
      & -\textstyle{\sum}_{I,J,K,L}\sum_{|\bfI|\leq \ell}\int_{\S}k_{JL}E_{\bfI}(\gamma_{IJK})E_{\bfI}(\gamma_{ILK})\md\mu_{h_{\refer}}\\
      & +\textstyle{\sum}_{I,J,K,L}\sum_{|\bfI|\leq \ell}\int_{\S}k_{KL}E_{\bfI}(\gamma_{IJK})E_{\bfI}(\gamma_{IJL})\md\mu_{h_{\refer}}+\dots,
    \end{split}
  \end{equation*}
  where the dots represent terms that can be estimated in absolute value by
  \begin{equation}\label{eq:protot error term vtwo}
    C_\ell t^{-1+\ve}\etot{\ell}+C_\ell t^{-1}\etot{\ell-1}^{1/2}\etot{\ell}^{1/2}.
  \end{equation}
  Next, consider
  \begin{equation}\label{eq:Rk contribution}
    2\textstyle{\sum}_{I,J}\sum_{|\bfI|\leq \ell}\int_{\S}E_{\bfI}(k_{IJ})E_{\bfI}(R_{k,IJ})\md\mu_{h_{\refer}},
  \end{equation}
  where $R_{k,IJ}$ is given by (\ref{eq:Rkdef}). The first term appearing in (\ref{eq:Rkdef}) is a leading order term; see
  Subsection~\ref{ssection:algorithm}. In fact, due to (\ref{eq:leading order terms}), (\ref{eq:leading order terms RB})
  and (\ref{eq:lapse estimate crude}), it gives the following contribution to (\ref{eq:Rk contribution}):
  \begin{equation*}
    \begin{split}
      & -2t^{-1}\textstyle{\sum}_{I,J}\sum_{|\bfI|\leq \ell}\int_{\S}|E_{\bfI}(k_{IJ})|^2\md\mu_{h_{\refer}}\\
      & -2t^{-1}\textstyle{\sum}_{I,J}\sum_{0<|\bfI|\leq \ell}\int_{\S}k_{IJ}E_{\bfI}(N)E_{\bfI}(k_{IJ})\md\mu_{h_{\refer}}+\dots,
    \end{split}
  \end{equation*}  
  where the dots represent terms that can be estimated in absolute value by (\ref{eq:protot error term vtwo}). The remaining terms appearing
  in $R_{k,IJ}$ are all good terms, and their contribution can, due to (\ref{eq:good term}) and (\ref{eq:lapse estimate crude}), be estimated by
  (\ref{eq:protot error term vtwo}). To summarise,
  \begin{equation*}
    \begin{split}
      & 2\textstyle{\sum}_{I,J}\sum_{|\bfI|\leq \ell}\int_{\S}E_{\bfI}(k_{IJ})E_{\bfI}(R_{k,IJ})\md\mu_{h_{\refer}}\\
      = & -2t^{-1}\textstyle{\sum}_{I,J}\sum_{|\bfI|\leq \ell}\int_{\S}|E_{\bfI}(k_{IJ})|^2\md\mu_{h_{\refer}}\\
      & -2t^{-1}\textstyle{\sum}_{I,J}\sum_{0<|\bfI|\leq \ell}\int_{\S}k_{IJ}E_{\bfI}(N)E_{\bfI}(k_{IJ})\md\mu_{h_{\refer}}+\dots,
    \end{split}
  \end{equation*}
  where the dots represent terms that can be estimated in absolute value by (\ref{eq:protot error term vtwo}).
  At this point, we return to (\ref{eq:contr MC ft}). Due to (\ref{eq:MC}), we need to estimate
  \begin{equation}\label{eq:contr MC ft exp}
    \begin{split}
      -2\textstyle{\sum}_{J}\sum_{|\bfI|\leq \ell}\int_{\S}E_{\bfI}(\g_{LII}k_{LJ} + \g_{IJL}k_{IL} + e_{0}(\varphi) e_{J}(\varphi))E_{\bfI}(e_{J} (N))\md\mu_{h_{\refer}}.
    \end{split}
  \end{equation}
  Since all the terms appearing on the right hand side of (\ref{eq:MC}) are bad terms, (\ref{eq:EbfI bad term}), (\ref{eq:RB est}),
  (\ref{eq:lapse estimate crude}) and the arguments preceding (\ref{eq:contr MC ft}) yield the conclusion that 
    \begin{equation}\label{eq:contr MC ft exp lb}
      \begin{split}
        & 2\textstyle{\sum}_{I,J}\sum_{|\bfI|\leq \ell}\int_{\S}E_{\bfI}(k_{IJ})E_{\bfI}(e_{I}e_{J} (N))\md\mu_{h_{\refer}}\\
        = & -2\textstyle{\sum}_{J}\sum_{|\bfI|\leq \ell}\int_{\S}(k_{LJ}E_{\bfI}(\g_{LII}) - k_{IL}E_{\bfI}(\g_{JIL})
        + e_{0}(\varphi) E_{\bfI}[e_{J}(\varphi)])E_{\bfI}(e_{J} (N))\md\mu_{h_{\refer}}+\dots,
      \end{split}
    \end{equation}
    where the dots represent terms that can be estimated in absolute value by (\ref{eq:protot error term vtwo}).
    
    Finally, we need to estimate the contribution from (\ref{eq:contr MC st}). Again, due to (\ref{eq:MC}), we need to estimate
    \[
      2\textstyle{\sum}_{J,K}\sum_{|\bfI|\leq \ell}\int_{\S}E_{\bfI}(\g_{LII}k_{LJ} + \g_{IJL}k_{IL} + e_{0}(\varphi) e_{J}(\varphi))E_{\bfI}(N\g_{JKK})\md\mu_{h_{\refer}}.
    \]
    Note that $N\g_{JKK}$ can be replaced by $\g_{JKK}$ up to an error that can be estimated in absolute value by $C_\ell t^{-1+\ve}\etot{\ell}$. For that reason,
    we, from now on, effectively set $N$ equal to $1$. The remaining terms can be estimated as in the argument concerning (\ref{eq:contr MC ft}).
    Combining these arguments with the observations made immediately prior to (\ref{eq:contr MC st}), it follows that 
    \begin{equation}\label{eq:contr MC st exp lb}
      \begin{split}
        & 2\textstyle{\sum}_{I,J,K}\sum_{|\bfI|\leq \ell}\int_{\S}E_{\bfI}(k_{IJ})E_{\bfI}(-Ne_I(\g_{JKK}))\md\mu_{h_{\refer}}\\
        ={} & 2\textstyle{\sum}_{J,K}\sum_{|\bfI|\leq \ell}\int_{\S}(k_{LJ}E_{\bfI}(\g_{LII}) + k_{IL}E_{\bfI}(\g_{IJL})
        + e_{0}(\varphi) E_{\bfI}[e_{J}(\varphi)])E_{\bfI}(\g_{JKK})\md\mu_{h_{\refer}}+\dots,
      \end{split}
    \end{equation}
    where the dots represent terms that can be estimated in absolute value by (\ref{eq:protot error term vtwo}).
    At this point, we can sum up the estimates in order to obtain (\ref{eq:dt E gamma k}), where we used the fact that there is
    a cancellation between the first term on the right hand side of (\ref{eq:fbt good comb}) and one of the terms arising from
    (\ref{eq:contr MC ft exp lb}). 
\end{proof}

\subsection{The scalar field}\label{ssection:the scalar field hod}
Next, we turn to energy estimates for the scalar field. 
\begin{lemma}\label{lemma:sf}
  Consider a solution to (\ref{eq: transport frame})--(\ref{eq:LapseEquation}) as described in Subsection~\ref{ssection:eqs}.
  Assume, in addition,  that $V$ is a $\s_V$-admissible potential and that the basic supremum assumptions, see
  Definition~\ref{def:bas sup ass}, are satisfied with $3\ve\leq \s_V$. Then, for each $\ell\in\nn{}$, there is a $t_\ell>0$ such that for $t\leq t_\ell$,
  \begin{equation}\label{eq:dt Ephi est}
    \begin{split}
      \d_tE_{(\varphi),\ell} = &  -2t^{-1}\textstyle{\sum}_{|\bfI|\leq \ell}\int_{\S}|E_{\bfI}(e_{0}\varphi)|^2\md\mu_{h_{\refer}}
      -2\textstyle{\sum}_{|\bfI|\leq \ell}\textstyle{\sum}_{I,J}\int_{\S}k_{IJ}(E_{\bfI}e_I\varphi)(E_{\bfI}e_J\varphi)\md\mu_{h_{\refer}}\\
      & -2t^{-1}(1-\ve)t^{-2+2\ve}\textstyle{\sum}_{|\bfI|\leq \ell}\int_{\S}|E_{\bfI}\varphi|^2\md\mu_{h_{\refer}}\\
      & +2\textstyle{\sum}_{|\bfI|\leq \ell}\textstyle{\sum}_{I}\int_{\S}e_0(\varphi)(E_{\bfI}e_I\varphi)[E_{\bfI}(e_IN)]\md\mu_{h_{\refer}} \\
      & -2t^{-1}\textstyle{\sum}_{0<|\bfI|\leq \ell}\int_{\S}e_{0}(\varphi)E_{\bfI}(N)E_{\bfI}(e_0\varphi)\md\mu_{h_{\refer}}+\mR,
    \end{split}
  \end{equation}
  where
  \[
  |\mR|\leq C_\ell t^{-1+\ve}\etot{\ell}+C_\ell t^{-1}\etot{\ell-1}^{1/2}\etot{\ell}^{1/2}.
  \]
\end{lemma}
\begin{proof}
  Let us begin by considering the last term on the right hand side of (\ref{eq:high ord en phi}). We wish to differentiate this term with respect
  to time. When the derivative hits the factor $t^{-2(1-\ve)}$, we simply obtain $-2(1-\ve)t^{-1}$ times the last term. Let us therefore focus on
  \[
  2t^{-2+2\ve}\textstyle{\sum}_{|\bfI|\leq \ell}\int_{\S}(E_{\bfI}\varphi)[E_{\bfI}(Ne_0\varphi)]\md\mu_{h_{\refer}}.
  \]
  Since $Ne_0\varphi$ is a leading order term, (\ref{eq:est lead order terms}) and (\ref{eq:lapse estimate crude}) imply that the absolute value of this
  expression can be estimated by (\ref{eq:protot error term}). Next, consider
  \begin{equation}\label{eq:time der first term in phi en}
    2\textstyle{\int}_{\S}(E_{\bfI}e_0\varphi)[E_{\bfI}(Ne_0^2\varphi)]\md\mu_{h_{\refer}}+2\sum_{I}\int_{\S}(E_{\bfI}e_I\varphi)[E_{\bfI}(\d_t e_I\varphi)]\md\mu_{h_{\refer}}.
  \end{equation}
  Due to (\ref{eq:spatial scalar field derivative}), the second term can be written
  \begin{equation}\label{eq:dt eIphisq en contr}
    2\textstyle{\sum}_{I}\int_{\S}(E_{\bfI}e_I\varphi)[E_{\bfI}(e_I(Ne_0\varphi)-Nk_{IJ}e_J(\varphi)]\md\mu_{h_{\refer}}.
  \end{equation}
  Since $Nk_{IJ}e_J(\varphi)$ is a bad term, the contribution from the second term in the second factor in the integrand can be written
  \[
  -2\textstyle{\sum}_{I,J}\int_{\S}k_{IJ}(E_{\bfI}e_I\varphi)(E_{\bfI}e_J\varphi)\md\mu_{h_{\refer}}+\dots,
  \]
  where the dots represent terms that can be estimated in absolute value by (\ref{eq:protot error term vtwo}).  
  Next, consider the contribution from the first term in the second factor in the integrand of (\ref{eq:dt eIphisq en contr}). It can be written
  \begin{equation}\label{eq:dt eIphisq en contr st}
    2\textstyle{\sum}_{I}\int_{\S}(E_{\bfI}e_I\varphi)[E_{\bfI}(e_I(N)e_0\varphi)+E_{\bfI}(Ne_I(e_0\varphi))]\md\mu_{h_{\refer}}.
  \end{equation}
  Again, $e_I(N)e_0\varphi$ is a bad term, and the contribution from the first term in the second factor in the integrand is
  \[
  2\textstyle{\sum}_{I}\int_{\S}e_0(\varphi)(E_{\bfI}e_I\varphi)[E_{\bfI}(e_IN)]\md\mu_{h_{\refer}}+\dots,
  \]
  where the dots represent terms that can be estimated in absolute value by (\ref{eq:protot error term vtwo}).  
  The contribution from the second term in the second factor in the integrand of (\ref{eq:dt eIphisq en contr st}) can be written as a sum of
  \begin{equation}\label{eq:N commuted through phi en}
    2\textstyle{\sum}_{I}\int_{\S}(E_{\bfI}e_I\varphi)NE_{\bfI}[e_I(e_0\varphi)]\md\mu_{h_{\refer}}
  \end{equation}
  plus a linear combination of terms of the form
  \begin{equation}\label{eq:distr der eIphi EN eIezphi}
    2\textstyle{\sum}_{I}\int_{\S}(E_{\bfI}e_I\varphi)E_{\bfI_1}[E_{i_1}(N-1)]E_{\bfI_2}[e_I(e_0\varphi)]\md\mu_{h_{\refer}},
  \end{equation}
  where $|\bfI_1|+|\bfI_2|=|\bfI|-1$. By arguments similar to the above, expressions of the form (\ref{eq:distr der eIphi EN eIezphi}) can be estimated
  in absolute value by (\ref{eq:protot error term}). Turning to (\ref{eq:N commuted through phi en}), commuting $e_I$ and $E_{\bfI}$ leads to
  a term that can be estimated in absolute value by (\ref{eq:protot error term}); cf. (\ref{eq:EbfI eI comm est}). We are thus left with
  \[
  2\textstyle{\sum}_{I}\int_{\S}(E_{\bfI}e_I\varphi)Ne_I[E_{\bfI}(e_0\varphi)]\md\mu_{h_{\refer}}.
  \]
  Integrating by parts in this expression generates two terms. The term involving the divergence of $e_I$ can be estimated in absolute value by
  $C_\ell t^{-1+3\ve}\etot{\ell}$. The remaining term can be written
  \[
  -2\textstyle{\sum}_{I}\int_{\S}e_I[N(E_{\bfI}e_I\varphi)]E_{\bfI}(e_0\varphi)\md\mu_{h_{\refer}}.
  \]
  The term that results when $e_I$ hits $N$ can be estimated in absolute value by $Ct^{-1+4\ve}\etot{\ell}$. Commuting $e_I$ and $E_{\bfI}$ again
  leads to expressions that can be bounded in absolute value by $C_\ell t^{-1+\ve}\etot{\ell}$. What remains is then
  \[
  -2\textstyle{\sum}_{I}\int_{\S}N(E_{\bfI}e_Ie_I\varphi)E_{\bfI}(e_0\varphi)\md\mu_{h_{\refer}}.
  \]
  Next, $N$ can be commuted through $E_{\bfI}$ at the price of a term that can be estimated in absolute value by $C_\ell t^{-1+\ve}\etot{\ell}$.
  We are left with
  \[
  -2\textstyle{\sum}_{I}\int_{\S}E_{\bfI}(Ne_Ie_I\varphi)E_{\bfI}(e_0\varphi)\md\mu_{h_{\refer}}.
  \]
  At this point, we can combine the result with the first term appearing in (\ref{eq:time der first term in phi en}). This leads to
  \[
  2\textstyle{\int}_{\S}E_{\bfI}[N(e_0e_0\varphi-e_Ie_I\varphi)]E_{\bfI}(e_0\varphi)\md\mu_{h_{\refer}}.
  \]
  Next, we appeal to (\ref{eq:scalarfield}). As a result, we need to estimate
  \[
  2\textstyle{\int}_{\S}E_{\bfI}(-t^{-1}Ne_{0}(\varphi)+e_{I}(N)e_{I}(\varphi) -N\g_{JII}e_{J}(\varphi)-NV'\circ\varphi)E_{\bfI}(e_0\varphi)\md\mu_{h_{\refer}}.
  \]
  The three last terms inside the paranthesis in the first factor in the integrand are good terms, and their contribution is bounded in absolute value
  by $C_\ell t^{-1+\ve}\etot{\ell}$. The first term is a leading order term. Due to (\ref{eq:leading order terms}), (\ref{eq:leading order terms RB})
  and (\ref{eq:lapse estimate crude}), we conclude that 
  \begin{equation*}
    \begin{split}
      & 2\textstyle{\int}_{\S}E_{\bfI}[N(e_0e_0\varphi-e_Ie_I\varphi)]E_{\bfI}(e_0\varphi)\md\mu_{h_{\refer}}\\
      = & -2t^{-1}\textstyle{\int}_{\S}|E_{\bfI}(e_{0}\varphi)|^2\md\mu_{h_{\refer}}
      -2t^{-1}\textstyle{\int}_{\S}e_{0}(\varphi)E_{\bfI}(N)E_{\bfI}(e_0\varphi)\md\mu_{h_{\refer}}+\dots,
    \end{split}
  \end{equation*}
  where the dots represent terms that can be estimated in absolute value by (\ref{eq:protot error term vtwo}) and the second term can be
  dropped if $\bfI=0$. The lemma follows.    
\end{proof}

\subsection{The energy of $a_I$}\label{ssection:the energy of aI}
Next, we consider the evolution of the energy $E_{(a),\ell}$, introduced in (\ref{eq:Ealdef}). 

\begin{lemma}\label{lemma:a}
  Consider a solution to (\ref{eq: transport frame})--(\ref{eq:LapseEquation}) as described in Subsection~\ref{ssection:eqs}.
  Assume, in addition,  that $V$ is a $\s_V$-admissible potential and that the basic supremum assumptions, see
  Definition~\ref{def:bas sup ass}, are satisfied with $3\ve\leq \s_V$. Then, for each $\ell\in\nn{}$, there is a $t_\ell>0$ such that for $t\leq t_\ell$,
  \begin{equation*}
    \begin{split}
      \d_t E_{(a),\ell} = & -2t^{-1}\textstyle{\sum}_{|\bfI|\leq \ell}\sum_{I}\int_{\S}(E_{\bfI}a_I)\cdot E_{\bfI}[e_I(N)]\md\mu_{h_{\refer}}\\
      & +2\textstyle{\sum}_{|\bfI|\leq \ell}\sum_{I,J}\int_{\S}k_{IJ}(E_{\bfI}a_I)\cdot E_{\bfI}[e_J(N)]\md\mu_{h_{\refer}}\\
      & +2\textstyle{\sum}_{|\bfI|\leq \ell}\sum_{I,J,L}\int_{\S}k_{JL}(E_{\bfI}a_I)\cdot E_{\bfI}(\g_{JIL})\md\mu_{h_{\refer}}\\
      & +2\textstyle{\sum}_{|\bfI|\leq \ell}\sum_{I}\int_{\S}(e_0\varphi)(E_{\bfI}a_I)\cdot E_{\bfI}(e_I\varphi)\md\mu_{h_{\refer}}+\mR,
    \end{split}
  \end{equation*}
  where
  \begin{equation}\label{eq:mR est aen}
    |\mR|\leq C_\ell t^{-1+\ve}\etot{\ell}+C_\ell t^{-1}\etot{\ell-1}^{1/2}\etot{\ell}^{1/2}.
  \end{equation}
\end{lemma}
\begin{proof}
  Differentiating $E_{(a),\ell}$ with respect to time yields
  \[
    \d_t E_{(a),\ell}=2\textstyle{\sum}_{|\bfI|\leq \ell}\sum_{I}\int_{\S}(E_{\bfI}a_I)E_{\bfI}(\d_ta_I)\md\mu_{h_{\refer}}.
  \]
  The contribution from the first term on the right hand side of (\ref{eq:dt aI}) is
  \[
  -2t^{-1}\textstyle{\sum}_{|\bfI|\leq \ell}\sum_{I}\int_{\S}(E_{\bfI}a_I)\cdot E_{\bfI}[e_I(N)]\md\mu_{h_{\refer}}.
  \]
  The last three terms on the right hand side of (\ref{eq:dt aI}) are all bad terms. Up to a term which can be estimated in absolute
  value by (\ref{eq:protot error term vtwo}), the corresponding contributions can therefore be written
  \[
  2\textstyle{\sum}_{|\bfI|\leq \ell}\sum_{I}\int_{\S}(E_{\bfI}a_I)[k_{IJ} E_{\bfI}(e_J N)+k_{JL}E_{\bfI}(\gamma_{JIL})+e_0(\varphi)E_{\bfI}(e_I\varphi)]\md\mu_{h_{\refer}};
  \]
  cf. (\ref{eq:EbfI bad term}) and (\ref{eq:RB est}). The lemma follows.    
\end{proof}

\subsection{Summing up}\label{ssection:summing up}

Finally, we turn to $\metot{\ell}$, introduced in (\ref{eq:metotldef}). 
\begin{lemma}\label{lemma:metotl}
  Consider a solution to (\ref{eq: transport frame})--(\ref{eq:LapseEquation}) as described in Subsection~\ref{ssection:eqs}.
  Assume, in addition,  that $V$ is a $\s_V$-admissible potential and that the basic supremum assumptions, see
  Definition~\ref{def:bas sup ass}, are satisfied with $3\ve\leq \s_V$. Then, for each $\ell\in\nn{}$, there is a $t_\ell>0$ such that for $t\leq t_\ell$,
  \begin{equation}\label{eq:dtmetotlest}
    \d_t\metot{\ell}\geq -8t^{-1}\metot{\ell}-C_\ell t^{-1+\ve}\metot{\ell}-C_\ell t^{-1}\metot{\ell-1}^{1/2}\metot{\ell}^{1/2}.
  \end{equation}
\end{lemma}
\begin{proof}
  Combining Lemmas~\ref{lemma:dt Eeomega lb}, \ref{lemma:E gamma k est} and \ref{lemma:sf} with (\ref{eq:lapse ito kEk etc}) yields
  \begin{equation*}
    \begin{split}
      \d_t \etot{\ell} \geq & -2t^{-1}(1+\ve)E_{(e,\omega),\ell}-2t^{-1}\textstyle{\sum}_{|\bfI|\leq \ell}\int_{\S}|E_{\bfI}(e_{0}\varphi)|^2\md\mu_{h_{\refer}}\\
      & -2\textstyle{\sum}_{|\bfI|\leq \ell}\textstyle{\sum}_{I,J}\int_{\S}k_{IJ}(E_{\bfI}e_I\varphi)(E_{\bfI}e_J\varphi)\md\mu_{h_{\refer}}\\
      & -2t^{-1}(1-\ve)t^{-2+2\ve}\textstyle{\sum}_{|\bfI|\leq \ell}\int_{\S}|E_{\bfI}\varphi|^2\md\mu_{h_{\refer}}\\
      & -\textstyle{\sum}_{I,J,K,L}\sum_{|\bfI|\leq \ell}\int_{\S}k_{IL}E_{\bfI}(\gamma_{IJK})E_{\bfI}(\gamma_{LJK})\md\mu_{h_{\refer}}\\
      & -\textstyle{\sum}_{I,J,K,L}\sum_{|\bfI|\leq \ell}\int_{\S}k_{JL}E_{\bfI}(\gamma_{IJK})E_{\bfI}(\gamma_{ILK})\md\mu_{h_{\refer}}\\
      & +\textstyle{\sum}_{I,J,K,L}\sum_{|\bfI|\leq \ell}\int_{\S}k_{KL}E_{\bfI}(\gamma_{IJK})E_{\bfI}(\gamma_{IJL})\md\mu_{h_{\refer}}\\
      & +2\textstyle{\sum}_{I,J,L,K}\sum_{|\bfI|\leq \ell}\int_{\S}(k_{LJ}E_{\bfI}(\g_{LII}) + k_{IL}E_{\bfI}(\g_{IJL}))E_{\bfI}(\g_{JKK})\md\mu_{h_{\refer}}\\    
      & -2t^{-1}\textstyle{\sum}_{I,J}\sum_{|\bfI|\leq \ell}\int_{\S}|E_{\bfI}(k_{IJ})|^2\md\mu_{h_{\refer}}+t^{-1}E_{(N),l}\\
      & -2\textstyle{\sum}_{I,J,L}\sum_{|\bfI|\leq \ell}\int_{\S}k_{LJ}E_{\bfI}(\g_{LII})E_{\bfI}(e_{J} (N))\md\mu_{h_{\refer}}\\
      & +2\textstyle{\sum}_{J,K}\sum_{|\bfI|\leq \ell}\int_{\S}e_{0}(\varphi) E_{\bfI}[e_{J}(\varphi)]E_{\bfI}(\g_{JKK})\md\mu_{h_{\refer}}+\mR,
    \end{split}
  \end{equation*}
  where $\mR$ satisfies an estimate of the form (\ref{eq:mR est aen}). Before proceeding, note that for any $\a_I$, 
  \begin{equation}\label{eq:kIJ aI aJ}
    \big|\textstyle{\sum}_{I,J}k_{IJ}\alpha_I\alpha_J\big|\leq \big(\textstyle{\sum}_{I,J}k_{IJ}^2\big)^{1/2}
    \big(\textstyle{\sum}_{I,J}\alpha_I^2\alpha_J^2\big)^{1/2}\leq (1+Ct^{\ve})t^{-1}\textstyle{\sum}_I\a_I^2,
  \end{equation}
  where we appealed to (\ref{eq:Cz k ezphi bd}) in the last step. Keeping this estimate in mind, including $E_{(a),\ell}$ and recalling
  Lemma~\ref{lemma:a} then yields
  \begin{equation*}
    \begin{split}
      \d_t \metot{\ell}\geq & -2t^{-1}(1+\ve)E_{(e,\omega),\ell}-2t^{-1}E_{(\varphi),\ell}-2t^{-1}\textstyle{\sum}_{I,J}\sum_{|\bfI|\leq \ell}\int_{\S}|E_{\bfI}(k_{IJ})|^2\md\mu_{h_{\refer}}\\
      & -\textstyle{\sum}_{I,J,K,L}\sum_{|\bfI|\leq \ell}\int_{\S}k_{IL}E_{\bfI}(\gamma_{IJK})E_{\bfI}(\gamma_{LJK})\md\mu_{h_{\refer}}\\
      & -\textstyle{\sum}_{I,J,K,L}\sum_{|\bfI|\leq \ell}\int_{\S}k_{JL}E_{\bfI}(\gamma_{IJK})E_{\bfI}(\gamma_{ILK})\md\mu_{h_{\refer}}\\
      & +\textstyle{\sum}_{I,J,K,L}\sum_{|\bfI|\leq \ell}\int_{\S}k_{KL}E_{\bfI}(\gamma_{IJK})E_{\bfI}(\gamma_{IJL})\md\mu_{h_{\refer}}\\
      & +2\textstyle{\sum}_{J,L}\sum_{|\bfI|\leq \ell}\int_{\S}(k_{LJ}E_{\bfI}(a_L) + 2k_{IL}E_{\bfI}(\g_{IJL}))E_{\bfI}(a_J)\md\mu_{h_{\refer}}\\
      & +4\textstyle{\sum}_{J}\sum_{|\bfI|\leq \ell}\int_{\S}e_{0}(\varphi) E_{\bfI}[e_{J}(\varphi)]E_{\bfI}(a_J)\md\mu_{h_{\refer}}\\
      & -2t^{-1}\textstyle{\sum}_{|\bfI|\leq \ell}\sum_{I}\int_{\S}(E_{\bfI}a_I)\cdot E_{\bfI}[e_I(N)]\md\mu_{h_{\refer}}+t^{-1}E_{(N),\ell}+\mR,
    \end{split}
  \end{equation*}
  where $\mR$ satisfies an estimate of the form (\ref{eq:mR est aen}). Next, note that, due to (\ref{eq:kIJ aI aJ}),
  \begin{equation*}
    \begin{split}
      & -\textstyle{\sum}_{I,J,K,L}\sum_{|\bfI|\leq \ell}\int_{\S}k_{IL}E_{\bfI}(\gamma_{IJK})E_{\bfI}(\gamma_{LJK})\md\mu_{h_{\refer}}\\
      & -\textstyle{\sum}_{I,J,K,L}\sum_{|\bfI|\leq \ell}\int_{\S}k_{JL}E_{\bfI}(\gamma_{IJK})E_{\bfI}(\gamma_{ILK})\md\mu_{h_{\refer}}\\
      & +\textstyle{\sum}_{I,J,K,L}\sum_{|\bfI|\leq \ell}\int_{\S}k_{KL}E_{\bfI}(\gamma_{IJK})E_{\bfI}(\gamma_{IJL})\md\mu_{h_{\refer}}\\
      \geq & -\textstyle{\sum}_{I,J,K}\sum_{|\bfI|\leq \ell}\int_{\S}t^{-1}(3+Ct^{\ve})|E_{\bfI}(\gamma_{IJK})|^2\md\mu_{h_{\refer}}.
    \end{split}
  \end{equation*}
  Moreover,  
  \begin{equation*}
    \begin{split}
      & -2t^{-1}\textstyle{\sum}_{|\bfI|\leq \ell}\sum_{I}\int_{\S}(E_{\bfI}a_I)\cdot E_{\bfI}[e_I(N)]\md\mu_{h_{\refer}}+t^{-1}E_{(N),\ell}\\
      \geq & -t^{-1}\textstyle{\sum}_{|\bfI|\leq \ell}\sum_{I}\int_{\S}|E_{\bfI}a_I|^2\md\mu_{h_{\refer}}.
    \end{split}
  \end{equation*}
  We also have
  \begin{equation*}
    \begin{split}
      & 4\textstyle{\sum}_{J}\sum_{|\bfI|\leq \ell}\int_{\S}e_{0}(\varphi) E_{\bfI}[e_{J}(\varphi)]E_{\bfI}(a_J)\md\mu_{h_{\refer}}\\
      \geq & -t^{-1}\textstyle{\sum}_{J}\sum_{|\bfI|\leq \ell}\int_{\S}|E_{\bfI}[e_{J}(\varphi)]|^2\md\mu_{h_{\refer}}
      -4t^{-1}\textstyle{\sum}_{J}\sum_{|\bfI|\leq \ell}\int_{\S}|te_{0}(\varphi)|^2 |E_{\bfI}(a_J)|^2\md\mu_{h_{\refer}}. 
    \end{split}
  \end{equation*}
  Next,
  \begin{equation*}
    \begin{split}
      & 4\textstyle{\sum}_{I,J,L}\sum_{|\bfI|\leq \ell}\int_{\S}k_{IL}E_{\bfI}(\g_{IJL})E_{\bfI}(a_J)\md\mu_{h_{\refer}}\\
      \geq & -t^{-1}\textstyle{\sum}_{I,J,K}\sum_{|\bfI|\leq \ell}\int_{\S}|E_{\bfI}(\gamma_{IJK})|^2\md\mu_{h_{\refer}}\\
      & -4t^{-1}\textstyle{\sum}_{I,J,L}\sum_{|\bfI|\leq \ell}\int_{\S}t^2k_{IL}k_{IL}|E_{\bfI}(a_J)|^2\md\mu_{h_{\refer}}. 
    \end{split}
  \end{equation*}
  Summing up and recalling (\ref{eq:almost asymptotic Hcon gen}) and (\ref{eq:kIJ aI aJ}) yields
  \begin{equation*}
    \begin{split}
      \d_t \metot{\ell}\geq & -2(1+\ve)t^{-1}E_{(e,\omega),\ell}-3t^{-1}E_{(\varphi),\ell}-2t^{-1}\textstyle{\sum}_{I,J}\sum_{|\bfI|\leq \ell}\int_{\S}|E_{\bfI}(k_{IJ})|^2\md\mu_{h_{\refer}}\\
      & -4t^{-1}\textstyle{\sum}_{I,J,K}\sum_{|\bfI|\leq \ell}\int_{\S}|E_{\bfI}(\gamma_{IJK})|^2\md\mu_{h_{\refer}}\\
      & -7t^{-1}\textstyle{\sum}_{|\bfI|\leq \ell}\sum_{I}\int_{\S}|E_{\bfI}a_I|^2\md\mu_{h_{\refer}}+\mR.
    \end{split}
  \end{equation*}
  Since $3\ve\leq\sigma_V$ and $\sigma_V<1$, we conclude that (\ref{eq:dtmetotlest}) holds.   
\end{proof}
Due to the above estimates, we are in a position to deduce energy estimates for all orders.
\begin{cor}\label{cor:energy estimates}
  Consider a solution to (\ref{eq: transport frame})--(\ref{eq:LapseEquation}) as described in Subsection~\ref{ssection:eqs}.
  Assume, in addition,  that $V$ is a $\s_V$-admissible potential and that the basic supremum assumptions, see
  Definition~\ref{def:bas sup ass}, are satisfied with $3\ve\leq \s_V$. Then, for each $\ell\in\nn{}$, there is a $C_\ell$ such that, for $t\leq t_0$, 
  \begin{equation}\label{eq:energy estimates}
    \etot{\ell}(t)\leq C_\ell t^{-9}.
  \end{equation}
\end{cor}
\begin{proof}
  We prove a slightly improved version of (\ref{eq:energy estimates}) by induction. The inductive hypothesis is that there are
  constants $C_\ell$ and $b_\ell$ such that
  \begin{equation}\label{eq:induction}
    \metot{\ell}\leq C_\ell\ldr{\ln t}^{b_\ell}t^{-8}
  \end{equation}
  for $t\leq t_0$. That (\ref{eq:induction}) holds for $\ell\leq 1$ is an immediate consequence of (\ref{eq:estdynvarmfltot gen})
  and (\ref{eq:phi Co be gen}). 

  Given that (\ref{eq:induction}) holds for some $\ell$, consider (\ref{eq:dtmetotlest}) with $\ell$ replaced by $\ell+1$. Let $C_{\ell+1}$ be the constant appearing in the
  second term on the right hand side of (\ref{eq:dtmetotlest}) (with $\ell$ replaced by $\ell+1$) and let
  \begin{equation}\label{eq:rlpo def}
    r_{\ell+1}(t):=8\ln t+\ve^{-1}C_{\ell+1}t^{\ve}.
  \end{equation}
  Compute
  \[
    \d_t\Big(\exp(r_{\ell+1})\metot{\ell+1}\Big)\geq -C_{\ell+1}t^{-1}\exp(r_{\ell+1})\metot{\ell}^{1/2}\metot{\ell+1}^{1/2}.
  \]
  Introducing
  \begin{equation}\label{eq:Fl def}
    F_\ell:=\exp(r_{\ell+1})\metot{\ell+1}
  \end{equation}
  and noting that, by the inductive hypothesis, 
  \[
    \exp(r_{\ell+1}(t)/2)\metot{\ell}^{1/2}(t)\leq C_{\ell+1}\ldr{\ln t}^{b_\ell/2},
  \]
  we conclude that
  \[
    \d_tF_{\ell}^{1/2}\geq -C_{\ell+1}t^{-1}\ldr{\ln t}^{b_\ell/2}.
  \]
  Integrating this estimate from $t$ to $t_{\ell+1}$ yields
  \[
    F_{\ell}^{1/2}(t_{\ell+1})-F_{\ell}^{1/2}(t)\geq -C_{\ell+1}\ldr{\ln t}^{b_\ell/2+1}.
  \]
  Since the solution is smooth, $F_{\ell}^{1/2}(t_{\ell+1})$ is bounded by some constant, so that
  \[
    F_{\ell}^{1/2}(t)\leq C_{\ell+1}\ldr{\ln t}^{b_\ell/2+1}.
  \]
  for $t\leq t_{\ell+1}$. Recalling (\ref{eq:rlpo def}) and (\ref{eq:Fl def}), we conclude that
  \[
    \metot{\ell+1}(t)\leq C_{\ell+1}\ldr{\ln t}^{b_\ell+2}t^{-8}
  \]
  for $t\leq t_{\ell+1}$. Since this estimate also holds for $t_{\ell+1}\leq t\leq t_0$, due to the fact that the solution is smooth, we conclude that
  (\ref{eq:induction}) holds with $\ell$ replaced by $\ell+1$. By induction, we conclude that it holds for all $\ell$. Since $t\ldr{\ln t}^{b_\ell}$
  is bounded for $t\leq t_0$, the estimate (\ref{eq:energy estimates}) holds. 
\end{proof}

Finally, we are in a position to prove Theorem~\ref{thm:hod}.

\begin{proof}[Proof of Theorem~\ref{thm:hod}]
  Due to (\ref{eq:psiinterpol}), (\ref{eq:e om g phi be gen}) and (\ref{eq:energy estimates}), it follows that
  \[
    \|e\|_{H^\ell(\S_t)}\leq C_{\ell,k}t^{(-1+3\ve)(1-\ell/k)}t^{-9\ell/(2k)}.
  \]
  Taking a large enough $k$, depending only on $\ell$ and $\ve$, it follows that
  \[
    \|e\|_{H^\ell(\S_t)}\leq C_{\ell,\ve}t^{-1+2\ve}.
  \]
  Due to Sobolev embedding, the same holds for all $C^\ell$ norms. Due to this and similar arguments for $\omega$, $\gamma$ and $e_I\varphi$,
  we deduce from (\ref{eq:e om g phi be gen}) and (\ref{eq:energy estimates}) that (\ref{eq:eogearphi Cl est}) holds.
  Similarly, using (\ref{eq:e om g phi be gen}) and (\ref{eq:N be gen}), we can argue that (\ref{eq:lapse est interm}) holds.
  Next, due to (\ref{eq:psiinterpol}), (\ref{eq:k phi be gen}) and (\ref{eq:energy estimates}), it follows that
  \[
    t^{1+\ve/4}\|k\|_{C^{\ell}(\S_t)}+t^{1+\ve/4}\|e_0 \varphi\|_{C^{\ell}(\S_t)} \leq C_\ell
  \]
  for all $\ell\in\nn{}_0$ and all $t\leq t_0$. Moreover, due to (\ref{eq:psiinterpol}), (\ref{eq:phi Co be gen}) and (\ref{eq:energy estimates}), it follows that
  \[
    t^{\ve}\|\varphi\|_{C^{\ell}(\S_t)} \leq C_\ell
  \]
  for all $\ell\in\nn{}_0$ and all $t\leq t_0$. Combining this observation with (\ref{eq:V and V prime Hl est}) yields
  \[
    t^2\|V\circ\varphi\|_{C^\ell(\S_t)}+t^2\|V'\circ\varphi\|_{C^\ell(\S_t)}\leq C_\ell t^{4\ve}
  \]
  for all $\ell\in\nn{}_0$ and all $t\leq t_0$. On the other hand, due to (\ref{eq:dtkIJ}),
  \begin{equation*}
    \begin{split}
      \d_t(tk_{IJ}) = & -(N-1)k_{IJ}+te_{(I}e_{J)} (N) - te_{K} (N\g_{K(IJ)})- Nte_{(I}(\g_{J)KK}) +Nt\g_{KLL}\g_{K(IJ)}\\
      & + Nt\g_{I(KL)}\g_{J(KL)}- \tfrac14 Nt\g_{KLI} \g_{KLJ} + Nte_I(\varphi)e_J(\varphi)+\tfrac{2N}{n-1}t(V\circ\varphi) \delta_{IJ}.
    \end{split}
  \end{equation*}
  Combining this equality with the estimates derived so far yields the conclusion that
  \[
    \|\d_t(tk_{IJ})\|_{C^\ell(\S_t)}\leq C_\ell t^{-1+\ve/4}
  \]  
  for all $\ell\in\nn{}_0$ and all $t\leq t_0$. This estimate can be integrated to yield the existence of a $\mrmfK_{IJ}\in C^{\infty}(\S)$ such that
  (\ref{eq:tkIJ conv}) holds. Next, note that (\ref{eq:scalarfield}) implies that
  \[
  \d_t\left( te_{0}(\varphi) \right)=-(N-1)e_0(\varphi)+tN e_{I}\left( e_{I}(\varphi) \right)
  +te_{I}(N)e_{I}(\varphi) -tN\g_{JII}e_{J}(\varphi)-tNV'\circ\varphi.
  \]
  By an argument similar to the derivation of (\ref{eq:tkIJ conv}), we conclude that there is a $\mrPsi\in C^{\infty}(\S)$ such that
  \begin{equation}\label{eq:ezPhi conv}
    \|t (e_0\varphi)(\cdot,t)-\mrPsi\|_{C^\ell(\S)}+\|(t\d_t\varphi)(\cdot,t)-\mrPsi\|_{C^\ell(\S)}\leq C_\ell t^{\ve/4}
  \end{equation}
  for all $\ell\in\nn{}_0$ and all $t\leq t_0$; the bound on the second term on the left hand side follows from the bound of the first term,
  combined with (\ref{eq:lapse est interm}). The bound on the second term can be integrated to conclude that there is a
  $\mrPhi\in C^{\infty}(\S)$ such that
  \begin{equation}\label{eq:ezPhi and Phi conv}
    \|\varphi(\cdot,t)-\mrPsi\ln t-\mrPhi\|_{C^\ell(\S)}\leq C_\ell t^{\ve/4}
  \end{equation}
  for all $\ell\in\nn{}_0$ and all $t\leq t_0$. Combining (\ref{eq:ezPhi conv}) and (\ref{eq:ezPhi and Phi conv}) yields
  (\ref{eq:ezPhi and Phi conv final}). Finally, due to (\ref{seq:k phi prel est}), it is clear that (\ref{eq:k ezphi opt Cl bd}) holds. 
\end{proof}

\section{Asymptotics}\label{section:asymptotics}

The first step in the derivation of the asymptotics is to deduce basic consequences of the assumptions of Theorem~\ref{thm:asymptotics}.

\subsection{Basic consequences of the assumptions}\label{ssection: basic consequences of the assumptions}

We begin by deriving bounds for $\rp_I$. 

\begin{lemma} \label{bounds for the pI}
  Let $3\leq n\in\nn{}$, $k_0\in\nn{}$, $r>0$, $\delta\in (0,1)$ and fix a $\delta$-admissible potential $V\in C^\infty(\R)$.
  If $(M,g,\varphi)$ is a solution to the Einstein--nonlinear scalar field equations satisfying the $(\delta,n,k_0,r)$-assumptions, then
    \[
    |\rp_I| < 1-2\delta
    \]
    on $\S$ for all $I$.
\end{lemma}
\begin{remark}
  Here and below, $\S$ is assumed to be the closed manifold appearing in Definition~\ref{def:asympt assump}. 
\end{remark}
\begin{proof}
    The conclusion is an immediate consequence of the assumptions and the argument given in the proof of \cite[Lemma~77, p.~29]{OPR}.
\end{proof}

\begin{lemma} \label{estimates for powers of theta}
  Let $\alpha_I$, $I=1,\dots,n$, and $\beta$ be constants. Let $3\leq n\in\nn{}$, $k_0\in\nn{}$, $r>0$, $\delta\in (0,1)$ and fix a $\delta$-admissible
  potential $V\in C^\infty(\R)$. If $(M,g,\varphi)$ is a solution to the Einstein--nonlinear scalar field equations satisfying the
  $(\delta,n,k_0,r)$-assumptions, then, using the notation of Subsection~\ref{ssection:norms and basic est}, there is a constant $C$ such that
  \begin{equation}\label{eq:basic theta est}
    \big|E_\I\big(\theta^{\sum_I\alpha_Ip_I + \beta}\big)\big| \leq C\langle \ln t \rangle^{|\I|} t^{-\sum_I\alpha_I\rp_I - \beta}, \qquad
    \|\ln\theta\|_{C^{k_0}(\Sigma)} \leq C\langle \ln t \rangle
  \end{equation}
  on $M$ for $|\I| \leq k_0$; and for $t\in (0,T)$, respectively.
\end{lemma}
\begin{remark}
  Here and below, $T>0$ is assumed to be the positive number appearing in Definition~\ref{def:asympt assump}. 
\end{remark}
\begin{proof}
  We begin with the case without derivatives. Write
  \[
  \theta^{\sum_I\alpha_Ip_I + \beta} = N^{-\sum_I\alpha_Ip_I - \beta} (t\Theta)^{\sum_I\alpha_Ip_I+\beta} t^{\sum_I\alpha_I(\rp_I-p_I)} t^{-\sum_I \alpha_I\rp_I - \beta}.
  \]
  Note that $|(\ln t)(\rp_I-p_I)| \leq C|\ln t|t^\delta$ is bounded for $t \in (0,T)$. Moreover, due to (\ref{eq:tTheta and N invertible}),
  it is clear that $(t\Theta)^{-1}\leq 2$. Due to this observation and the assumptions,
  it is clear that $N$, $N^{-1}$, $t\Theta$ and $(t\Theta)^{-1}$ are bounded. Combining the above observations, it is clear that
  \[
  |\theta^{\sum_I\alpha_Ip_I+\beta}| \leq Ct^{-\sum_I\alpha_I\rp_I-\beta}.
  \]
  Next,
  \[
  |\ln \theta| = |\ln(t\Theta) - \ln N - \ln t|\leq C\ldr{\ln t}.
  \]
  Turning to the derivatives, if $1 \leq |\I| \leq k_0$,
  \[
  E_\I \ln\theta = E_\I(\ln(t\Theta) - \ln N-\ln t)= E_\I(\ln(t\Theta) - \ln N).
  \]
  Thus $|E_\I \ln \theta| \leq C$ and the second estimate in (\ref{eq:basic theta est}) follows. Finally,
  \[
  E_\I \big(\theta^{\sum_I\alpha_Ip_I + \beta}\big)
  = \theta^{\sum_I\alpha_Ip_I + \beta} \tsum E_{\I_1} \big( (\ln\theta)(\tsum_I\alpha_Ip_I+\beta) \big)
  \cdots E_{\I_r}\big( (\ln\theta)(\tsum_I\alpha_Ip_I+\beta) \big),
  \]
  from which the result follows.
\end{proof} 

\begin{lemma} \label{bounds for the scalar field}
  Let $3\leq n\in\nn{}$, $k_0\in\nn{}$, $r>0$, $\delta\in (0,1)$ and fix a $\delta$-admissible potential $V\in C^\infty(\R)$. If $(M,g,\varphi)$ is a
  solution to the Einstein--nonlinear scalar field equations satisfying the $(\delta,n,k_0,r)$-assumptions, then, using the notation of
  Subsection~\ref{ssection:norms and basic est}, there is a constant $C$ such that
  \begin{equation}\label{eq:basic vphi and Vvphi est}
    \|\varphi\|_{C^{k_0}(\Sigma)} \leq C\langle \ln t \rangle, \qquad \|V \circ \varphi\|_{C^{k_0}(\Sigma)} \leq C\langle \ln t \rangle^{k_0} t^{-2+2\delta}
  \end{equation}
  for all $t \in (0,T)$.
\end{lemma}

\begin{proof}
  The first estimate in (\ref{eq:basic vphi and Vvphi est}) is an immediate consequence of the assumptions, the second estimate in
  (\ref{eq:basic theta est}) and
  \[
  \varphi = -\rpsi\ln\theta + (\rpsi-\Psi)\ln\theta + \Phi.
  \]
  Next, since
  \[
  |\rpsi| = \big(1 - \tsum_I \rp_I^2\big)^{1/2} \leq 1
  \]
  and $(\rpsi - \Psi)\ln\theta + \Phi$ is bounded,
  \[
  |V^{(k)}\circ\varphi| \leq Ce^{2(1-\delta)|\varphi|} \leq Ce^{2(1-\delta)|\rpsi\ln\theta|} \leq Ct^{-2+2\delta}
  \]
  for $k \leq k_0$, where we appealed to $\ln\theta = \ln(t\Theta) - \ln N - \ln t$. Next, 
  \[
  |E_\I(V\circ\varphi)| = \big|\tsum (V^{(k)}\circ\varphi) E_{\I_1}(\varphi) \cdots E_{\I_k}(\varphi)\big|\leq C\langle \ln t \rangle^{|\I|}t^{-2+2\delta}
  \]
  for $|\I| \leq k_0$. The lemma follows. 
\end{proof}

The following lemma is necessary to obtain bounds that are uniform on $\Sigma$ in Theorem~\ref{asymptotics of the eigenframe} below.

\begin{lemma} \label{technical lemma}
  Let $3\leq n\in\nn{}$, $k_0\in\nn{}$, $r>0$, $\delta\in (0,1)$ and fix a $\delta$-admissible potential $V\in C^\infty(\R)$. If $(M,g,\varphi)$ is a
  solution to the Einstein--nonlinear scalar field equations satisfying the $(\delta,n,k_0,r)$-assumptions, then, using the notation of
  Subsection~\ref{ssection:norms and basic est}, there is a finite open cover $\{U_j\}$ of $\Sigma$, and constants
  $0 < \alpha_j \leq \delta/4$ such that
  \begin{subequations}
    \begin{align}
      |-1+\rp_I(x) + \rp_J(x) - \rp_K(x) + m\delta + \alpha_j| &\geq \tfrac{\alpha_j}{2},\label{eq:subdominant bd away from zero}\\
      |-1-\rp_I(x) + m\delta + \alpha_j| &\geq \tfrac{\alpha_j}{2}
    \end{align}
  \end{subequations}
  for all $x \in U_j$, $m\in\zn{}$ and $I,J,K\in\{1,\dots,n\}$ with $I\neq J$.
\end{lemma}

\begin{proof}
  Let $a_1,\ldots,a_k$ be real numbers and let $m_i\delta$ be the integer multiple of $\delta$ that is closest to $a_i$ (if there are two such numbers,
  either of them may be chosen). By rearranging the $a_i$, we may assume that the first $K$ of them are the ones satisfying $a_i = m_i\delta$. Then
  there is an $0 < \alpha \leq \delta/4$ such that
  \[
  |a_i - m_i\delta - \alpha| \geq \alpha
  \]
  for all $i$. Indeed, we can take $\alpha$ to be half the minimum of $|a_i - m_i\delta|$ for $i \geq K+1$. In fact, the inequality
  \[
  |a_i - m\delta - \alpha| \geq \alpha
  \]
  holds for every integer $m$.

    For every $x \in \Sigma$, we think of the $a_i$ as the numbers $1-\rp_I(x)-\rp_J(x) + \rp_K(x)$ for $I \neq J$, and the $1+\rp_I(x)$. Then by the previous paragraph, there is an $0 < \alpha(x) \leq \delta/4$ such that
    \[
    \begin{split}
        |1-\rp_I(x) - \rp_J(x) + \rp_K(x) - m\delta - \alpha(x)| &\geq \alpha(x),\\
        |1+\rp_I(x) - m\delta - \alpha(x)| &\geq \alpha(x).
    \end{split}
    \]
    By continuity, there is an open set $U_x \subset \Sigma$ containing $x$, such that
    \[
    \begin{split}
        |1-\rp_I(y) - \rp_J(y) + \rp_K(y) - m\delta - \alpha(x)| &\geq \tfrac{\alpha(x)}{2},\\
        |1+\rp_I(y) - m\delta - \alpha(x)| &\geq \tfrac{\alpha(x)}{2},
    \end{split}
    \]
    for all $y \in U_x$. Finally, since the $U_x$ form an open cover of $\Sigma$, there are finitely many of them that cover $\Sigma$.
\end{proof}

\subsection{The eigenframe}\label{ssection:the eigenframe}
Due to the assumptions, the eigenvalues $p_I$ of $\K$ are distinct everywhere. Due to \cite[Lemma~A.1, p.~223]{RinWave}, we can, by taking a finite
covering space of $\S$, if necessary, assume that there are smooth vector fields $\eig_I$ on $M$, for $I = 1,\ldots,n$, tangential
to the hypersurfaces $\S_t$ and such that $\K(\eig_I) = p_I\eig_I$ and $h(\eig_I,\eig_J) = \delta_{IJ}$. In particular, $\{\eig_I\}_{I=1}^n$ is then an
orthonormal frame for $h$ on each $\S_t$. Denote the corresponding dual frame $\{\deig^I\}_{I=1}^n$. Next, we deduce evolution equations for the $\eig_I$.

\begin{lemma} \label{equations for the eigenframe}
  Let $3\leq n\in\nn{}$, $k_0\in\nn{}$, $r>0$, $\delta\in (0,1)$ and fix a $\delta$-admissible potential $V\in C^\infty(\R)$. Given that
  $(M,g,\varphi)$ is a solution to the Einstein--nonlinear scalar field equations satisfying the $(\delta,n,k_0,r)$-assumptions, assume that there is
  an eigenframe $\{\eig_I\}_{I=1}^n$ with the properties listed above. Then $\{\eig_I\}_{I=1}^n$ and the eigenvalues $p_I$ satisfy the system of equations
  \begin{subequations}
    \begin{align}
      [e_0,\eig_I] &= \eig_I(\ln N)e_0 - \theta p_I\eig_I + \tsum_{J \neq I} \tfrac{1}{\theta(p_I-p_J)} h\big((\lie_{e_0}K)(\eig_I),\eig_J\big)\eig_J,\\
      e_0p_I &= -\theta^{-1}e_0(\theta)p_I + \theta^{-1}g\big((\lie_{e_0}K)(\eig_I),\eig_I\big)\label{equation for eigenvalues}
    \end{align}
  \end{subequations}
  (no summation on $I$). In particular,
  \[
    \p_t\eig_I^{\;i} = -\Theta p_I\eig_I^{\;i} + \tsum_{J \neq I} \tfrac{N^2}{\Theta(p_I-p_J)}g\big( (\lie_{e_0}K)(\eig_I),\eig_J \big)\eig_J^{\;i}
  \]
  (no summation on $I$), where $\eig_I = \eig_I^{\;i}E_i$.
\end{lemma}
\begin{remark}
  Recall from Subsection~\ref{ssection:conv geo}, that with respect to the frame $\{E_i\}_{i=1}^n$ and dual frame $\{\eta^i\}_{i=1}^n$,
  \[
    \lie_{e_0}K=N^{-1}\d_t(K^i_{\phantom{i}j})E_i\otimes \eta^j.
  \]
\end{remark}
\begin{proof}
  Compute
  \[
    0 = e_0\big( g(\eig_I,\eig_J) \big) = (\lie_{e_0}g)(\eig_I,\eig_J) + g\big([e_0,\eig_I],\eig_J\big) + g\big(\eig_I,[e_0,\eig_J]\big).
  \]
  Hence
  \[
    2\theta p_I\delta_{IJ} = -g\big([e_0,\eig_I],\eig_J\big) - g\big(\eig_I,[e_0,\eig_J]\big)
  \]
  (no summation on $I$ on the left hand side). In particular, it follows that
  \[
    g\big([e_0,\eig_I],\eig_I\big) = -\theta p_I
  \]
  (no summation on $I$ on the left hand side). For $X \in \mfx(M)$, denote by $\overline{X}$ the part of $X$ tangential to the $\Sigma_t$. Then, applying
  $\overline{[e_0,\,\cdot\,]}$ to both sides of $K(\eig_I) = \theta p_I \eig_I$ (no summation) yields
  \[
    (\lie_{e_0}K)(\eig_I) + K\big(\overline{[e_0,\eig_I]}\big) = e_0(\theta p_I)\eig_I + \theta p_I \overline{[e_0,\eig_I]}
  \]
  (no summation). Note that $K\big(\overline{[e_0,\eig_I]}\big) = \sum_J g\big([e_0,\eig_I],\eig_J\big)\theta p_J \eig_J$. The equation for $p_I$ follows by
  applying $g(\eig_I,\,\cdot\,)$ to the resulting equation. On the other hand, applying $g(\eig_J,\,\cdot\,)$ with $J \neq I$, yields
  \[
    g\big([e_0,\eig_I],\eig_J\big) = \tfrac{1}{\theta(p_I - p_J)}g\big((\lie_{e_0}K)(\eig_I),\eig_J\big).
  \]
  Next, 
  \[
    -g\big([e_0,\eig_I],e_0\big) = g\big([\eig_I,N^{-1}\p_t],e_0\big) = g\big(-N^{-2} \eig_I(N)\p_t,e_0\big)=g\big(-\eig_I(\ln N)e_0,e_0\big) = \eig_I(\ln N).
  \]
  Finally, we can write
  \[
    [e_0,\eig_I] = -g\big([e_0,\eig_I],e_0\big)e_0 + \tsum_J g\big([e_0,\eig_I],\eig_J\big)\eig_J.
  \]
  The lemma follows.
\end{proof}

The idea is to translate estimates for the Fermi-Walker frame $\{e_I\}_{I=1}^n$, its dual frame $\{\omega^I\}_{I=1}^n$ and its structure coefficients,
to the eigenframe $\{\eig_I\}_{I=1}^n$. To this end, we consider the family of matrices $A$ whose elements are defined by 
\[
\eig_I = A_I^Je_J.
\]
Note that $A$ is a family of orthogonal matrices, so that $e_I = A_J^I\eig_J$. For the dual frames, $\deig^I = A_I^J \omega^J$. Moreover,
$A_I^J = g(\eig_I,e_J)$, so that $A_I^J$ is smooth, and the bound $|A_I^J| \leq 1$ holds everywhere on $M$. Therefore, in order to translate estimates
for the Fermi-Walker frame, we only need to control the derivatives of $A$. To this end, we deduce an evolution equation for $A$.

\begin{lemma}
  Let $3\leq n\in\nn{}$, $k_0\in\nn{}$, $r>0$, $\delta\in (0,1)$ and fix a $\delta$-admissible potential $V\in C^\infty(\R)$. Given that
  $(M,g,\varphi)$ is a solution to the Einstein--nonlinear scalar field equations satisfying the $(\delta,n,k_0,r)$-assumptions, assume that there is
  an eigenframe $\{\eig_I\}_{I=1}^n$ with the properties listed at the beginning of the present subsection. Then the elements of $A$ satisfy the evolution
  equation
  \begin{equation}\label{eq:AIJ dot}
    \p_tA_I^J = \tsum_{K \neq I} \tfrac{N^2}{\Theta(p_I-p_K)}g\big((\lie_{e_0}K)(\eig_I),\eig_K\big)A_K^J.
  \end{equation}
\end{lemma}
\begin{proof}
    First, note that
    \begin{equation}\label{eq:eigIi dot}
      \p_t\eig_I^{\;i} = \big(\p_tA_I^J\big)e_J^i + A_I^J\p_te_J^i.
    \end{equation}
    From \eqref{fermi walker equation}, it follows that
    \[
    \p_te_I^i = -Nk_{IJ}e_J^i.
    \]
    Therefore,
    \begin{equation}\label{eq:AJI eJi dot}
      \begin{split}
        A_I^J\p_te_J^i = & -NA_I^Jk_{JK}e_K^i = -NA_M^Kk(\eig_I,\eig_M)e_K^i\\
        = & -NA_I^K\theta p_Ie_K^i = -N\theta p_I\eig_I^{\;i}
      \end{split}
    \end{equation}    
    (no summation on $I$). We can then conclude from Lemma~\ref{equations for the eigenframe} that
    \[
    \big(\p_tA_I^J\big)e_J = \tsum_{K \neq I} \tfrac{N^2}{\Theta(p_I-p_K)}g\big( \lie_{e_0}K(\eig_I),\eig_K \big)\eig_K.
    \]
    Since $\eig_K = A_K^Je_J$, the result follows.
\end{proof}

Since we only have information about the Fermi-Walker frame, we rewrite the evolution equation for $A_I^J$ as follows. Define $B_{IKLM}$ by
\[
B_{IKLM} := \tfrac{N^2}{\Theta(p_I-p_K)}\big( -\sric(e_L,e_M) + N^{-1}(\ovn^2N)(e_L,e_M) + e_L(\varphi)e_M(\varphi) \big),
\]
for $I \neq K$, and $B_{IILM} = 0$. Then, \eqref{evolution equation for the second fundamental form} implies that $A_I^J$ satisfies the equation
\begin{equation} \label{equation for A}
    \p_tA_I^J = \tsum_K B_{IKLM}A_I^LA_K^MA_K^J.
\end{equation}
It is thus of interest to deduce estimates for $B_{IKLM}$.

\begin{lemma} \label{estimate for B}
  Let $3\leq n\in\nn{}$, $k_0\in\nn{}$, $r>0$, $\delta\in (0,1)$ and fix a $\delta$-admissible potential $V\in C^\infty(\R)$. Given that
  $(M,g,\varphi)$ is a solution to the Einstein--nonlinear scalar field equations satisfying the $(\delta,n,k_0,r)$-assumptions, assume that there is
  an eigenframe $\{\eig_I\}_{I=1}^n$ with the properties listed at the beginning of the present subsection. Then there is a constant $C$ such that
  \begin{equation}\label{eq:Ric and ovn sq N est}
    \|\sric(e_I,e_J)\|_{C^{k_0-1}(\Sigma)} \leq Ct^{-2+6\delta}, \qquad \|N^{-1}(\ovn^2N)(e_I,e_J)\|_{C^{k_0-1}(\Sigma)} \leq Ct^{-2+2\delta}
  \end{equation}
  for $t<T$. Moreover, for $t<T$, 
  \[
  \begin{split}
    \|B_{IKLM}\|_{C^{k_0-1}(\Sigma)} \leq Ct^{-1+2\delta}.
  \end{split}
  \]
\end{lemma}

\begin{proof}
  We begin with $\sric(e_I,e_J)$. Due to Lemma~\ref{lemma: spatial curvatures}, 
  \[
  \sric(e_I,e_J) = e_K\gamma_{K(IJ)} + e_{(I}\gamma_{J)KK} - \gamma_{I(KL)}\gamma_{J(KL)} + \tfrac{1}{4}\gamma_{KLI}\gamma_{KLJ} - \gamma_{KLL}\gamma_{K(IJ)}.
  \]
  Thus the first estimate in (\ref{eq:Ric and ovn sq N est}) follows from \eqref{bounds on the solution}. Next, note that
  \[
  N^{-1}(\ovn^2N)(e_I,e_J) = e_Ie_J(\ln N) + e_I(\ln N)e_J(\ln N) - \Gamma_{IJK}e_K(\ln N),
  \]
  where $\Gamma_{IJK} = g(\n_{e_I}e_J,e_K)$ are the connection coefficients of the Fermi-Walker frame. Since
  $\Gamma_{IJK} = \frac{1}{2}(\gamma_{IJK} + \gamma_{KIJ} - \gamma_{JKI})$ due to the Koszul formula, we conclude that
  the second estimate in (\ref{eq:Ric and ovn sq N est}) holds. The estimate for $B_{IKLM}$ now follows by its definition,
  and the estimates for the $p_I$, $N$, $\Theta$ and $e_I(\varphi)$,
  including the estimate (\ref{eq:tTheta and N invertible}).
\end{proof}

\begin{prop} \label{bounds for change of basis matrix}
  Let $3\leq n\in\nn{}$, $k_0\in\nn{}$, $r>0$, $\delta\in (0,1)$ and fix a $\delta$-admissible potential $V\in C^\infty(\R)$. Given that
  $(M,g,\varphi)$ is a solution to the Einstein--nonlinear scalar field equations satisfying the $(\delta,n,k_0,r)$-assumptions, assume that there is
  an eigenframe $\{\eig_I\}_{I=1}^n$ with the properties listed at the beginning of the present subsection. Given $t_0\in(0,T)$, there is a constant
  $C$ and a $C^{k_0-1}(\Sigma)$ family $\mathring{A}$ of orthogonal matrices such that, for $t\leq t_0$, 
  \begin{equation}\label{eq:A minus Aring}
    \|A_I^J(t) - \mathring{A}_I^J\|_{C^{k_0-1}(\Sigma)} \leq Ct^{2\delta}.
  \end{equation}
\end{prop}
\begin{remark}
  When we speak of $t_0$ in what follows, we take it for granted to be the number fixed in this proposition. 
\end{remark}
\begin{proof}
  For $m <k_0$, define the functions
  \[
  \rho_m := \tsum_{I,J} \tsum_{|\I| = m} \big(E_\I A_I^J\big)^2.
  \]
  Since $|A_I^J| \leq 1$ everywhere, it follows that $\rho_0 \leq n^2$. We derive uniform bounds for the remaining $\rho_m$ inductively.
  Fix $m<k_0$. Suppose that there is a constant $C$, such that $\rho_k \leq C$ for all $k \leq m-1$. Then, by Lemma~\ref{estimate for B} and
  \eqref{equation for A}, 
  \[
  \begin{split}
    |\p_t\rho_{m}| &= 2\tsum_{I,J} \tsum_{|\I| = m}\big(E_\I A_I^J\big)\big(E_\I\p_t A_I^J\big) \leq 2\rho_{m}^{1/2}
    \big( \tsum_{I,J} \tsum_{|\I| = m} \big(E_\I\p_tA_I^J\big)^2 \big)^{1/2}\\
    &\leq 2\rho_{m}^{1/2}\big( Ct^{-1+2\delta}\rho_{m}^{1/2} + Ct^{-1+2\delta} \big)
    \leq Ct^{-1+2\delta}\rho_{m} + Ct^{-1+2\delta}\rho_{m}^{1/2}\\
    &\leq Ct^{-1+2\delta}\rho_{m} + Ct^{-1+2\delta},
  \end{split}
  \]
  where we used Young's inequality. Since $t^{-1+2\delta}$ is integrable on $(0,t_0]$, it follows that
  \[
  \rho_{m}(s) = \rho_{m}(t_0) + \tint_s^{t_0} - \p_t\rho_{m}(t)\md t \leq C + C\tint_s^{t_0}t^{-1+2\delta}\rho_{m}(t)\md t
  \]
  for all $s \in (0,t_0]$. By Grönwall's lemma, we conclude that $\rho_{m}(t) \leq C$ for some constant $C$ and all $t \in (0,t_0]$.

  By the above, $\|A_I^J(t)\|_{C^{k_0-1}(\Sigma)} \leq C$ for all $t \in (0,t_0]$. To find the matrix $\mathring{A}_I^J$, note that for $u < s$ and
  $|\I| \leq k_0-1$,
  \[
  \big|E_\I A_I^J(s) - E_\I A_I^J(u)\big| \leq \tint_u^s \big|E_\I\p_t A_I^J(t)\big|\md t \leq C(s^{2\delta} - u^{2\delta}).
  \]
  Hence, there are functions $\mathring{A}_I^J \in C^{k_0-1}(\Sigma)$ such that (\ref{eq:A minus Aring}) holds. Finally, $\mathring{A}_I^J$ is
  orthogonal, since it is the pointwise limit of $A_I^J$, which is orthogonal for all $t \in (0,T)$.
\end{proof}

Now that we know $A_I^J$ to be bounded in $C^{k_0-1}(\Sigma)$, we can obtain estimates in $C^{k_0-1}(\Sigma)$ for the $\eig_I$, the $\deig^I$, and the corresponding structure coefficients.

\begin{lemma}\label{lemma:eom hat est}
  Let $3\leq n\in\nn{}$, $k_0\in\nn{}$, $r>0$, $\delta\in (0,1)$ and fix a $\delta$-admissible potential $V\in C^\infty(\R)$. Given that
  $(M,g,\varphi)$ is a solution to the Einstein--nonlinear scalar field equations satisfying the $(\delta,n,k_0,r)$-assumptions, assume that there is
  an eigenframe $\{\eig_I\}_{I=1}^n$ with the properties listed at the beginning of the present subsection. Let
  $\lambda_{IJK} := g\big([\eig_I,\eig_J],\eig_K\big)$. Then there is a constant $C$ such that, for all $t \in (0,t_0]$,
  \[
  \|\eig_I^{\;i}\|_{C^{k_0-1}(\Sigma)} + \|\deig_i^I\|_{C^{k_0-1}(\Sigma)} + \|\lambda_{IJK}\|_{C^{k_0-2}(\Sigma)} \leq Ct^{-1+3\delta}.
  \]
\end{lemma}

\begin{proof}
    The estimates for $\eig_I$ and $\deig^I$ follow immediately from Proposition~\ref{bounds for change of basis matrix}, since $\eig_I = A_I^Je_J$ and $\deig^I = A_I^J\omega^J$. For $\lambda_{IJK}$,
    \[
    [\eig_I,\eig_J] = \big[A_I^Le_L,A_J^Me_M\big] = A_I^Le_L\big(A_J^M\big)e_M - A_J^Me_M\big(A_I^L\big)e_L + A_J^MA_I^L[e_L,e_M].
    \]
    Hence,
    \[
    \lambda_{IJK} = A_I^Le_L\big(A_J^M\big)A_K^M - A_J^Me_M\big(A_I^L\big)A_K^L + A_J^MA_I^LA_K^N\gamma_{LMN}.
    \]
    The estimates now follow from Proposition~\ref{bounds for change of basis matrix}.
\end{proof}

\subsection{The expansion normalized eigenframe} \label{ssection: the expansion normalized eigenframe}

We introduce an expansion normalized eigenframe $\{\ce_I\}_{I=1}^n$, with dual frame $\{\co^I\}_{I=1}^n$, defined by
\[
\ce_I := \theta^{-p_I}\eig_I, \qquad \co^I := \theta^{p_I}\deig^I
\]
(no summation). Denote by $\Lambda_{IJ}^K$ the corresponding structure coefficients, given by
\[
\Lambda_{IJ}^K := \co^K\big([\ce_I,\ce_J]\big).
\]

\begin{lemma} \label{initial estimates for expansion normalized variables}
  Let $3\leq n\in\nn{}$, $k_0\in\nn{}$, $r>0$, $\delta\in (0,1)$ and fix a $\delta$-admissible potential $V\in C^\infty(\R)$. Given that
  $(M,g,\varphi)$ is a solution to the Einstein--nonlinear scalar field equations satisfying the $(\delta,n,k_0,r)$-assumptions, assume that there is
  an eigenframe $\{\eig_I\}_{I=1}^n$ with the properties listed at the beginning of Subsection~\ref{ssection:the eigenframe}. Then there is a constant
  $C$ such that 
  \[
  \big|E_\I \big(\ce_I^i\big)\big| \leq Ct^{-1+\rp_I+2\delta}, \quad \big|E_\I \big(\co_i^I\big)\big|
  \leq Ct^{-1-\rp_I+2\delta}, \quad \big|E_\J\big(\Lambda_{IJ}^K\big)\big| \leq Ct^{-1+\rp_I+\rp_J-\rp_K + 2\delta}
  \]
  hold for $t\leq t_0$, $|\I| \leq k_0-1$ and $|\J| \leq k_0-2$.
\end{lemma}

\begin{proof}
  The estimates for $\ce_I^i$ and $\co_i^I$ follow directly from the definitions, Lemma~\ref{estimates for powers of theta} and
  Lemma~\ref{lemma:eom hat est}, since $\langle \ln t \rangle^{k_0-1}t^\delta$ is bounded for $t \in (0,t_0]$. For $\Lambda_{IJ}^K$, 
  \[
  \begin{split}
    [\ce_I,\ce_J] &= [\theta^{-p_I}\eig_I, \theta^{-p_J}\eig_J]\\
    &= -\theta^{-p_I}\eig_I(p_J\ln\theta)\ce_J + \theta^{-p_J}\eig_J(p_I\ln\theta)\ce_I + \tsum_K \theta^{p_K-p_I-p_J} \lambda_{IJK}\ce_K.
  \end{split}
  \]
  Hence,
  \[
  \Lambda_{IJ}^K = \theta^{p_K-p_I-p_J} \big( -\eig_I(p_J\ln\theta)\delta_J^K + \eig_J(p_I\ln\theta)\delta_I^K + \lambda_{IJK} \big).
  \]
  The last estimate now follows from Lemmas~\ref{estimates for powers of theta} and \ref{lemma:eom hat est}.
\end{proof}

Next, we deduce evolution equations for these objects.

\begin{lemma} \label{equations for the expansion normalized eigenframe lemma}
  Let $3\leq n\in\nn{}$, $k_0\in\nn{}$, $r>0$, $\delta\in (0,1)$ and fix a $\delta$-admissible potential $V\in C^\infty(\R)$. Given that
  $(M,g,\varphi)$ is a solution to the Einstein--nonlinear scalar field equations satisfying the $(\delta,n,k_0,r)$-assumptions, assume that there is
  an eigenframe $\{\eig_I\}_{I=1}^n$ with the properties listed at the beginning of Subsection~\ref{ssection:the eigenframe}. Then $\ce_I$, $\co^I$ and
  $\Lambda_{IJ}^K$ satisfy the following system of equations
  \begin{subequations}
    \begin{align}
      \p_t\ce_I^i &= \cb_I^J\ce_J^i,\label{equation for expansion normalized frame}\\
      \p_t \co_i^I &= -\cb_J^I \co_i^J,\label{equation for expansion normalized dual frame}\\
      \p_t\Lambda_{IJ}^K &= -\cb_K^K\Lambda_{IJ}^K + \cb_I^L\Lambda_{LJ}^K - \cb_J^L\Lambda_{LI}^K + \ce_I\big(\cb_J^K\big) - \ce_J\big(\cb_I^K\big)
      + \tsum_{L \neq K} \theta^{2(p_K-p_L)} \cb_K^L\Lambda_{IJ}^L,\label{equation for expansion normalized structure coefficients}
    \end{align}
  \end{subequations}
  where $\cb_I^J$ is given by
  \begin{align*}
    \cb_I^I := -(\ln\theta)\p_t p_I - \Theta^{-1}N^2p_I\big( -\roScal_h + N^{-1}\Delta_h N + |\md\varphi|_h^2 + \tfrac{2n}{n-1}V\circ\varphi \big),
  \end{align*}
  (no summation) and
  \[
  \cb_I^J := \tfrac{N^2}{\Theta(p_I-p_J)} (\lie_{e_0}K)\big(\ce_I,\co^J\big),
  \]
  for $I \neq J$.
\end{lemma}

\begin{proof}
  Due to (\ref{eq:AIJ dot}), (\ref{eq:eigIi dot}) and (\ref{eq:AJI eJi dot}), 
  \[
  \begin{split}
    \p_t\ce_I^i &= \p_t( \theta^{-p_I}\eig_I^{\;i} )= -\p_t(p_I\ln\theta)\ce_I^i + \theta^{-p_I}\p_t\eig_I^{\;i}\\
    &= -\p_t(p_I\ln\theta)\ce_I^i + \theta^{-p_I}\big( -\Theta p_I\eig_I^{\;i}
    + \tsum_{J \neq I} \tfrac{N^2}{\Theta(p_I-p_J)} g\big((\lie_{e_0}K)(\eig_I),\eig_J\big)\eig_J^{\;i} \big)\\
    &= -\big( \p_t(p_I\ln\theta) + \Theta p_I \big)\ce_I^i + \tsum_{J\neq I} \tfrac{N^2}{\Theta(p_I-p_J)} (\lie_{e_0}K)\big( \ce_I,\co^J \big)\ce_J^i.
  \end{split}
  \]
  The result for $\ce_I^i$ now follows by computing $\p_t\theta$ using \eqref{evolution equation for the second fundamental form}.
  Since $\delta^I_J=\co^I(\ce_J)$, (\ref{equation for expansion normalized dual frame}) follows from
  (\ref{equation for expansion normalized frame}). Finally, for $\Lambda_{IJ}^K$, note that
  \[
  \Lambda_{IJ}^K = \theta^{2p_K} g\big([\ce_I,\ce_J],\ce_K\big).
  \]
  Hence, appealing to the Jacobi identity, 
  \[
  \begin{split}
    \p_t\Lambda_{IJ}^K &= \p_t(\theta^{2p_K})g\big([\ce_I,\ce_J],\ce_K\big) + N\theta^{2p_K} \Big( 2k\big([\ce_I,\ce_J],\ce_K\big)\\
    &\quad + g\big([[e_0,\ce_I],\ce_J],\ce_K\big) + g\big([\ce_I,[e_0,\ce_J]],\ce_K\big) + g\big([\ce_I,\ce_J],[e_0,\ce_K]\big) \Big).
  \end{split}
  \]
  Note that 
  \[
  [e_0,\ce_I] = N^{-1}[\p_t,\ce_I] + \ce_I(\ln N)e_0 = N^{-1}\cb_I^L\ce_L + \ce_I(\ln N)e_0,
  \]
  so that 
  \[
  \begin{split}
    [[e_0,\ce_I],\ce_J] &= N^{-1}\cb_I^L[\ce_L,\ce_J] + N^{-1}\ce_J(\ln N)\cb_I^L\ce_L - N^{-1}\ce_J(\cb_I^L)\ce_L\\
    &\quad + \ce_I(\ln N)[e_0,\ce_J] - \ce_J\ce_I(\ln N)e_0.
  \end{split}
  \]
  Going back to $\p_t\Lambda_{IJ}^K$, our previous observations yield (no summation on $K$)
  \begin{align*}
    \p_t\Lambda_{IJ}^K &= 2\p_t(p_K\ln\theta)\Lambda_{IJ}^K + 2\Theta p_K \Lambda_{IJ}^K\\
    &\quad + \theta^{2p_K}\Big( \cb_I^Lg\big([\ce_L,\ce_J],\ce_K\big) + \ce_J(\ln N)\cb_I^Lg(\ce_L,\ce_K) - \ce_J\big(\cb_I^L\big)g(\ce_L,\ce_K)\\
    &\quad + N\ce_I(\ln N)g\big([e_0,\ce_J],\ce_K\big) - \cb_J^Lg\big([\ce_L,\ce_I],\ce_K\big) - \ce_I(\ln N)\cb_J^Lg(\ce_L,\ce_K)\\
    &\quad + \ce_I\big(\cb_J^L\big)g(\ce_L,\ce_K) - N\ce_J(\ln N)g\big([e_0,\ce_I],\ce_K\big) + \cb_K^Lg\big([\ce_I,\ce_J],\ce_L\big) \Big)\\
    &= -2\cb_K^K\Lambda_{IJ}^K + \cb_I^L\Lambda_{LJ}^K - \ce_J\big(\cb_I^K\big) - \cb_J^L\Lambda_{LI}^K + \ce_I\big(\cb_J^K\big)
    + \tsum_L \theta^{2(p_K-p_L)} \cb_K^L\Lambda_{IJ}^L,
  \end{align*}
  where we used that $-\cb_K^K = \p_t(p_K\ln \theta) + \Theta p_K$ (no summation). The lemma follows.
\end{proof}

\subsection{Asymptotics} \label{ssection: asymptotics}

In order to use the equations of Lemma~\ref{equations for the expansion normalized eigenframe lemma}, we need an expression for the off-diagonal
components of the spatial Ricci tensor in terms of $\ce_I$ and $\Lambda_{IJ}^K$.

\begin{lemma} \label{spatial ricci}
  Let $3\leq n\in\nn{}$, $k_0\in\nn{}$, $r>0$, $\delta\in (0,1)$ and fix a $\delta$-admissible potential $V\in C^\infty(\R)$. Given that
  $(M,g,\varphi)$ is a solution to the Einstein--nonlinear scalar field equations satisfying the $(\delta,n,k_0,r)$-assumptions, assume that there is
  an eigenframe $\{\eig_I\}_{I=1}^n$ with the properties listed at the beginning of Subsection~\ref{ssection:the eigenframe}. For $I \neq J$, 
  \[
  \sric^{\sharp}\big(\ce_I,\co^J\big) = \cR_{1,I}{}^J + \cR_{2,I}{}^J + \cR_{3,I}{}^J,
  \]
  where $\cR_{i,I}{}^J$ are given schematically by
  \begin{align*}
    \cR_{1,I}{}^J &= \theta^{2p_J} * \ce_I * \ce_J,\\
    \cR_{2,I}{}^J &= \theta^{2p_J} * \ce_I * \Lambda_{JK}^K + \theta^{2p_J} * \ce_K * \Lambda_{IJ}^K + \tsum_K \theta^{2(p_J+p_K-p_I)} * \ce_K * \Lambda_{KJ}^I\\
    &\quad + \tsum_K \theta^{2p_K} * \ce_K * \Lambda_{IK}^J + \theta^{2p_J} * \ce_J * \Lambda_{IK}^K,\\
    \cR_{3,I}{}^J &= \tsum_{K,L} \theta^{2(p_K+p_L-p_I)} * \Lambda_{KL}^I * \Lambda_{KL}^J + \theta^{2p_J} * \Lambda_{IK}^L * \Lambda_{LJ}^K
    + \tsum_K \theta^{2p_K} * \Lambda_{IK}^L * \Lambda_{LK}^J\\
    &\quad + \tsum_{K,L} \theta^{2(p_J+p_K-p_L)} * \Lambda_{IK}^L * \Lambda_{JK}^L + \tsum_L \theta^{2p_L} * \Lambda_{LK}^K * \Lambda_{IL}^J
    + \theta^{2p_J} * \Lambda_{LK}^K * \Lambda_{IJ}^L\\
    &\quad + \tsum_L \theta^{2(p_L+p_J-p_I)} * \Lambda_{LK}^K * \Lambda_{LJ}^I + \tsum_K \theta^{2(p_J+p_K-p_I)} * \Lambda_{KL}^I * \Lambda_{KJ}^L.
  \end{align*}
  Here, each type of term may contain up to two factors of the form $E_\I(p_I\ln\theta)$, with $|\I| \leq 2$. Moreover, each factor $\ce_I$
  and $\Lambda_{IJ}^K$ may have at most one $E_i$ derivative applied to it.
\end{lemma}

\begin{proof}
  Define the connection coefficients
  \[
  \Gamma_{IJ}^K := \co^K\big(\ovn_{\ce_I}\ce_J\big).
  \]
  Then, by the Koszul formula,
  \[
  \begin{split}
    2\theta^{-2p_K} \Gamma_{IJ}^K &= 2h\big(\ovn_{\ce_I}\ce_J,\ce_K\big)\\
    &= \ce_I(\theta^{-2p_J}\delta_{JK}) + \ce_J(\theta^{-2p_I}\delta_{IK}) - \ce_K(\theta^{-2p_I}\delta_{IJ})\\
    &\quad + h\big(\ce_K,[\ce_I,\ce_J]\big) - h\big(\ce_J,[\ce_I,\ce_K]\big) - h\big(\ce_I,[\ce_J,\ce_K]\big)
  \end{split}
  \]
  (no summation on any index). Hence,
  \begin{equation}\label{eq:Koszul ech}
    \begin{split}
      \Gamma_{IJ}^K = & -\ce_I(p_J\ln\theta)\delta_{JK} - \ce_J(p_I\ln\theta)\delta_{IK} + \theta^{2(p_K-p_I)} \ce_K(p_I\ln\theta) \delta_{IJ}\\
      & + \tfrac{1}{2}\big( \Lambda_{IJ}^K - \theta^{2(p_K-p_J)} \Lambda_{IK}^J - \theta^{2(p_K-p_I)} \Lambda_{JK}^I \big).
    \end{split}
  \end{equation}
  Next, 
  \[
  \begin{split}
    \sric^{\sharp}(\ce_I) &= \tsum_K \theta^{2p_K} \bar R(\ce_I,\ce_K)\ce_K\\
    &= \tsum_K \theta^{2p_K} \big( \ovn_{\ce_I} \ovn_{\ce_K} \ce_K - \ovn_{\ce_K} \ovn_{\ce_I} \ce_K - \ovn_{[\ce_I,\ce_K]} \ce_K \big)\\
    &= \tsum_K \theta^{2p_K} \big( \ce_I\big(\Gamma_{KK}^L\big) \ce_L - \ce_K\big(\Gamma_{IK}^L\big) \ce_L + \Gamma_{KK}^L \Gamma_{IL}^M \ce_M
    - \Gamma_{IK}^L \Gamma_{KL}^M \ce_M - \Lambda_{IK}^L \Gamma_{LK}^M \ce_M \big).
  \end{split}
  \]
  Therefore,
  \[
  \sric^{\sharp}\big(\ce_I,\co^J\big) = \tsum_K \theta^{2p_K} \big( \ce_I\big(\Gamma_{KK}^J\big) - \ce_K\big(\Gamma_{IK}^J\big)
  + \Gamma_{KK}^L \Gamma_{IL}^J - \Gamma_{IK}^L \Gamma_{KL}^J - \Lambda_{IK}^L \Gamma_{LK}^J \big).
  \]
  Now we expand each term in terms of $\Lambda_{IJ}^K$. First,
  \begin{align*}
    \tsum_K \theta^{2p_K} \ce_I\big(\Gamma_{KK}^J\big) &= \theta^{2p_J} \big( -2\ce_I\ce_J(p_J\ln\theta) + \ce_I\ce_J(\ln\theta)\\
    &\quad + 2\tsum_K \ce_J(p_K\ln\theta) \ce_I\big( (p_J-p_K)\ln\theta \big)\\
    &\quad + \ce_I\big(\Lambda_{JK}^K\big) + 2\tsum_K \ce_I\big( (p_J-p_K)\ln\theta \big) \Lambda_{JK}^K \big).
  \end{align*}
  Next,
  \begin{align*}
    &\tsum_K \theta^{2p_K} \ce_K\big(\Gamma_{IK}^J\big)\\
    = & -\theta^{2p_J} \ce_J\ce_I(p_J\ln\theta) - \tsum_K \theta^{2p_K} \ce_K\ce_K(p_I\ln\theta) \delta_{IJ}
    + \theta^{2p_J} \ce_I\ce_J(p_I\ln\theta)\\
    & + 2\theta^{2p_J} \ce_J(p_I\ln\theta) \ce_I\big( (p_J-p_I)\ln\theta \big)
    + \tfrac{1}{2} \tsum_K \theta^{2p_K} \ce_K\big(\Lambda_{IK}^J\big) - \tfrac{1}{2} \tsum_K \theta^{2p_J} \ce_K\big(\Lambda_{IJ}^K\big)\\
    & - \tsum_K \theta^{2p_J} \ce_K\big( (p_J-p_K)\ln\theta \big) \Lambda_{IJ}^K
    - \tfrac{1}{2} \tsum_K \theta^{2(p_K+p_J-p_I)} \ce_K\big(\Lambda_{KJ}^I\big)\\
    & - \tsum_K \theta^{2(p_J+p_K-p_I)} \ce_K\big( (p_J-p_I)\ln\theta \big) \Lambda_{KJ}^I.
  \end{align*}
  In addition,
  \begin{align*}
    \tsum_K \theta^{2p_K} \Lambda_{IK}^L \Gamma_{LK}^J =& -\theta^{2p_J} \Lambda_{IJ}^L \ce_L(p_J\ln\theta)
    - \tsum_K \theta^{2p_K} \Lambda_{IK}^J \ce_K(p_J\ln\theta) + \tsum_K \theta^{2p_J} \ce_J(p_K\ln\theta) \Lambda_{IK}^K\\
    & + \tfrac{1}{2} \tsum_K \theta^{2p_K} \Lambda_{IK}^L \Lambda_{LK}^J - \tfrac{1}{2} \theta^{2p_J} \Lambda_{IK}^L \Lambda_{LJ}^K
    + \tfrac{1}{2} \tsum_{K,L} \theta^{2(p_J+p_K-p_L)} \Lambda_{IK}^L \Lambda_{JK}^L.
  \end{align*}
  Moreover,
  \begin{align*}
    \tsum_K \theta^{2p_K} \Gamma_{KK}^L \Gamma_{IL}^J &= \tsum_{K,L} \theta^{2p_K} \big( -2\ce_K(p_K\ln\theta)\delta_{KL}
    + \theta^{2(p_L-p_K)}\ce_L(p_K\ln\theta)\\
    &\quad + \theta^{2(p_L-p_K)} \Lambda_{LK}^K \big) \big( -\ce_I(p_L\ln\theta) \delta_{LJ} - \ce_L(p_I\ln\theta) \delta_{IJ}\\
    &\quad + \theta^{2(p_J-p_I)}\ce_J(p_I\ln\theta) \delta_{IL} + \tfrac{1}{2} \Lambda_{IL}^J
    - \tfrac{1}{2}\theta^{2(p_J-p_L)} \Lambda_{IJ}^L - \tfrac{1}{2}\theta^{2(p_J-p_I)} \Lambda_{LJ}^I \big).
  \end{align*}
  Finally,
  \begin{align*}
    \tsum_K \theta^{2p_K} \Gamma_{IK}^L \Gamma_{KL}^J &= \tsum_{K,L} \theta^{2p_K} \big( -\ce_I(p_K\ln\theta) \delta_{KL}
    - \ce_K(p_I\ln\theta) \delta_{IL} + \theta^{2(p_L-p_I)} \ce_L(p_I\ln\theta) \delta_{IK}\\
    &\quad + \tfrac{1}{2}\Lambda_{IK}^L - \tfrac{1}{2} \theta^{2(p_L-p_K)} \Lambda_{IL}^K - \tfrac{1}{2} \theta^{2(p_L-p_I)} \Lambda_{KL}^I \big)\\
    &\quad \,\big( -\ce_K(p_L\ln\theta) \delta_{LJ} - \ce_L(p_K\ln\theta) \delta_{KJ} + \theta^{2(p_J-p_K)} \ce_J(p_K\ln\theta) \delta_{KL}\\
    &\quad + \tfrac{1}{2} \Lambda_{KL}^J - \tfrac{1}{2} \theta^{2(p_J-p_L)} \Lambda_{KJ}^L - \tfrac{1}{2} \theta^{2(p_J-p_K)} \Lambda_{LJ}^K \big).
  \end{align*}
  It may now be verified that for $I \neq J$, every term that appears in $\sric^{\sharp}\big(\ce_I,\co^J\big)$ is of one of the types
  described in the statement of the lemma.
\end{proof}

In the following two lemmas, we deduce the bounds for $\cb_I^J$ that will be required for improving on the estimates for the expansion normalized eigenframe.

\begin{lemma} \label{estimates for curly B}
  Let $3\leq n\in\nn{}$, $k_0\in\nn{}$, $r>0$, $\delta\in (0,1)$ and fix a $\delta$-admissible potential $V\in C^\infty(\R)$. Given that
  $(M,g,\varphi)$ is a solution to the Einstein--nonlinear scalar field equations satisfying the $(\delta,n,k_0,r)$-assumptions, assume that there is
  an eigenframe $\{\eig_I\}_{I=1}^n$ with the properties listed at the beginning of Subsection~\ref{ssection:the eigenframe}. Then there is a constant
  $C$ such that 
  \[
  \big|E_\I\big(\cb_I^I\big)\big| \leq Ct^{-1+\delta}, \qquad \big|E_\I\big(\cb_I^J\big)\big| \leq C\langle \ln t \rangle^{|\I|} t^{-1+\rp_I-\rp_J+2\delta}
  \]
  for $t\leq t_0$, where $I \neq J$, there is no summation over $I$, and $|\I| \leq k_0-1$.
\end{lemma}

\begin{proof}
  We begin by estimating $e_0p_I$ using Lemma~\ref{equations for the eigenframe}. By \eqref{evolution equation for the second fundamental form},
  \begin{align*}
    e_0 \theta &= -\roScal_h - \theta^2 + N^{-1} \Delta_h N + |\md\varphi|_h^2 + \tfrac{2n}{n-1} V\circ\varphi,\\
    g\big((\lie_{e_0}K)(\eig_I),\eig_I\big) &= -\sric(\eig_I,\eig_I) - \theta^2p_I + N^{-1}(\ovn^2 N)(\eig_I,\eig_I)
    + \eig_I(\varphi)^2 + \tfrac{2}{n-1} V\circ\varphi
  \end{align*}
  (no summation). Note that the terms involving $\theta^2$ cancel out in the right-hand side of \eqref{equation for eigenvalues}. Hence, we can
  estimate the remaining terms using the assumptions, Lemma~\ref{bounds for the scalar field}, Lemma~\ref{estimate for B} and
  Proposition~\ref{bounds for change of basis matrix} to obtain
  \[
  t\|e_0p_I\|_{C^{k_0-1}(\Sigma)} \leq C\langle \ln t \rangle^{k_0-1}t^{2\delta}.
  \]
  The result for $\cb_I^I$ now follows from its definition and Lemma~\ref{estimate for B}.

  For $\cb_I^J$, we have
  \[
  (\lie_{e_0}K)\big(\ce_I,\co^J\big) = \theta^{p_J-p_I} g\big((\lie_{e_0}K)(\eig_I),\eig_J\big).
  \]
  The result follows from the assumptions, the definition of $\cb_I^J$, \eqref{evolution equation for the second fundamental form},
  Lemmas~\ref{estimates for powers of theta} and \ref{estimate for B}, and Proposition~\ref{bounds for change of basis matrix}.
\end{proof}

\begin{lemma} \label{improved estimates for curly B}
  Let $3\leq n\in\nn{}$, $2\leq k_0\in\nn{}$, $r>0$, $\delta\in (0,1)$ and fix a $\delta$-admissible potential $V\in C^\infty(\R)$. Given that
  $(M,g,\varphi)$ is a solution to the Einstein--nonlinear scalar field equations satisfying the $(\delta,n,k_0,r)$-assumptions, assume that there is
  an eigenframe $\{\eig_I\}_{I=1}^n$ with the properties listed at the beginning of Subsection~\ref{ssection:the eigenframe}. Let $0<\bfc_1\in\rn{}$,
  $1\leq k\leq k_0-1$, $x \in \Sigma$, and suppose that the following estimates hold for some $m\in\nn{}$, $t\leq t_0$ and all $|\I| \leq k$: For each
  $\ce_I$ separately, either
  \begin{equation}\label{eq:ceI ind assump}
    \big|E_\I\big(\ce_I^i\big)(x)\big| \leq \bfc_1t^{-1+\rp_I(x)+m\delta}, \quad \textrm{or} \quad \big|E_\I\big(\ce_I^i\big)(x)\big| \leq \bfc_1;
  \end{equation}
  and for each $\Lambda_{IJ}^K$ separately, either
  \[
  \big|E_\I\big(\Lambda_{IJ}^K\big)(x)\big| \leq \bfc_1t^{-1+\rp_I(x)+\rp_J(x)-\rp_K(x) + m\delta},
  \quad \textrm{or} \quad \big|E_\I\big(\Lambda_{IJ}^K\big)(x)\big| \leq \bfc_1.
  \]
  Then there is a constant $C_2$, depending only on $\bfc_1$ and the constants appearing in the assumptions of Theorem~\ref{thm:asymptotics} and the
  previous lemmas of the present section, such that, for $t\leq t_0$ and $I \neq J$,
  \[
  \big|E_\I\big(\cb_I^J\big)(x)\big| \leq C_2t^{-1+\rp_I(x)-\rp_J(x) + (m+1)\delta} + C_2t^{-1+2\delta} \min\{1,t^{2(\rp_I(x)-\rp_J(x))}\},
  \]
  where $|\I| \leq k-1$. In particular, if all the $E_\I\big(\ce_I^i\big)(x)$ and the $E_\I\big(\Lambda_{IJ}^K\big)(x)$ are bounded, then
  \[
  \big|E_\I\big(\cb_I^J\big)(x)\big| \leq C_2t^{-1+2\delta} \min\{1,t^{2(\rp_I(x)-\rp_J(x))}\}.
  \]
\end{lemma}

\begin{proof}
  For each $\ce_I$, we, in this proof, say that case 1 (or 2) holds if the left (or right) inequality in (\ref{eq:ceI ind assump}) holds. We use similar
  terminology for the $\Lambda_{IJ}^K$. In addition, when we write $\langle \ln t \rangle^*$, it is understood that $*$ stands for a positive integer
  whose exact value is unimportant (but only depends on $k$). In particular, its exact value may change from line to line. Finally, to simplify the
  notation, we omit the argument $x$ throughout the proof. Thus, it is to be understood that all functions on $\Sigma$ that appear below are evaluated
  at the point $x \in \Sigma$.
    
  Let $I \neq J$. We illustrate the idea with one of the types of terms appearing in Lemma~\ref{spatial ricci}. Let $|\I| \leq k-1$. Consider
  $\theta^{2p_J} * \ce_I * \ce_J$. If case 1 holds for both $\ce_I$ and $\ce_J$, then, keeping
  Lemma~\ref{initial estimates for expansion normalized variables} in mind, 
  \[
  |E_\I(\theta^{2p_J} * \ce_I * \ce_J)| \leq C\langle \ln t \rangle^* t^{-2+\rp_I-\rp_J+2\max\{m,2\}\delta}.
  \]
  If case 2 holds for $\ce_I$, and case 1 holds for $\ce_J$, Lemma~\ref{bounds for the pI} implies
  \[
  |E_\I(\theta^{2p_J} * \ce_I * \ce_J)| \leq C\langle \ln t \rangle^* t^{-1-\rp_J+m\delta} \leq C\langle \ln t \rangle^* t^{-2+\rp_I-\rp_J + (m+2)\delta}.
  \]
  If case 1 holds for $\ce_I$ and case 2 holds for $\ce_J$, a similar estimate results. If case 2 holds for both and $I < J$, then
  \begin{equation}\label{eq:simple two two term}
    |E_\I(\theta^{2p_J} * \ce_I * \ce_J)| \leq C \langle \ln t \rangle^* t^{-2\rp_J} \leq C\langle \ln t \rangle^* t^{-2+4\delta};
  \end{equation}
  and if $I > J$, then
  \begin{equation}\label{eq:simple two two term I gt J}
    |E_\I(\theta^{2p_J} * \ce_I * \ce_J)| \leq C\langle \ln t \rangle^* t^{-2+2(\rp_I-\rp_J) + 2(1-\rp_I)} \leq C\langle \ln t \rangle^* t^{-2+2(\rp_I-\rp_J) + 4\delta}.
  \end{equation}
  If $\theta * \ce * \Lambda$ and $\theta * \Lambda * \Lambda$ are any of the types of terms appearing in
  $\sric^{\,\sharp}\big(\ce_I,\co^J\big)$, it is readily verified that, whenever both of the non $\theta$ factors are in case 1, or one is in case 1 and
  the other in case 2, the same pattern as above arises. Thus
  \[
  |E_\I(\theta * \ce * \Lambda)| + |E_\I(\theta * \Lambda * \Lambda)| \leq C \langle \ln t \rangle^* t^{-2+\rp_I-\rp_J+(m+2)\delta}
  \]
  when both non $\theta$ factors are in case 1, or one is in case 1 and the other in case 2. Note, however, that if one is in case 1 and the other
  in case 2, we sometimes need to appeal to (\ref{subcritical condition}). It remains to check what happens when both factors
  are in case 2. Note that in that case, the only time dependent contribution to the bound comes from the power of $\theta$.

  The following are the possible upper bounds for the terms in $\sric^{\,\sharp}\big(\ce_I,\co^J\big)$, for $I \neq J$, when both of the factors are in
  case 2:
  \begin{subequations}\label{seq:two two all cases}
    \begin{align}
      &C\langle \ln t \rangle^* t^{-2\rp_J};\label{first case}\\
      &C\langle \ln t \rangle^* t^{2(\rp_I-\rp_J-\rp_K)}, & K &\neq J;\label{second case}\\
      &C\langle \ln t \rangle^* t^{-2\rp_K}, &K &\neq I;\label{third case}\\
      &C\langle \ln t \rangle^* t^{2(\rp_I-\rp_K-\rp_L)}, &K &\neq L;\label{fourth case}\\
      &C\langle \ln t \rangle^* t^{2(\rp_L-\rp_J-\rp_K)}, &K &\neq J, \quad K \neq I.\label{fifth case}
    \end{align}
  \end{subequations}
  All the expressions in (\ref{seq:two two all cases}) can be estimated by the far right hand side of (\ref{eq:simple two two term}) due to
  (\ref{subcritical condition}) and Lemma~\ref{bounds for the pI}. Assume now that $I>J£$. Then the expression \eqref{first case} can be estimated
  as in (\ref{eq:simple two two term I gt J}). In the case of \eqref{second case}, 
  \[
  2(\rp_I-\rp_J-\rp_K) = -2+2(\rp_I-\rp_J) + 2(1-\rp_K) > -2+2(\rp_I-\rp_J) + 4\delta.
  \]
  In the case of \eqref{third case}, since $K\neq I$, 
  \[
  -2\rp_K = -2+2(\rp_I-\rp_J) + 2(1+\rp_J-\rp_I-\rp_K) > -2+2(\rp_I-\rp_J) + 4\delta.
  \]
  In the case of \eqref{fourth case}, since $K \neq L$,
  \[
  2(\rp_I-\rp_K-\rp_L) = -2+2(\rp_I-\rp_J) + 2(1+\rp_J-\rp_K-\rp_L) > -2+2(\rp_I-\rp_J) + 4\delta.
  \]
  Finally, in the case of \eqref{fifth case}, since $K \neq I$,
  \[
  2(\rp_L-\rp_J-\rp_K) = -2+2(\rp_I-\rp_J) + 2(1+\rp_L-\rp_I-\rp_K) > -2 + 2(\rp_I-\rp_J) + 4\delta.
  \]
  Summing up, we have shown that
  \[
  \big|E_\I\big( \sric^{\sharp}\big(\ce_I,\co^J\big) \big)\big| \leq Ct^{-2+\rp_I-\rp_J + (m+1)\delta} + Ct^{-2+2\delta} \min\{1,t^{2(\rp_I-\rp_J)}\},
  \]
  for $I \neq J$ and $|\I| \leq k-1$. In particular, if all the $E_\I\big(\ce_I^i\big)$ and the $E_\I\big(\Lambda_{IJ}^K\big)$ are bounded for
  $|\I| \leq k$, then
  \[
  \big|E_\I\big( \sric^{\sharp}\big(\ce_I,\co^J\big) \big)\big| \leq Ct^{-2+2\delta} \min\{1,t^{2(\rp_I-\rp_J)}\}.
  \]
  Next, we consider the remaining terms in $\cb_I^J$, starting with
  \begin{equation}\label{eq:dvp gradphi B est}
    (\md\varphi \otimes \rograd_h\varphi)\big(\ce_I,\co^J\big) = \theta^{2p_J} \ce_I(\varphi) \ce_J(\varphi).
  \end{equation}
  Moreover,
  \begin{equation}\label{eq: nsq N sharp B est}
    \begin{split}
      (\ovn^2 N)^\sharp\big(\ce_I,\co^J\big)
      &= \theta^{2p_J} (\ovn^2N)(\ce_I,\ce_J)\\
      &= \theta^{2p_J} \ce_I\ce_JN + \theta^{2p_J} \ce_I(p_J\ln\theta) \ce_J(N) + \theta^{2p_J} \ce_J(p_I\ln\theta)\ce_I(N)\\
      &\quad - \tfrac{1}{2} \theta^{2p_J} \Lambda_{IJ}^K \ce_K(N) + \tfrac{1}{2} \tsum_K \theta^{2p_K} \Lambda_{IK}^J \ce_K(N)\\
      &\quad + \tfrac{1}{2} \tsum_K \theta^{2(p_J+p_K-p_I)} \Lambda_{JK}^I \ce_K(N).
    \end{split}
  \end{equation}
  This observation finishes the proof, since each of the terms on the far right hand sides of (\ref{eq:dvp gradphi B est}) and
  (\ref{eq: nsq N sharp B est}) is qualitatively of the same type as a term that appeared when estimating the components of the spatial Ricci tensor. 
\end{proof}

Next, we prove the crucial inductive step required to obtain the asymptotics: it is possible to improve on the estimates for the expansion normalized
eigenframe, at the cost of losing two derivatives at each step.

\begin{lemma} \label{improved estimates for expansion normalized variables}
  Let $3\leq n\in\nn{}$, $4\leq k_0\in\nn{}$, $r>0$, $\delta\in (0,1)$ and fix a $\delta$-admissible potential $V\in C^\infty(\R)$. Given that
  $(M,g,\varphi)$ is a solution to the Einstein--nonlinear scalar field equations satisfying the $(\delta,n,k_0,r)$-assumptions, assume that there is
  an eigenframe $\{\eig_I\}_{I=1}^n$ with the properties listed at the beginning of Subsection~\ref{ssection:the eigenframe}. Fix an open cover
  $U_j$ of $\Sigma$, with corresponding constants $\alpha_j$, as in Lemma~\ref{technical lemma}, and let $m\in\nn{}$. Suppose that for each $j$, the
  following pointwise estimates hold in $U_j$ for all $|\I| \leq k$, $t\leq t_0$ and some $k\leq k_0-2$: for $\ce_I$,
  \begin{align*}
    \big|E_\I\big(\ce_I^i\big)\big| &\leq C_jt^{-1+\rp_I+m\delta + \alpha_j}, & \text{if} \quad -1+\rp_I+m\delta + \alpha_j &< 0,\\
    \big|E_\I\big(\ce_I^i\big)\big| &\leq C_j, & \text{if} \quad -1+\rp_I+m\delta + \alpha_j &> 0;
  \end{align*}
  for $\co^I$,
  \begin{align*}
    \big|E_\I\big(\co_i^I\big)\big| &\leq C_jt^{-1-\rp_I+m\delta + \alpha_j}, & \text{if} \quad -1-\rp_I+m\delta + \alpha_j &< 0,\\
    \big|E_\I\big(\co_i^I\big)\big| &\leq C_j, & \text{if} \quad -1-\rp_I+m\delta + \alpha_j &> 0;
  \end{align*}
  and for $\Lambda_{IJ}^K$,
  \begin{align*}
    \big|E_\I\big(\Lambda_{IJ}^K\big)\big| &\leq C_jt^{-1+\rp_I+\rp_J-\rp_K + m\delta + \alpha_j},
    & \text{if} \quad -1+\rp_I+\rp_J-\rp_K + m\delta + \alpha_j &< 0,\\
    \big|E_\I\big(\Lambda_{IJ}^K\big)\big| &\leq C_j, & \text{if} \quad -1+\rp_I+\rp_J-\rp_K + m\delta + \alpha_j &> 0.
  \end{align*}
  Then the same estimates hold for $|\I| \leq k-2$ and $t\leq t_0$, with $m$ replaced by $m+1$.
\end{lemma}

\begin{proof}
  Fix one of the open sets $U_j$. For the rest of the proof it is understood that all functions on $\Sigma$ are restricted to $U_j$. For each
  $\ce_I$, we say that case 1 holds if $-1+\rp_I+m\delta + \alpha_j < 0$, and we say that case 2 holds if $-1+\rp_I+m\delta + \alpha_j > 0$. The conventions
  concerning the $\co^I$ and the $\Lambda_{IJ}^K$ are similar. Moreover, we use the notation $\langle \ln t \rangle^*$ introduced in the proof of
  Lemma~\ref{improved estimates for curly B}. 

  Consider $\Lambda_{IJ}^K$. The idea is to estimate $\p_t\Lambda_{IJ}^K$ by estimating all the terms on the right-hand side of
  \eqref{equation for expansion normalized structure coefficients}. Let $|\I| \leq k-2$. By Lemma~\ref{estimates for curly B}, if case 1 holds
  for $\Lambda_{IJ}^K$,
  \[
  \big|E_\I\big( \cb_L^L \Lambda_{IJ}^K \big)\big| \leq C_jt^{-2+\rp_I+\rp_J-\rp_K + (m+1)\delta + \alpha_j},
  \]
  (no summation on $L$). Next, we look at $\tsum_{L \neq K} \theta^{2(p_K-p_L)} \cb_K^L \Lambda_{IJ}^L$. If case 1 holds for $\Lambda_{IJ}^L$, then
  Lemma~\ref{estimates for curly B} implies
  \[
  \begin{split}
    \big|E_\I\big( \theta^{2(p_K-p_L)} \cb_K^L \Lambda_{IJ}^L \big)\big|
    &\leq C_j\langle \ln t \rangle^* t^{-2(\rp_K-\rp_L)} t^{-1+\rp_K-\rp_L + 2\delta} t^{-1+\rp_I+\rp_J-\rp_L + m\delta + \alpha_j}\\
    &= C_j\langle \ln t \rangle^* t^{-2+\rp_I+\rp_J-\rp_K+ (m+2)\delta + \alpha_j};
  \end{split}
  \]
  and if case 2 holds for $\Lambda_{IJ}^L$, then Lemma~\ref{improved estimates for curly B} and \eqref{subcritical condition} yield ($I\neq J$)
  \[
  \begin{split}
    \big|E_\I\big( \theta^{2(p_K-p_L)} \cb_K^L \Lambda_{IJ}^L \big)\big|
    &\leq C_j\langle \ln t \rangle^* t^{-2(\rp_K-\rp_L)}\big( t^{-1+\rp_K-\rp_L + (m+1)\delta} + t^{-1+2\delta}\min\{1,t^{2(\rp_K-\rp_L)}\} \big)\\
    &\leq C_j\langle \ln t \rangle^* \big( t^{-2+\rp_I+\rp_J-\rp_K + (m+3)\delta} + t^{-1+2\delta} \big).
  \end{split}
  \]
  The reader may verify that the remaining terms on the right-hand side of \eqref{equation for expansion normalized structure coefficients}
  can be estimated in a similar way. We conclude that
  \[
  \big|E_\I\big(\p_t\Lambda_{IJ}^K\big)\big| \leq C_jt^{-2+\rp_I+\rp_J-\rp_K+(m+1)\delta + \alpha_j} + C_jt^{-1+\delta}.
  \]
  Note that Lemma~\ref{technical lemma} ensures that (\ref{eq:subdominant bd away from zero}) holds in $U_j$ for $I\neq J$. Therefore, for $s < t_0$,
  \[
  \begin{split}
    \big|E_\I\big(\Lambda_{IJ}^K\big)(t_0) - E_\I\big(\Lambda_{IJ}^K\big)(s)\big| &\leq \tint_s^{t_0}\big|E_\I\big(\p_t\Lambda_{IJ}^K\big)(t)\big|\md t\\
    &\leq C_j\big(t_0^{\delta} - s^{\delta}\big)\\
    &\quad + C_j \frac{t_0^{-1+\rp_I+\rp_J-\rp_K+(m+1)\delta + \alpha_j} - s^{-1+\rp_I+\rp_J-\rp_K+(m+1)\delta + \alpha_j}}{-1+\rp_I+\rp_J-\rp_K+(m+1)\delta + \alpha_j}. 
  \end{split}
  \]
  There is thus a constant $C_j$ such that, if $-1+\rp_I+\rp_J-\rp_K + (m+1)\delta + \alpha_j > 0$, 
  \[
  \big|E_\I\big(\Lambda_{IJ}^K\big)(s)\big| \leq C_j
  \]
  in $U_j$ and for all $s < t_0$. On the other hand, if $-1+\rp_I+\rp_J-\rp_K + (m+1)\delta + \alpha_j < 0$, then
  \[
  \big|E_\I\big(\Lambda_{IJ}^K\big)(s)\big| \leq C_j s^{-1+\rp_I+\rp_J-\rp_K + (m+1)\delta + \alpha_j}
  \]
  in $U_j$ and for all $s < t_0$. This proves the result for the $\Lambda_{IJ}^K$. The result for the $\ce_I$ and the $\co^I$ follows similarly by
  using \eqref{equation for expansion normalized frame} and \eqref{equation for expansion normalized dual frame}.
\end{proof}

Now we are ready to deduce asymptotics for $\ce_I$ and $\co^I$.

\begin{thm} \label{asymptotics of the eigenframe}
  Let $3\leq n\in\nn{}$, $\ell\in\nn{}$, $r>0$, $\delta\in (0,1)$ and fix a $\delta$-admissible potential $V\in C^\infty(\R)$. If $k_0$ is large enough,
  depending only on $\delta$ and $\ell$; $(M,g,\varphi)$ is a solution to the Einstein--nonlinear scalar field equations satisfying the
  $(\delta,n,k_0,r)$-assumptions; and there is an eigenframe $\{\eig_I\}_{I=1}^n$ with the properties listed at the beginning of
  Subsection~\ref{ssection:the eigenframe}, then there is a $C^\ell$ frame $\{\rce_I\}_{I=1}^n$ on $\Sigma$, with dual frame $\{\ro^I\}_{I=1}^n$, such that
  \[
  \|\ce_I(t) - \rce_I\|_{C^\ell(\Sigma)} + \|\co^I(t) - \ro^I\|_{C^\ell(\Sigma)} \leq Ct^\delta
  \]
  for all $t\leq t_0$. Define the $C^\ell$ Riemannian metric $\rch$ and the $C^\ell$ $(1,1)$-tensor field $\rK$ on $\Sigma$ by
  \begin{equation}\label{eq:K ring H ring def}
    \rch := \tsum_I \ro^I \otimes \ro^I, \qquad \rK := \tsum_I \rp_I \ro^I \otimes \rce_I.
  \end{equation}
  Then
  \[
  \|\ch(t) - \rch\|_{C^\ell(\Sigma)} + \|\K(t) - \rK\|_{C^\ell(\Sigma)} \leq Ct^\delta
  \]
  for all $t\leq t_0$. Finally, $(\Sigma,\rch,\rK,\rphi,\rpsi)$ are $C^\ell$ robust nondegenerate quiescent initial data on the singularity for the
  Einstein--nonlinear scalar field equations. 
\end{thm}

\begin{proof}
  First, note that if $I$, $J$ and $K$ are distinct, then
  \[
  1-\rp_I-\rp_J+\rp_K \leq 1+\sqrt{3}\big(\rp_I^2 + \rp_J^2 + \rp_K^2\big)^{1/2} \leq 1 + \sqrt{3}.
  \]
  Moreover, by Lemma~\ref{bounds for the pI}, $1\pm\rp_I < 2-2\delta$ for all $I$. Hence, there is a large enough positive integer $N$,
  depending only on $\delta$, such that 
  \[
  -1+\rp_I+\rp_J-\rp_K + N\delta > 0, \qquad -1-\rp_I + N\delta > 0,
  \]
  where $I \neq J$ in the first inequality. Fix an open cover $U_j$ of $\Sigma$ as in Lemma~\ref{technical lemma}. Starting from the bounds of
  Lemma~\ref{initial estimates for expansion normalized variables}, we can now apply Lemma~\ref{improved estimates for expansion normalized variables}
  repeatedly, at most $N-1$ times, to obtain
  \[
  \big|E_\I\big(\ce_I^i\big)\big| + \big|E_\I\big(\co_i^I\big)\big| + \big|E_\I\big(\Lambda_{IJ}^K\big)\big| \leq C_j
  \]
  in $(0,t_0] \times U_j$, and for $|\I| \leq \ell + 2$, assuming $k_0 \geq 2N + \ell + 2$. By taking the maximum of the constants $C_j$,
  we get a constant $C$ such that
  \[
  \big|E_\I\big(\ce_I^i\big)\big| + \big|E_\I\big(\co_i^I\big)\big| + \big|E_\I\big(\Lambda_{IJ}^K\big)\big| \leq C,
  \]
  for $|\I| \leq \ell + 2$, in $(0,t_0] \times \Sigma$.

  Now knowing that all the $\ce_I$ and the $\Lambda_{IJ}^K$ are bounded in $C^{\ell+1}(\Sigma)$ for all $t \in (0,t_0]$,
  Lemma~\ref{improved estimates for curly B} yields
  \[
  \big|E_\I\big(\cb_I^J\big)\big| \leq Ct^{-1+2\delta} \min\{1,t^{2(\rp_I-\rp_J)}\},
  \]
  for $I \neq J$ and $|\I| \leq \ell+1$. Combining this estimate with Lemmas~\ref{estimates for powers of theta},
  \ref{equations for the expansion normalized eigenframe lemma} and \ref{estimates for curly B} yields
  \[
  \|\p_t\ce_I^i\|_{C^\ell(\Sigma)} + \|\p_t\co_i^I\|_{C^\ell(\Sigma)} + \|\p_t \Lambda_{IJ}^K\|_{C^\ell(\Sigma)} \leq C t^{-1+\delta}.
  \]
  But then, for $u < s\leq t_0$,
  \[
  \big|E_\I\big(\ce_I^i\big)(s) - E_\I\big(\ce_I^i\big)(u)\big| \leq \tint_u^s \big|E_\I\big(\p_t\ce_I^i\big)(t)\big|\md t \leq
  C(s^\delta - u^\delta).
  \]
  for all $t\leq t_0$ and all $|\I| \leq \ell$. Thus, there are functions $\rce_I^i \in C^\ell(\Sigma)$ such that 
  \[
  \|\ce_I^i(t) - \rce_I^i\|_{C^\ell(\Sigma)} \leq C t^\delta,
  \]
  and we define $\rce_I := \rce_I^iE_i$. Similarly, there are functions $\ro_i^I, \mathring{\Lambda}_{IJ}^K \in C^\ell(\Sigma)$ such that
  \[
  \|\co_i^I(t) - \ro_i^I\|_{C^\ell(\Sigma)} + \|\Lambda_{IJ}^K(t) - \mathring{\Lambda}_{IJ}^K\|_{C^\ell(\Sigma)} \leq C t^\delta,
  \]
  and we define $\ro^I := \ro_i^I\eta^i$. Letting $t\downarrow 0$ in the relation $\co_i^J \ce_I^i = \delta_I^J$, it is clear that
  $\{\rce_I\}_{I=1}^n$ is a well defined frame on $\Sigma$ with dual frame $\{\ro^I\}_{I=1}^n$. Moreover, since
  \[
  \Lambda_{IJ}^K = \co^K\big([\ce_I,\ce_J]\big) = \ce_I\big(\ce_J^j\big)\co_j^K - \ce_J\big(\ce_I^i\big) \co_i^K + \ce_I^i\ce_J^j \co^K\big([E_i,E_j]\big),
  \]
  we can let $t \downarrow 0$ to conclude that $\mathring{\Lambda}_{IJ}^K$ are the structure coefficients of the frame $\{\rce_I\}_{I=1}^n$. It only
  remains to prove the statement about the initial data on the singularity. Note, to this end, that
  \[
  \ch = \tsum_I \co^I \otimes \co^I, \qquad \K = \tsum_I p_I \co^I \otimes \ce_I.
  \]
  Hence, from the definition of $\rch$ and $\rK$, it is clear that
  \[
  \|\ch(t) - \rch\|_{C^\ell(\Sigma)} + \|\K(t) - \rK\|_{C^\ell(\Sigma)} \leq C t^\delta.
  \]
  Now, we verify that the conditions of initial data on the singularity hold. First,
  \[
  \tr\rK = \tsum_I \rp_I = 1, \qquad \tr\rK^2 + \rpsi^2 = \tsum_I \rp_I^2 + \rpsi^2 = 1.
  \]
  Next, it is clear that $\rK$ is self adjoint with respect to $\rch$. Moreover, $|\rp_I-\rp_J|\geq r$ for $I\neq J$, since $|p_I-p_J|\geq r$ for $I\neq J$;
  and (\ref{eq:cond on eigenvalues}) follows from (\ref{subcritical condition}).
  Finally, for the asymptotic momentum constraint (the second equation in part (2) of Definition~\ref{def:idos}), we compute, appealing to
  (\ref{eq:Koszul ech}),
  \begin{align*}
    &\diver_h K(\ce_I)\\
    &= \tsum_J \theta^{2p_J} h\big(\ovn_{\ce_J} K(\ce_I), \ce_J\big)\\
    &= \tsum_J \theta^{2p_J} h\big(\ovn_{\ce_J}(\theta p_I \ce_I) - \Gamma_{JI}^K K(\ce_K), \ce_J\big)\\
    &= \tsum_J \theta^{2p_J} h\big( \ce_J(\theta p_I) \ce_I + \theta p_I \Gamma_{JI}^L \ce_L - \tsum_K \theta p_K \Gamma_{JI}^K \ce_K, \ce_J \big)\\
    &= \ce_I(\theta p_I) + \tsum_J \theta(p_I-p_J)\Gamma_{JI}^J\\
    &= \theta\ce_I(p_I) + p_I \ce_I(\theta) + \tsum_J \theta(p_I-p_J)\big( -\ce_I(p_J\ln\theta) + \Lambda_{JI}^J \big)\\
    &= \theta \ce_I(p_I) + p_I\ce_I(\theta) - \theta p_I \ce_I\big( \tsum_J p_J\ln\theta \big)
    + \displaystyle\tsum_J\theta p_J \ce_I(p_J\ln\theta) + \tsum_J \theta(p_I-p_J) \Lambda_{JI}^J\\
    &= \theta\ce_I(p_I) + \theta\ce_I(\ln\theta) \tsum_J p_J^2 + \tfrac{\theta\ln\theta}{2} \ce_I\big( \tsum_J p_J^2 \big)
    + \displaystyle\tsum_J \theta(p_I-p_J) \Lambda_{JI}^J
  \end{align*}
  (no summation on $I$). On the other hand,
  \[
  e_0(\varphi)\ce_I(\varphi)  = \theta\Psi\ce_I( \Phi - \Psi\ln\theta ) = \theta\Psi\ce_I(\Phi) - \theta\ce_I(\ln\theta) \Psi^2
  - \tfrac{\theta\ln\theta}{2} \ce_I\big(\Psi^2\big).
  \]
  Hence, after multiplying by $1/\theta$, the momentum constraint becomes
  \[
  \ce_I(p_I) + \ce_I(\ln\theta)\big( \tsum_J p_J^2 + \Psi^2 - 1 \big) + \frac{1}{2}\ln\theta\ce_I\big( \tsum_Jp_J^2 + \Psi^2 \big)
  + \tsum_J(p_I-p_J)\Lambda_{JI}^J = \Psi \ce_I(\Phi).
  \]
  Now, after multiplication by $\theta^{-2}$, the Hamiltonian constraint reads
  \[
  \tsum_J p_J^2 + \Psi^2 - 1 = \theta^{-2}( \roScal_h - |\md\varphi|_h^2 - 2V \circ\varphi ).
  \]
  Thus, by the assumptions and Lemmas~\ref{bounds for the scalar field} and \ref{estimate for B},
  \[
  \|\tsum_J p_J^2 + \Psi^2 - 1\|_{C^{k_0-1}(\Sigma)} \leq Ct^\delta.
  \]
  Going back to the momentum constraint, we can now let $t \downarrow 0$ to conclude that
  \[
  \rce_I(\rp_I) + \tsum_J (\rp_I-\rp_J) \mathring{\Lambda}_{JI}^J = \rpsi \rce_I(\rphi).
  \]
  Since the left-hand side of this equation coincides with $\diver_{\rch} \rK(\rce_I)$, the lemma follows. 
\end{proof}

\begin{proof}[Proof of Theorem~\ref{thm:asymptotics}]
  The statement follows essentially immediately from Theorem~\ref{asymptotics of the eigenframe}. The only difference is that we do not
  assume the existence of an eigenframe in the statement of Theorem~\ref{thm:asymptotics}. However, this is easily achieved by taking a
  finite covering space; cf. the comments at the beginning of Subsection~\ref{ssection:the eigenframe}. On this covering space, the
  conclusions of Theorem~\ref{asymptotics of the eigenframe} hold. Moreover, since the only ambiguity in the choice of frame is a sign, the objects
  introduced in (\ref{eq:K ring H ring def}) descend to the quotient. Moreover, the conclusions of Theorem~\ref{thm:asymptotics} follow
  from the conclusions of Theorem~\ref{asymptotics of the eigenframe}. 
\end{proof}

\section{Proof of the main result}\label{section: proof of main theorem}

At this point, we are ready to prove Theorem~\ref{thm: big bang formation} by combining the main theorem of \cite{OPR} with
Theorems~\ref{thm:hod} and \ref{thm:asymptotics}.

\begin{proof}[Proof of Theorem~\ref{thm: big bang formation}]
  Given admissibility thresholds $\sigma_V,\sigma_p\in (0,1)$ as in the statement of the theorem, $\sigma$ is defined as in
  \cite[(11), p.~6]{OPR}:
  \begin{equation}\label{eq:sigmadef}
    \sigma=\min\big(\tfrac{\sigma_V}{3},\tfrac{\sigma_p}{5}\big).
  \end{equation}
  In \cite[(12), p.~6]{OPR}, lower bounds on $k_0$ and $k_1$ are specified. Here we let 
  \[
  k_0:=\left\lceil \tfrac {n+1}2 \right\rceil;
  \]
  i.e., $k_0$ is the minimal number such that \cite[(12a), p.~6]{OPR} holds. Note, however, that $k_0\geq 2$. We also assume that
  $k_1$ is the minimal integer such that \cite[(12b), p.~6]{OPR} holds. In particular, $\sigma$ is determined by $\sigma_V$ and $\sigma_p$;
  $k_0$ is determined by $n$; and $k_1$ is determined by $\sigma_V$, $\sigma_p$ and $n$. We define $\kappa:=k_1+2$.

  Next, note that the assumptions include the requirement that the eigenvalues of $\mK_0$ be distinct. Due to the arguments presented in the
  proof of \cite[Lemma~A.1, p.~223]{RinWave}, it follows that a finite (at most a $2^n$-fold) covering space of $\S$ has a global frame. Going
  to this finite cover, we can assume
  that there is a global orthonormal frame $\{E_i\}_{i=1}^n$ of $(\S,h_{\refer})$. Considering the Sobolev norms on the left hand side of
  (\ref{eq: Sobolev bound assumption}), they are at most increased by a factor of $2^n$ when going to the finite cover. Moreover, the norms
  used in the statement of \cite[Theorem~12, pp.~6-7]{OPR} are defined using a frame $\{E_i\}_{i=1}^n$ and the norms in
  (\ref{eq: Sobolev bound assumption}) are defined using the metric $h_{\refer}$. However, the two norms are equivalent up to constants depending
  only on $(\S,h_{\refer})$ and the frame. In other words, this difference does not affect the statement. Fixing a $\zeta_0$ as in the statement
  of the theorem, there is a $\zeta_0$ as in the statement of \cite[Theorem~12, pp.~6-7]{OPR} (obtained by compensating for the covering space
  and the equivalence of the norms), say $\zeta_0'$. Due to \cite[Theorem~12, pp.~6-7]{OPR}, we then obtain a $\zeta_1$. Given that the assumptions
  of the theorem are satisfied with this $\zeta_1$ (and the original $\zeta_0$), the conclusions of \cite[Theorem~12, pp.~6-7]{OPR} hold. Moreover,
  the estimates derived in the proof of this theorem hold.

  Due to \cite[Proposition~74, p.~29]{OPR}, there is a constant $\rho_0>0$ depending only on $\zeta_0$, $\sigma$, $k_0$, $k_1$, $(\Sigma,h_{\refer})$
  and $\{E_i\}_{i=1}^n$. In our setting, $\rho_0$ thus only depends on $\zeta_0$, $\sigma_V$, $\sigma_p$, $(\Sigma,h_{\refer})$ and $\{E_i\}_{i=1}^n$. In
  accordance with \cite[Notation~75, p.~29]{OPR}, we then say that $C$ is a \textit{standard constant} if it is strictly positive and only depends on
  $\rho_0$, $\s_p$, $\s_V$, $k_0$, $k_1$, $c_{k_{1}+2}$ (see Definition~\ref{def:Vadm}), $(\S, h_\refer)$ and the orthonormal frame $\{E_i\}_{i=1}^{n}$.

  Next, we wish to appeal to Theorem~\ref{thm:hod}. Note, therefore, that we have a solution to
  (\ref{eq: transport frame})--(\ref{eq:LapseEquation}) as described in Subsection~\ref{ssection:eqs}; the equations solved in \cite{OPR}
  are the ones written down in \cite[Proposition~70, pp.~23-24]{OPR}, and these equations coincide with
  (\ref{eq: transport frame})--(\ref{eq:LapseEquation}). Next, note that (\ref{eq:estdynvarmfltot gen}) with $\varepsilon=\sigma$ is an immediate
  consequence of \cite[(110), p.~37]{OPR} and \cite[Lemma~100, p.~39]{OPR}; recall \cite[Definition~87, p.~36]{OPR}. The estimate
  (\ref{eq:almost asymptotic Hcon gen}) is an immediate consequence of \cite[(147), p.~49]{OPR}. Note also that $3\sigma\leq \sigma_V$ due to
  (\ref{eq:sigmadef}). This means that the conclusions of Theorem~\ref{thm:hod} hold.
  
  Next, we wish to combine the above conclusions with Theorem~\ref{thm:asymptotics}. From \cite[Proposition~70, pp.~23-24]{OPR}, it follows
  that $g$ takes the form (\ref{eq:standard form g ito h}) and that (\ref{fermi walker equation}) holds. Next, let $\mrPhi$ and $\mrPsi$ be as
  in the statement of Theorem~\ref{thm:hod} and let $\mrp_I$ be the eigenvalues of $\mrmfK$ (the matrix with components $\mrmfK_{IJ}$). Then the
  $\mrp_I$ are distinct due to \cite[Theorem~12, pp.~6-7]{OPR} and smooth since the $\mrmfK_{IJ}$ are smooth. Due to \cite[(14), p.~7]{OPR}, it is
  clear that (\ref{eq:Kasner relations}) holds. Next, due to \cite[(226), p.~70]{OPR}, we can, by demanding that $\zeta_1$ be large enough,
  ensure that $p_I+p_J-p_K\leq 1-\sigma_p/2$ for all $t\leq t_0$ and $I$, $J$, $K$ such that $I\neq J$. This means that (\ref{subcritical condition})
  holds for any $\delta\leq\sigma$. Next, due to (\ref{seq:k phi prel est}), the estimate (\ref{eq:pI lim Phi lim Psi lim}) holds for all $k_0$,
  assuming $\delta<\sigma/4$. Assuming $\delta\leq \sigma/3$, the first three estimates in (\ref{bounds on the solution}) follow immediately
  from (\ref{seq:hod}). The final estimate follows from (\ref{seq:hod}) by demanding that $T$ be small enough. By demanding that $T$ be small enough,
  we can also assume that $|p_I-p_J|>r$ for $I\neq J$ and some $r>0$. Fixing such $r$, $T$ and $\delta<\sigma/4$, it follows that $(M,g,\varphi)$
  is a solution to the Einstein--nonlinear scalar field equations satisfying the $(\delta,n,k_0,r)$-assumptions for any $k_0\in\nn{}$.

  At this point we can appeal to Theorem~\ref{thm:asymptotics}. This yields the conclusion that there are smooth data on the singularity, as
  desired. Moreover, (\ref{eq:Cl asymptotics as data}) holds. The remaining conclusions of the theorem follow from \cite[Theorem~12, pp.~6-7]{OPR}.
  This finishes the proof of the theorem, with one caveat: we have demonstrated the desired conclusions for the finite covering space. On the
  other hand, the covering maps are local isometries, so that the conclusions descend to the quotient. The theorem follows. 
\end{proof}

\appendix

\section{Notation, conventions, estimates and identities}

\subsection{Conventions concerning number systems}
Throughout this article, we use the convention that $\nn{}=\{1,2,\dots\}$ and that $\nn{}_0=\{0,1,2,3,\dots\}$. 

\subsection{Symmetrisation and antisymmetrisation}\label{ssection:symm and antisymm}
Parentheses around indices in an expression denote symmetrisation and brackets denote antisymmetrisation. For example, in
(\ref{eq:concoef}), 
\[
- 2 N^{-1} e_{[I}(Nk_{J]K})=- N^{-1} e_{I}(Nk_{JK})+N^{-1} e_{J}(Nk_{IK}). 
\]
Moreover, in (\ref{eq:sff}), 
\[
\g_{K(IJ)}=\tfrac{1}{2}(\g_{KIJ}+\g_{KJI}).
\]

\subsection{Conventions concerning the geometry}\label{ssection:conv geo}
A \textit{spacetime} is a time oriented Lorentz manifold. Let $(M,g)$ be a spacetime with Levi-Civita connection $\nabla$ and let $\Sigma \subset M$ a
spacelike hypersurface with future pointing unit normal $e_0$ and induced metric $h$. Then the \textit{second fundamental form} of $\Sigma$ is given by $k(X,Y):=g(\n_X e_0,Y)$
for $X$ and $Y$ tangential to $\S$. If $s$ is a symmetric covariant $\ell$-tensor field on $\S$, $s^\sharp$ denotes the $(1,\ell-1)$--tensor field on $\S$
obtained by raising one index with $h$. In particular, the \textit{Weingarten map} is the $(1,1)$-tensor field defined by $K:=k^\sharp$. The mean
curvature is defined by $\theta:=\tr K$.

Let $(M,g)$ be an $n+1$--dimensional spacetime. Then $(M,g)$ is said to have a \textit{smooth foliation by spacelike hypersurfaces} if there is a smooth
$n$--dimensional manifold $\S$, an interval $\mI$ and a diffeomorphism $\Xi:\S\times \mI\rightarrow M$ such that the hypersurfaces
$\S_t:=\S\times \{t\}$ are spacelike with respect to $\Xi^*g$. The foliation is said to have \textit{lapse} $N>0$ and \textit{vanishing shift vector
field} if 
\begin{equation}\label{eq:Xi pullb g}
  \Xi^*g = -N^2\md t \otimes \md t + h,
\end{equation}
where $h$ denotes the family of induced (Riemannian) metrics on the hypersurfaces $\Sigma_t$. When we are in this situation, we usually omit reference
to the
diffeomorphism $\Xi$ and identify $M$ with $\S\times\mI$. In particular, we normally replace $\Xi^*g$ on the left hand side of (\ref{eq:Xi pullb g})
by $g$. Next, note that $e_0 := N^{-1}\d_t$ is the future pointing unit normal of the $\Sigma_t$.
If $\{E_i\}_{i=1}^n$ is a smooth global frame for $\Sigma$, it can be extended to $\Sigma \times \mI$ by requiring that $[\p_t,E_i] = 0$ for all $i = 1,\ldots,n$. Let $\{\eta^i\}_{i=1}^n$ be the corresponding dual frame. Then, the normal Lie derivative of $K$ is given in terms of this frame by  
\[
    \lie_{e_0}K=N^{-1}\d_t(K^i_{\phantom{i}j})E_i\otimes \eta^j;
\]
see, e.g., \cite[Subsection~A.2, pp.~225-227]{RinWave}. Moreover, the Levi-Civita connection of $h$ is
denoted by $\ovn$. We also use the notation
\[
\ldr{u,v}_h:=h_{p_1r_1}\cdots h_{p_kr_k}h^{q_1s_1}\cdots h^{q_\ell s_\ell}u^{p_1\cdots p_k}_{q_1\cdots q_\ell}v^{r_1\cdots r_k}_{s_1\cdots s_\ell}
\]
for $(k,\ell)$-tensor fields $u$ and $v$ on $\S_t$. Moreover, 
\[
|u|_h^2:=\ldr{u,u}_h.
\]
The notation $|u|_g^2$ is defined similarly, but now using the Lorentz metric $g$ and assuming $u$ to be a $(k,\ell)$-tensor field on $M$. Note,
however, that $|u|_g^2$ need not be non-negative in this case. 

\subsection{Norms}\label{ssection:norms and basic est}
Let $(\S,h_{\refer})$ be a closed Riemannian manifold with an orthonormal frame $\{E_i\}_{i=1}^n$. Denote the dual frame by $\{\eta^i\}_{i=1}^n$. 
In this article, bold capitalised Latin indices denote elements of $\{1,\dots,n\}^\ell$ for some $\ell\in\nn{}_0$. If
$\bfI\in \{1,\dots,n\}^\ell$, we let $|\bfI|:=\ell$ and if $\bfI=(i_1,\dots,i_\ell)$, then $E_{\bfI}:=E_{i_1}\cdots E_{i_\ell}$.
If $k\in\nn{}_0$ and $\psi\in C^\infty(\S)$, we define the $C^k$- and $H^k$-norms of $\psi$ on $\Sigma$ as follows
\begin{align*}
  \|\psi\|_{C^k(\S)} &:= \textstyle{\sum}_{|\bfI|\leq k}\sup_{x\in\S}|(E_{\bfI}\psi)(x)|,\\
  \|\psi\|_{H^k(\S)} &:=  \big(\textstyle{\sum}_{|\bfI|\leq k}\int_{\S}|E_{\bfI}\psi|^2\md\mu_{h_{\refer}}\big)^{1/2}.
\end{align*}
If $u$ is a tensor field on $\Sigma$, the $C^k$-norm of $u$ on $\Sigma$ is defined as the sum of the $C^k$-norms of the components of $u$, with respect to the frame $\{E_i\}_{i=1}^n$. Just as in \cite[(71), p.~28]{OPR}, we use the notation
\begin{subequations}\label{seq:e g Ck Hk norms}
  \begin{align}
    \|e\|_{C^k(\S)} &:= \textstyle{\sum}_{I,i}\|e_I^i\|_{C^k(\S)},\\
    \|\g\|_{H^k(\S)} &:= \big(\textstyle{\sum}_{I,J,K}\|\gamma_{IJK}\|_{H^k(\S)}^2\big)^{1/2}
  \end{align}
\end{subequations}
and similarly for $\omega^I_i$ and $k_{IJ}$. We also use the notation
\[
  \|\ear\varphi\|_{C^k(\S)} := \textstyle{\sum}_{I}\|e_I\varphi\|_{C^k(\S)},\ \ \
  \|\ear\varphi\|_{H^k(\S)} := \big(\textstyle{\sum}_{I}\|e_{I}\varphi\|_{H^k(\S)}^2\big)^{1/2}
\]
and similarly when $\varphi$ is replaced by $N$. 

\subsection{Basic estimates and identities}\label{ssection:bas es and identities}

Next, we recall the Gagliardo-Nirenberg estimates, which will be of central importance in what follows.
\begin{lemma}[Gagliardo-Nirenberg estimates] \label{lemma: products in L2}
  Let $(\S,h_{\refer})$ be a closed Riemannian manifold and $\{E_{i}\}_{i=1}^{n}$ be a global orthonormal frame with respect to $h_{\refer}$. Let
  $\ell_{i}\in\nn{}_0$, $i=1,\dots,j$, and $\ell=\ell_{1}+\dots+\ell_{j}$. If $|\bfI_{i}|=\ell_{i}$ and $\psi_{i}\in C^{\infty}(\S)$, $i=1,\dots,j$, then
\begin{equation}\label{eq:moser}
    \|E_{\bfI_{1}}\psi_{1}\cdots E_{\bfI_{j}}\psi_{j}\|_{L^2(\S)}
    \leq C\textstyle{\sum}_{i=1}^{j}\|\psi_{i}\|_{H^{\ell}(\S)}
        \prod_{m\neq i}\|\psi_{m}\|_{C^{0}(\S)},
  \end{equation}
where $C$ only depends on $\ell$, $(\S,h_{\refer})$ and $\{E_i\}_{i=1}^n$. In particular,
\begin{equation}\label{eq:moser2}
    \|\psi_1 \cdots \psi_j \|_{H^\ell(\S)}
    \leq C \textstyle{\sum}_{i=1}^{j}\|\psi_{i}\|_{H^{\ell}(\S)}
        \prod_{m\neq i}\|\psi_{m}\|_{C^{0}(\S)}.
\end{equation}
\end{lemma}
\begin{proof}
  The first statement is a special case of \cite[Corollary~B.8]{RinWave}. The second statement is an immediate consequence of the first. 
\end{proof}
We also need the following interpolation estimate.
\begin{lemma}[Lemma~135, p.~74, \cite{OPR}] \label{le: interpolation}
  With $(\S,h_{\refer})$ and $\{E_{i}\}_{i=1}^{n}$ as in Subsection~\ref{ssection:eqs}, let
  $k$, $\ell\in\nn{}_0$ be such that $0\leq \ell \leq k$ and $k\geq 1$. Then there is a constant $C$, depending
  only on $k$, $\ell$, $(\S,h_{\refer})$ and $(E_i)_{i=1}^n$, such that for all $\psi \in C^{\infty}(\S)$,  
\begin{equation}\label{eq:psiinterpol}
    \|\psi\|_{H^{\ell}(\S)}\leq C\|\psi\|_{L^{2}(\S)}^{1-\ell/k}\|\psi\|_{H^{k}(\S)}^{\ell/k}.
\end{equation}
\end{lemma}
Finally, in the derivation of the energy estimates, we make use of the following elementary observation. If $\psi\in C^{\infty}(\S)$ and
$X$ is a smooth vector field on $\S$, then 
\begin{equation}\label{eq:divergence formula}
  \textstyle{\int}_{\S}X(\psi)\md\mu_{h_{\refer}}=-\textstyle{\int}_{\S}\psi(\rodiv_{h_{\refer}}X)\md\mu_{h_{\refer}};
\end{equation}
see, e.g., \cite[(B.7), p.~233]{RinWave}.



\begin{thebibliography}{1}
\bibitem{ABIF} Ames, E.; Beyer, F.; Isenberg, J.; LeFloch, P. G.: Quasilinear hyperbolic Fuchsian systems and AVTD
behavior in $T^2$-symmetric vacuum spacetimes. Ann. Henri Poincar\'{e} {\bf 14}, no. 6, 1445--1523 (2013)
\bibitem{aeta} Ames, E.; Beyer, F.; Isenberg, J.; LeFloch, P. G.: A class of solutions to the Einstein equations with AVTD behavior
in generalized wave gauges. J. Geom. Phys. {\bf 121}, 42--71 (2017)
\bibitem{anetal} An, X.; He, T.; Shen, D.: Stability of Big Bang singularity for the Einstein-Maxwell-scalar field-Vlasov system in
  the full strong sub-critical regime. Preprint, \href{https://arxiv.org/abs/2507.18585}{arXiv:2507.18585}
\bibitem{aarendall} Andersson, L.; Rendall, A. D.: Quiescent cosmological singularities. Comm. Math. Phys. {\bf 218}, 479--511 (2001)
\bibitem{aaf} Athanasiou, N.; Fournodavlos, G.: A localized construction of Kasner-like singularities.
  Comm. Math. Phys. {\bf 406}, no. 10, Paper No. 252, 51 pp. (2025)
\bibitem{beguin}  B\'{e}guin, F.: Aperiodic oscillatory asymptotic behavior for some Bianchi spacetimes. Class. Quantum Grav. {\bf 27},
  185005 (2010)
\bibitem{bad} B\'{e}guin, F.; Dutilleul, T.: Chaotic dynamics of spatially homogeneous spacetimes.
Comm. Math. Phys. {\bf 399}, no.2, 737--927 (2023)
\bibitem{bat} Beyer, F.; Oliynyk, T. A.: Localized big bang stability for the Einstein-scalar field equations.
  Arch. Ration. Mech. Anal. {\bf 248}, no. 1, Paper No. 3 (2024)
\bibitem{baotwo} Beyer, F.; Oliynyk, T. A.: Past stability of FLRW solutions to the Einstein-Euler-scalar field equations and their big bang
  singularites. Beijing J. of Pure and Appl. Math. {\bf 1}, no. 2, 515--637, (2024)
\bibitem{boz} Beyer, F.; Oliynyk, T. A.; Zheng, W.: Localized past stability of the subcritical Kasner-scalar field spacetimes.
   Preprint, \href{https://arxiv.org/abs/2502.09210}{arXiv:2502.09210}
\bibitem{bkl70} Belinski\v{\i}, V. A.; Khalatnikov, I. M.; Lifshitz, E. M.: Oscillatory approach to a singular point in the relativistic
  cosmology. Adv. Phys. {\bf 19}, 525--573 (1970)
\bibitem{bkl82} Belinski\v{\i}, V. A.; Khalatnikov, I. M.; Lifshitz, E. M.: A general solution of the Einstein equations with a time
  singularity. Adv. Phys. {\bf 31}, 639--667 (1982)
\bibitem{brehm} Brehm, B.: Bianchi VIII and IX vacuum cosmologies: Almost every solution forms particle horizons and converges to the
  Mixmaster attractor. Preprint, \href{https://arxiv.org/abs/1606.08058}{arXiv:1606.08058}
\bibitem{cim} Chru\'{s}ciel, P. T.; Isenberg, J.; Moncrief, V.: Strong cosmic censorship in polarised Gowdy spacetimes
Classical Quantum Gravity {\bf 7}, no. 10, 1671--1680 (1990)
\bibitem{damouretal} Damour, T.; Henneaux, M.; Rendall, A. D.; Weaver, M.: Kasner-like behaviour for subcritical Einstein-matter systems. 
Ann. Henri Poincar\'{e} {\bf 3}, 1049--1111 (2002)
\bibitem{dhn} Damour, T., Henneaux, M., Nicolai, H.: Cosmological billiards. Class. Quantum Grav. {\bf 20}, R145 (2003)
\bibitem{dn}  Damour, T., Nicolai, H.: Higher order M theory corrections and the Kac-Moody algebra E10. Class. Quantum
Grav. {\bf 22} 2849 (2005)
\bibitem{dhs} Demaret, J.; Henneaux, M.; Spindel, P.: Nonoscillatory Behavior In Vacuum Kaluza-Klein Cosmologies. Phys. Lett. {\bf 164B}, 27--30 (1985)
\bibitem{fau} Fajman, D.; Urban, L.: Cosmic Censorship near FLRW spacetimes with negative spatial curvature.
  Anal. PDE {\bf 18}, no. 7, 1615--1713 (2025)
\bibitem{fauVl} Fajman, D.; Urban, L.: On the past maximal development of near-FLRW data for the Einstein scalar-field Vlasov system.
  Preprint, \href{https://arxiv.org/abs/2402.08544}{arXiv:2402.08544}
\bibitem{fal} Fournodavlos, G.; Luk, J.: Asymptotically Kasner-like singularities. Amer. J. Math. {\bf 145}, no.4, 1183--1272 (2023)
\bibitem{FRS} Fournodavlos, G.; Rodnianski, I.; Speck, J.: Stable big bang formation for Einstein's equations: the complete sub-critical regime.
  J. Amer. Math. Soc. {\bf 36}, no.3, 827--916 (2023)
\bibitem{andres} Franco-Grisales, A.: Developments of initial data on big bang singularities for the Einstein--nonlinear scalar field equations.
  Preprint, \href{https://arxiv.org/abs/2409.17065}{arXiv:2409.17065}
\bibitem{cbag} Choquet-Bruhat, Y.; Geroch, R.: Global aspects of the Cauchy problem in general relativity. Comm. Math. Phys. {\bf 14} 329--335 (1969)
\bibitem{hur} Heinzle, J. M.; Uggla, C., Rohr, N.: The cosmological billiard attractor. Adv. Theor. Math. Phys. {\bf 13},
  293–407 (2009)
\bibitem{hul}  Heinzle, J. M.; Uggla, C.; Lim, W. C.: Spike oscillations. Phys. Rev. D {\bf 86}, 104049 (2012)  
\bibitem{iak} Isenberg, J.; Kichenassamy, S.: Asymptotic behavior in polarized $T^{2}$-symmetric vacuum space-times. J. Math. Phys. {\bf 40},
  no. 1, 340--352 (1999)
\bibitem{iam} Isenberg, J.; Moncrief, V.: Asymptotic behaviour in polarized and half-polarized $U(1)$ symmetric vacuum spacetimes.
  Class. Quantum Grav. {\bf 19}, no. 21, 5361--5386 (2002)
\bibitem{kar} Kichenassamy, S.; Rendall, A.: Analytic description of singularities in Gowdy spacetimes. Class. Quantum Grav. {\bf 15}, no. 5,
  1339--1355 (1998)  
\bibitem{klinger} Klinger, P.: A new class of asymptotically non-chaotic vacuum singularities. Ann. Physics {\bf 363}, 1--35 (2015)
\bibitem{lhwg} Liebscher, S.; Härterich, J.; Webster, K.; Georgi, M.: Ancient dynamics in Bianchi models: approach to
  periodic cycles. Comm. Math. Phys. {\bf 305}, no. 1, 59–83 (2011)
\bibitem{lrt} Liebscher, S.; Rendall, A. D.; Tchapnda, S. B.: Oscillatory singularities in Bianchi models with magnetic fields.
Ann. Henri Poincar\'{e} {\bf 14}, no. 5, 1043--1075  (2013) 
\bibitem{OPR} Oude Groeniger, H.; Petersen, O.; Ringstr\"{o}m, H.: Formation of quiescent big bang singularities. Preprint,
  \href{https://arxiv.org/abs/2309.11370}{arXiv:2309.11370}
\bibitem{ren} Rendall, A. D.: Fuchsian analysis of singularities in Gowdy spacetimes beyond analyticity. Class. Quantum Grav. {\bf 17},
  no. 16, 3305--3316 (2000)
\bibitem{raw} Rendall, A. D.; Weaver, M.: Manufacture of Gowdy spacetimes with spikes. Class. Quantum Grav. {\bf 18},
  no. 16, 2959--2975 (2001)
\bibitem{cbu} Ringstr\"{o}m,  H.: Curvature blow up in Bianchi VIII and IX  vacuum spacetimes. Class. Quantum Grav. \textbf{17}, 713--731 (2000)
\bibitem{BianchiIXattr} Ringstr\"{o}m, H.: The Bianchi IX attractor, Ann. Henri Poincar\'{e} {\bf 2}, 405--500 (2001)
\bibitem{RinAsVel} Ringstr\"{o}m, H.:  Existence of an asymptotic velocity and implications for the asymptotic behavior in the direction of
  the singularity in $T^3$-Gowdy. Comm. Pure Appl. Math. {\bf 59}, no. 7, 977--1041 (2006)
\bibitem{RinCauchy} Ringstr\"{o}m, H.: The Cauchy problem in General Relativity, 
  European Mathematical Society, Z\"{u}rich (2009)
\bibitem{RinSCC} Ringstr\"{o}m, H.: Strong cosmic censorship in $T^3$-Gowdy spacetimes. Ann. Math. (2) {\bf 170}, no. 3, 1181--1240 (2009)
\bibitem{stab} Ringstr\"{o}m, H.: On the Topology and Future Stability of 
the Universe, Oxford University Press, Oxford (2013)
\bibitem{RinWave} Ringstr\"{o}m, H.: Wave equations on silent big bang backgrounds. Memoirs of the AMS {\bf 314}, no. 1595, (2025)
\bibitem{RinGeo} Ringstr\"{o}m, H.: On the geometry of silent and anisotropic big bang singularities (extended arXiv version). Preprint,
  \href{https://arxiv.org/abs/2101.04955}{arXiv:2101.04955}
\bibitem{RinGeoJDG} Ringstr\"{o}m, H.: On the geometry of silent and anisotropic big bang singularities.
  J. Differential Geom. {\bf 132}, no. 3, 461--545 (2026)
\bibitem{RinQC} Ringstr\"{o}m, H.: Initial data on big bang singularities. J. Eur. Math. Soc. (JEMS), Online first,
  \href{https://ems.press/journals/jems/articles/14298869}{DOI 10.4171/JEMS/1645}
\bibitem{RinSHID} Ringstr\"{o}m, H.: Initial data on big bang singularities in symmetric settings. Pure Appl. Math. Q. {\bf 20}, no. 4,
  1505--1539 (2024)
\bibitem{RinModel} Ringstr\"{o}m, H.: On the structure of big bang singularities in spatially homogenous solutions to the Einstein non-linear
  scalar field equations. Preprint, \href{https://arxiv.org/abs/2505.00429}{arXiv:2505.00429}
\bibitem{RinLocEx} Ringstr\"{o}m, H.: Local existence theory for a class of CMC gauges for the Einstein-non-linear scalar field equations. Preprint,
  \href{https://arxiv.org/abs/2509.14110}{arXiv:2509.14110}  
\bibitem{rasql} Rodnianski, I.; Speck, J.: A regime of linear stability for the Einstein-scalar field system with applications to
  non-linear big bang formation. Ann. Math. (2) {\bf 187}, no. 1, 65--156 (2018)
\bibitem{rasq} Rodnianski, I.; Speck, J.: Stable big bang formation in near-FLRW solutions to the Einstein-scalar field and 
Einstein-stiff fluid systems. Selecta Math. (N.S.) {\bf 24}, no. 5, 4293--4459 (2018)
\bibitem{rsh} Rodnianski, I.; Speck, J.: On the nature of Hawking's incompleteness for the Einstein-vacuum equations: the regime of 
  moderately spatially anisotropic initial data. J. Eur. Math. Soc. (JEMS) {\bf 24}, no. 1, 167--263 (2022)
\bibitem{specks3} Speck, J.:  The maximal development of near-FLRW data for the Einstein-scalar field system with spatial topology $\sn{3}$.
  Comm. Math. Phys. {\bf 364}, no. 3, 879--979 (2018)
\bibitem{sta} St\aa hl, F.: Fuchsian analysis of $S^{2}\times S^{1}$ and $S^{3}$ Gowdy spacetimes. Class. Quantum Grav. {\bf 19}, no. 17,
  4483--4504 (2002)  
\bibitem{weaver} Weaver, M.: Dynamics of magnetic Bianchi $\mathrm{VI}_0$ -cosmologies. Class. Quantum Grav. {\bf 17}, no. 2, 421--434 (2000)
\bibitem{zheng} Zheng, W.: Localized Big Bang Stability of Spacetime Dimensions $n\geq 4$. Preprint,
  \href{https://arxiv.org/abs/2601.21677}{arXiv:2601.21677}
\end{thebibliography}
\end{document}